\documentclass[11pt]{amsart}

\usepackage{amssymb,latexsym,amsmath,extarrows, mathrsfs, amsthm,color,amsrefs}
\usepackage{graphicx}

\usepackage[margin=1in, centering]{geometry}
\usepackage[bookmarksnumbered, colorlinks, plainpages]{hyperref}
\usepackage{hyperref} 
\hypersetup{
    colorlinks=true,       
    linkcolor=blue,          
    citecolor=magenta,        
    filecolor=magenta,      
    urlcolor=cyan           
}

\usepackage{mathtools}
\mathtoolsset{showonlyrefs,showmanualtags}

\newtheorem{theorem}{Theorem}[section]

\newtheorem{assumption}[theorem]{Assumption}
\newtheorem{lemma}[theorem]{Lemma}

\newtheorem{remark}[theorem]{Remark}
\newtheorem{example}[theorem]{Example}

\newtheorem{definition}{Definition}[section]
\newtheorem{corollary}[theorem]{Corollary}

\newtheorem{fact}[theorem]{Fact}

\newtheorem{conjecture}[theorem]{Conjecture}

\newcommand{\Z}{\mathbb Z}
\newcommand{\R}{\mathbb R}
\newcommand{\C}{\mathbb C}
\newcommand{\T}{\mathbb T}
\newcommand{\N}{\mathbb N}
\newcommand{\Q}{\mathbb Q}
\newcommand{\hv}{\hat{v}}

\def\mes{\mathrm{mes}}
\definecolor{deepgreen}{cmyk}{1,0,1,0.5}

\title[Localization on strip and ARC]{Non-perturbative localization on the strip and Avila's almost reducibility conjecture}

\author[R.\ Han]{Rui Han}
\address{Department of Mathematics \\ Louisiana State University  \\  Baton Rouge, LA 70803, USA}
\email{rhan@lsu.edu}

\author[W.\ Schlag]{Wilhelm Schlag}
\address{Department of Mathematics \\ Yale University \\ New Haven, CT 06511, USA}
\email{wilhelm.schlag@yale.edu}

\thanks{
R.\ Han is partially supported by NSF DMS-2143369. 
W.\ Schlag is partially supported by NSF grant DMS-1902691.
}

\begin{document}
\begin{abstract}
We prove non-perturbative Anderson localization and almost localization for a family of quasi-periodic operators on the strip. As an application we establish Avila's almost reducibility conjecture for Schr\"odinger operators with trigonometric potentials and all Diophantine frequencies, whose proof for analytic potentials was announced in 2015 in \cite{Global}.
As part of our analysis, we derive a non-selfadjoint version of Haro and Puig's formula connecting Lyapunov exponents of the dual model to those of the original operator.
\end{abstract}

\thanks{}

\maketitle

\section{Introduction}
In the first part of this paper, we prove non-perturbative Anderson localization for a family of one-dimensional finite-range quasi-periodic Schr\"odinger operators, which can also be viewed as an operator on the strip $\Z\times \{1,...,d\}$ or block-valued quasi-periodic Jacobi matrices.
Here non-perturbative refers to the positive Lyapunov exponent (half of exponents of the corresponding symplectic cocycles are positive) regime.
We also prove almost localization all $\theta$ and any Diophantine $\alpha$ (those satisfy $\beta(\alpha)=0$, see \eqref{def:beta}).
This together with the dual (almost) reducibility techniques developed by Avila-Jitomirskaya in \cite{AJ2}, proves Avila's almost reducibility conjecture (ARC) in \cite{Global} for Schr\"odinger operators with trigonometric potentials and all Diophantine $\alpha$'s. 
Combined with Avila's proof of ARC for Liouville $\alpha$ (those satisfying $\beta(\alpha)>0$), this gives a complete proof of ARC for trigonometric potentials.

Let 
\begin{align}
(H_{\alpha,\theta,\varepsilon}u)_n=u_{n+1}+u_{n-1}+v(\theta+n\alpha+i\varepsilon)u_n,
\end{align}
$v:\T=\R/\Z\to \R$ analytic, be a quasi-periodic Schr\"odinger operator on $\Z$. 
Let
\begin{align}\label{def:AE}
A_{E,\varepsilon}(\theta)=
\left(\begin{matrix} 
E-v(\theta+i\varepsilon) &-1\\
1 &0
\end{matrix}\right)
\end{align}
be the transfer matrix, and $(\alpha, A_{E,\varepsilon})$ be the associated cocycle. Let $L(\alpha, A_{E,\varepsilon})$ be the Lyapunov exponent. 
In his seminal work \cite{Global}, Avila introduced the {\it acceleration}, defined as:
\begin{align}\label{def:acc}
\kappa(\alpha, A_{E,\varepsilon})=\lim_{\delta\to 0^+} \frac{L(\alpha, A_{E,\varepsilon+\delta})-L(\alpha, A_{E,\varepsilon})}{2\pi \delta}.
\end{align}
Furthermore, he proved the acceleration is quantized, meaning $\kappa(\alpha, A_{E,\varepsilon})\in \Z$.
Using the quantized acceleration, Avila divides the spectrum $\sigma(H_{\alpha,\theta,0})=\sigma(H_{\alpha})$ (we restrict to the irrational $\alpha$ case, for which the spectrum is independent of $\theta$) into three regimes: 
the supercritical regime $$\Sigma_{\alpha,+}:=\{E\in \sigma(H_{\alpha}):\, L(\alpha, A_{E,0})>0, \text{ and } \kappa(\alpha, A_{E,0})>0\},$$ the subcritical regime $$\Sigma_{\alpha,0}:=\{E\in \sigma(H_{\alpha}):\, L(\alpha, A_{E,0})=0, \text{ and } \kappa(\alpha, A_{E,0})=0\},$$ and the critical regime $$\Sigma_{\alpha,c}:=\{E\in \sigma(H_{\alpha}):\, L(\alpha, A_{E,0})=0, \text{ and } \kappa(\alpha, A_{E,0})>0\}.$$

In the supercritical regime, non-perturbative proofs of Anderson localization have been developed for the Harper operator with all $\alpha\in \mathrm{DC}$\footnote{$\mathrm{DC}=\bigcup_{c>0}\bigcup_{b>1}\{\alpha\in \T:\, \|n\alpha\|_{\T}\geq c\, |n|^{-b}, n\neq 0\}\subsetneq \{\alpha\in \T:\, \beta(\alpha)=0\}$} in \cite{J99} (where $v(\theta)=2\lambda\cos(2\pi\theta)$ with $|\lambda|>1$), for general analytic potentials with a.e.\ $\alpha$ in \cite{BG}, and for all $\alpha\in\mathrm{DC}$ in the first supercritical stratum (where $\kappa(\alpha, A_{E,0})=1$) in \cite{HS2}.
In the critical regime, the spectrum is conjectured to be singular continuous (see Conjecture 3.1 in \cite{AJM}). This conjecture has been proved for the critical Harper operator (where $v(\theta)=2\cos(2\pi\theta)$) in \cites{A_c,AJM,J_c}.

In the subcritical regime, Avila formulated the almost reducibility conjecture (ARC), which we shall elaborate on below.
A cocycle $(\alpha,A)$ is called (analytically) {\it reducible} if there exists $D\in C^{\omega}(\T, \mathrm{PSL}(2,\R))$ such that the analytic conjugacy $D^{-1}(\theta+\alpha) A(\theta) D(\theta)$ is a  constant.
A cocycle is called {\it almost reducible} if the closure of its analytic conjugacies contains a constant. 
\begin{conjecture}[ARC by Avila, \cite{Global}]
Subcritical cocycles are almost reducible.
\end{conjecture}
As remarked by Avila in \cite{Global}, ``{\it Proving this almost reducibility conjecture would at once provide an almost complete understanding of subcriticality.}''
This conjecture has thus far been proved in two different cases:  
\begin{itemize}
    \item for Diophantine $\alpha$'s: (i)  for general analytic potentials in the small coupling regime  (ii) for the entire subcritical regime of the almost Mathieu (or Harper) operator; both  in \cite{AJ2}. See also \cite{GYZ}.
    \item  for Liouville $\alpha$'s (those $\beta(\alpha)>0$) in \cite{A_ar_ac}. 
\end{itemize}
In this paper, we prove the following non-perturbative quantitative ARC for Schr\"odinger cocycles with trigonometric potentials and all Diophantine $\alpha$. In the rest of the paper, we assume that $v(x)=\sum_{1\leq |k|\leq d} \hat{v}_k e^{2\pi i kx}$ is a real-valued trigonometric polynomial, which implies for $1\leq k\leq d$, $\hat{v}_k=\overline{\hat{v}_{-k}}$. Without loss of generality we assume that $\hat{v}_0=0$.
\begin{theorem}\label{thm:ARC}
Let $\beta(\alpha)=0$, and $E\in \Sigma_{\alpha,0}$. 
Let $0<\varepsilon_0<\frac{\hat{L}_d(\alpha, \hat{A}_E)}{30\hat{L}^d(\hat{A}_E)+100 C_3d}$ where $C_3$ is the absolute constant in Lemma \ref{lem:uniform}, and $\delta=\min(\frac{\hat{L}_d(\alpha, \hat{A}_E)}{5000}, \frac{\varepsilon_0}{25})$. Let $\theta=\theta(E)$ be given by Theorem \ref{thm:bdd_solution}. Let $\{n_j\}$ be the sequence of $\varepsilon_0$-resonances of $\theta$ and $L_j:=\|2\theta-n_j\alpha\|_{\T}$. Then there exists an absolute constant $C>0$ such that for $j>j_0(\alpha, v, E, \varepsilon_0)$ there exists $W\in C^{\omega}_{\frac{\delta}{2\pi}}(\T, \mathrm{SL}(2,\R))$ such that $|\mathrm{deg}(W)|\leq C|n_j|$, $\|W\|_{\frac{\delta}{2\pi}}\leq L_j^{-C}$ and
\begin{align}
\|W(x+\alpha)^{-1} A_E(x) W(x)-R_{-\theta}\|_{\frac{\delta}{2\pi}}\leq e^{-\frac{\hat{L}_d(\alpha, \hat{A}_E)}{300}|n_{j+1}|}.
\end{align}
\end{theorem}
Above, $\hat{L}_d(\alpha, \hat{A}_E)>0$ is the $d$-th Lyapunov exponent of the dual operator, see Corollary \ref{cor:non-sa-HP}.
The core of the proof is to establish the almost localization for the dual model (Theorem \ref{thm:almost_AL} below). Once established, one can combine it with the quantitative duality arguments in \cite{AJ2} to prove Theorem \ref{thm:ARC}. We include the duality argument in Appendix \ref{app:A} in order to keep the paper self-contained. 

One example to which Theorem \ref{thm:ARC} applies is the perturbation of the subcritical almost Mathieu operator, which is also subcritical due to the upper semi-continuity of the acceleration.
\begin{example}
Let $H^{\mathrm{PAMO}}_{\alpha,\theta,\nu, w}$ be the operator with $v(\theta)=2\lambda\cos(2\pi\theta)+\nu\cdot w(\theta)$ and $\varepsilon=0$, where $w(\theta)$ is a trigonometric polynomial. For $|\lambda|<1$ and sufficiently small $|\nu|<\nu_0(\alpha,\lambda,w)$, 
$\sigma(H^{\mathrm{PAMO}}_{\alpha,\theta,\nu, w})=\Sigma_{\alpha,0}$.
\end{example}

\begin{theorem}\label{thm:PAMO_ARC}
For $|\lambda|<1$, $\beta(\alpha)=0$, and sufficiently small $|\nu|<\nu_0(\alpha,\lambda,w)$, for any $E\in\sigma(H^{\mathrm{PAMO}}_{\alpha,\theta,\nu,w})$, the cocycle $(\alpha, A_E^{\mathrm{PAMO}})$ associated to $H^{\mathrm{PAMO}}_{\alpha,\theta,\nu,w}$ is quantitatively almost reducible.
\end{theorem}
\begin{remark}
It is well-known that for $E\notin \sigma(H^{\mathrm{PAMO}}_{\alpha,\theta,\nu,w})$, $(\alpha, A_E^{\mathrm{PAMO}})$ is $\mathrm{PSL}(2,\R)$ reducible, see e.g. Theorem 2.1 of \cite{AJ2}.
\end{remark}

As a consequence of Theorem \ref{thm:PAMO_ARC}, we can prove the dry Ten Martini problem for $H^{\mathrm{PAMO}}_{\alpha,\theta,\nu,w}$.
\begin{theorem}\label{thm:dryTen}
For $|\lambda|<1$, $\beta(\alpha)=0$, and for sufficiently small $|\nu|<\nu_0(\alpha,\lambda,w)$, all the spectral gaps of $H^{\mathrm{PAMO}}_{\alpha,\theta,\nu w}$ as predicted by the gap labelling theorem are open.
\end{theorem}
This is an analogue of Theorem 1.1 of \cite{AJ2}, we include its proof in the appendix for completeness. Next, we introduce the dual model of $H_{\alpha,\theta,\varepsilon}$:
\begin{align}
(\hat{H}_{\alpha,\theta,\varepsilon}\phi)_n=\sum_{0<|k|\leq d} \hat{v}_k \phi_{n-k} e^{-2\pi k\varepsilon}+2\cos(2\pi(\theta+n\alpha)) \phi_n.
\end{align}
This long-range operator can also be written in terms of block-valued Jacobi matrix, acting on $\ell^2(\Z, \C^d)$, as follows.
\begin{align}\label{def:H_dual}
(\hat{H}_{\alpha,\theta,\varepsilon}\Phi)_n=B_{\varepsilon} \Phi_{n+1}+\tilde{B}_{\varepsilon} \Phi_{n-1}+C_{\varepsilon}(\theta+nd\alpha) \Phi_n,
\end{align}
where $\Phi_n=(\phi_{nd+d-1},...,\phi_{nd+1},\phi_{nd})^T\in \C^d$, and $B_{\varepsilon}, \tilde{B}_{\varepsilon}$ and $C_{\varepsilon}$ are defined in Section \ref{sec:transfer_matrix}.
The transfer matrix formulation of the block-valued Hamiltonian eigenvalue equation $\hat{H}_{\alpha,\theta,\varepsilon}\Phi=E\Phi$ is
\begin{align}\label{eq:hH_transfer_block}
\left(\begin{matrix}
B_{\varepsilon}\Phi_{n+1}\\
\Phi_n
\end{matrix}\right)=
M_E^{\varepsilon}(\theta+nd\alpha)
\left(\begin{matrix}
B_{\varepsilon}\Phi_n\\
\Phi_{n-1}
\end{matrix}\right),
\end{align}
where
\begin{align}\label{def:ME^eps}
M_E^{\varepsilon}(\theta)=
\left(\begin{matrix}
(E-C_{\varepsilon}(\theta))B_{\varepsilon}^{-1} &-\tilde{B}_{\varepsilon}\\
B_{\varepsilon}^{-1} &0
\end{matrix}\right).
\end{align}
We denote its Lyapunov exponents by $\hat{L}_1(\alpha,M_E^{\varepsilon})\geq...\geq \hat{L}_{2d}(\alpha,M_E^{\varepsilon})$.
The transfer matrix $M_E^{\varepsilon}$ can also be viewed as $d$-step transfer matrix for the scalar-valued Hamiltonian, where
\begin{align}\label{eq:M=Ad}
M_E^{\varepsilon}(\theta)=\mathrm{diag}(B_{\varepsilon}, I_d) \cdot \hat{A}_{d,E}^{\varepsilon}(\alpha,\theta)\cdot  \mathrm{diag}(B_{\varepsilon}^{-1}, I_d),
\end{align}
in which the one step transfer matrix $\hat{A}_E^{\varepsilon}$ is as in~\eqref{def:hAE}.
It is easy to see by \eqref{eq:M=Ad} that for $1\leq j\leq 2d$,
\begin{align}
\hat{L}_j(d\alpha, M_E^{\varepsilon})=d\cdot \hat{L}_j (\alpha, \hat{A}_E^{\varepsilon}).
\end{align}
The transfer matrix $M_E^{0}$ at an arbitrary energy $E\in \R$ is complex symplectic, see \eqref{eq:symplectic}. Hence 
\begin{align}
\hat{L}_j(d\alpha, M_E^0)=-\hat{L}_{2d+1-j}(d\alpha, M_E^0), \text{ for } 1\leq j\leq 2d.
\end{align}
Proving Anderson localization in the non-perturbative $\hat{L}_d(\alpha,M_E^0)>0$ regime has been a long standing problem in quasi-periodic Schr\"odinger operators, while the Anderson model counterpart has been known for over 40 years, see \cites{G,L,KLS,MS}. 
Resolving this problem is one of the main accomplishments of this paper.
\begin{theorem}\label{thm:AL_cos}
Let $\beta(\alpha)=0$ and $\theta$ be such that 
\begin{align}\label{def:gamma_alpha_theta}
\gamma(\alpha,\theta):=\lim_{n\to\infty} \left(-\frac{\log \|2\theta+n\alpha\|_{\T}}{|n|}\right)=0,
\end{align}
Then $\hat{H}_{\alpha,\theta,0}$ has pure point spectrum with exponentially decaying eigenfunctions in $$\sigma(\hat{H}_{\alpha,\theta,0})\cap \{E:\, \hat{L}_d(d\alpha, M_E^0)>0\}.$$
\end{theorem}
We briefly discuss the difficulties. We write $\hat{L}_j(d\alpha, M_E^0)$ as $\hat{L}_j$.
The proof of Anderson localization in the literature relies on the exponential decay of the Green's function (with Dirichlet boundary condition). In fact, let $\hat{H}_{\alpha,n}(\theta)$ be the finite volume operator restricted to $[0,nd-1]$, see \eqref{def:H_Dirichlet}.
Let $\tilde{G}_{E,n}(\theta):=(\hat{H}_{\alpha,n}(\theta)-E)^{-1}$ be the resolvent. The key to prove localization is to prove exponential decay of  $|G_{E,n}(\theta+j\alpha;1,x)|$ for $x\sim n/2$. By Cramer's rule, one has
$$G_{E,n}(\theta+j\alpha;1,x)=\tilde{\mu}_{n,1,x}(\theta+j\alpha)/\tilde{f}_{E,n}(\theta+j\alpha),$$ where $\tilde{f}_{E,n}(\theta)=\det(\hat{H}_{\alpha,n}(\theta)-E)$.

In the scalar-valued case, $d=1$, the numerator $\hat{\mu}_{n,1,x}(\theta+j\alpha)$ is equal to $\tilde{f}_{E,n-x}(\theta+(x+j)\alpha)$, which is an element (out of four in total) of $M_{n-x,E}^0(\theta+(x+j)\alpha)$, and can thus be easily controlled by 
$$\|M_{n-x,E}^0(\theta+(x+j)\alpha)\|\leq \exp((n-x)(\hat{L}^1+\varepsilon)).$$
The key for the $d=1$ case is to find a lower bound for the  denominator $|\tilde{f}_{E,n}(\theta+j\alpha)|$ for a properly chosen $j$. 
It was proved by Jitomirskaya \cite{J99} that $\tilde{f}_{E,n}(\theta)$ is a polynomial in $\cos(2\pi(\theta+(n-1)\alpha/2))$ of degree $n$, hence if it is uniformly small at $(n+1)$ uniformly distributed $\theta+j\alpha$'s and at three consecutive scales $n_0-2, n_0-1, n_0$, then $$\sup_{n\in \{n_0-2,n_0-1,n_0\}} \sup_{\theta\in \T} |\tilde{f}_{E,n}(\theta)|\leq \exp(n(\hat{L}_1-\varepsilon)).$$
However in (and {\it only in}) the scalar-valued case, we have $$\sup_{\theta\in \T} \|M_{n_0,E}^0(\theta)\|\leq \max_{n\in \{n_0-2, n_0-1, n_0\}} \sup_{\theta\in \T} |\tilde{f}_{E,n}(\theta)|\leq \exp(n(\hat{L}_1-\varepsilon)),$$
which would lead to a contradiction.

In the block-valued case $d\geq 1$, there are significant difficulties in estimating both the numerator and the denominator.
First of all the numerator can no longer be identified with a determinant $\tilde{f}_{E,n}$.
In \cites{BJ,SK_block} and \cite{AJ2}\footnote{In \cite{AJ2}, $d=\infty$}, the numerators are estimated through a random walk type expansion, where the expansion is only efficient in the large coupling regime (replacing $2\cos$ with $2\lambda^{-1} \cos$ and $\lambda$ small). 
In the non-perturbative regime, a much tighter upper bound is need, in fact we prove \[|\mu_{n,1,x}(\theta)|\leq |\hat{v}_d|^{nd}\cdot \max\left(e^{d^{-1}x(\hat{L}^{d-1}+\varepsilon)+(n-d^{-1}x)(\hat{L}^d+\varepsilon)}, e^{d^{-1}x(\hat{L}^d+\varepsilon)+(n-d^{-1}x)(\hat{L}^{d+1}+\varepsilon)}\right)\] uniformly in $\theta$ and $\mu_{n,1,x}(\theta)$ is the numerator of the Green's function for the finite volume operator $P_{\alpha,n}(\theta)$ with periodic boundary condition (see \eqref{def:Pn}). 
We can prove a similar bound for $\tilde{\mu}_{n,1,x}$ as well. The reason why we choose the periodic instead Dirichlet boundary condition lies in the estimate of the denominator which we elaborate on below.

For the denominator, we only know that $|\tilde{f}_{E,n}(\theta)|/|\hat{v}_d|^{nd}$ is the determinant of the upper left $d\times d$ corner of $M_{n,E}^0(\theta)$, hence it is unclear why $\sup_{\theta}|\tilde{f}_{E,n}(\theta)|/|\hat{v}_d|^{nd}$ is close to $$\sup_{\theta}\|\wedge^d M_{n,E}^0(\theta)\|=\exp(n(\hat{L}^d+o(1))).$$ 
One might attempt to argue that 
\begin{align}\label{eq:wrong}
\frac{1}{nd}\log |\tilde{f}_{E,n}(\theta)|=&\int_{\R} \log |E-E'|\, \hat{\mu}_{n,\alpha,\theta}(\mathrm{d}E')\longrightarrow \\ &\qquad\qquad\qquad \int_{\R} \log |E-E'|\, \widehat{\mathcal{N}}(\mathrm{d} E')=\hat{L}^d(\alpha,\hat{A}_E^0)+\log |\hat{v}_d|,
\end{align}
where $\hat{\mu}_{n,\alpha,\theta}$ is the finite-scale density of states measure (see \eqref{def:finite_DOS}) and the second line is the Thouless formula \eqref{eq:Thouless} that was proved in the $d>1$ case in \cites{CS2,KS,HP}.
However the finite volume density of states measure $\hat{\mu}_{n,\alpha,\theta}$ only converges in the weak-$*$  sense (when integrated against continuous function of compact support) to the density of states measure $\widehat{\mathcal{N}}$, which does not apply when integrated against $f(\cdot)=\log|E-\cdot|$ for $E\in \R$. 
Therefore the approach \eqref{eq:wrong} does not lead to the desired result when $E\in \R$. See Remark \ref{rmk:Dirichlet_vs_periodic}, too.
We resolve this difficulty by studying the Green's function $G_{E,n}(\theta)=(P_{\alpha,n}(\theta)-E)^{-1}$ with a periodic boundary condition.
Let $f_{E,n}(\theta)=\det(P_{\alpha,n}(\theta)-E)$ be the denominator of the Green's function $G_{E,n}(\theta)$. We prove that
$$\log |f_{E,n}(\theta)|=nd\log |\hat{v}_d|+\log |\det (M^0_{n,E}(\theta)-I_{2d})|.$$
We then show that for an admissible sequence of $n$, on a large measure set of $\theta$, 
\[\log |\det (M^0_{n,E}(\theta)-I_{2d})|=n(\hat{L}^d+o(1)).\]
The choice of the admissible sequence relies crucially on Lemma 3.2 of \cite{GS3}, which guarantees that the most expanding $d$-planes of $M_{n,E}^0$, respectively of $(M_{n,E}^0)^*$,  are not entirely orthogonal to each other.
Our lower bound Lemma \ref{lem:deno} for the denominator is robust, and can be applied to establish Anderson localization for operators $\hat{H}_{\alpha,\theta}$ on strips with general analytic potentials defined on higher-dimensional tori. We will address this problem in an upcoming work. For the cosine potential which arises in this paper, there is a more elementary argument using Herman's subharmonic argument, see Lemma~\ref{lem:deno_ARC}.

A stronger version of {\it almost localization} holds for arbitrary $\theta$, see Theorem \ref{thm:almost_AL}.
The link between Theorem \ref{thm:almost_AL} and Theorem \ref{thm:ARC} is provided by Corollary \ref{cor:non-sa-HP} that if $E\in \Sigma_{\alpha,0}$, then the dual Lyapunov exponent satisfies $\hat{L}_d(d\alpha, M_E^0)>0$.
This was announced in Jitomirskaya's 2022 ICM talk, and is indeed a consequence of a non-selfadjoint variant of Haro and Puig's result on dual Lyapunov exponents, which we elaborate on below.

In \cite{HP}, Haro and Puig proved the following remarkable  theorem concerning Lyapunov exponents of dual models. The  techniques in~\cite{HP} are  inspired by Johnson's  parametrization of the stable and unstable subspaces in the uniformly hyperbolic setting via the Weyl-Kodaira $M$-matrix, see~\cites{Joh, JPS, JONNF}.

\begin{theorem}\cite{HP}*{Theorem 1.3 and Corollary 1.6}\label{thm:HP}
For any irrational $\alpha$, 
\begin{align}
L^1(\alpha,A_{E,0})=\hat{L}^d(\alpha,\hat{A}_E^0)+\log |\hat{v}_{-d}|
=\sum_{j: \hat{L}_j(\alpha, \hat{A}_E^0)\geq 0} \hat{L}_j(\alpha, \hat{A}_E^0)+\log |\hat{v}_{-d}|.
\end{align}
\end{theorem}
The proof is built on two key ingredients.
One is the Thouless formula, that links the Lyapunov exponent to the density of states measure, for both $H_{\alpha,\theta,\varepsilon=0}$ (the selfadjoint Schr\"odinger operator) and $\hat{H}_{\alpha,\theta,\varepsilon=0}$ (the dual model on the strip, which is also selfadjoint). 
The second is the fact that the two dual selfadjoint operators share the same density of states measure.

The following non-selfadjoint version of the Haro, Puig formula was announced in Jitomirskaya's ICM talk in 2022.
\begin{theorem}\label{announce}[Jitomirskaya, ICM talk 2022]
For any irrational $\alpha$,
\begin{align}
L^1(\alpha,A_{E,\varepsilon})=L^1(\alpha,A_{E,0})-\sum_{j: 0\leq \hat{L}_j(\alpha,\hat{A}_E^0)<2\pi\varepsilon} \hat{L}_j(\alpha,\hat{A}_E^0)+\#\{j: 0\leq \hat{L}_j(\alpha,\hat{A}_E^0)<2\pi\varepsilon\}\cdot (2\pi\varepsilon).
\end{align}
\end{theorem}
Noting that $\hat{L}_j(\alpha,\hat{A}_E^\varepsilon)=\hat{L}_j(\alpha,\hat{A}_E^0)-2\pi \varepsilon$, see \eqref{eq:LA_eps=LA_0}, and taking Haro and Puig's formula (which holds at $\varepsilon=0$) into account, we see that Theorem \ref{announce} in the same as a non-selfadjoint version of Haro and Puig's formula, see Theorem~\ref{thm:non-sa-HP} below.

Since the proof of Theorem \ref{announce} does not seem available in the literature, and because it appears in our resolution of the almost reducibility conjecture, we provide a proof in Section~\ref{sec:HP} for the reader's convenience.
\begin{theorem}\label{thm:non-sa-HP}
For any irrational $\alpha$,
\begin{align}\label{eq:non-sa_HP}
L^1(\alpha,A_{E,\varepsilon})=\sum_{j: \hat{L}_j(\alpha, \hat{A}_E^\varepsilon)\geq 0} \hat{L}_j(\alpha, \hat{A}_E^\varepsilon)+\log |\hat{v}_{-d} e^{2\pi d\varepsilon}|.
\end{align}
\end{theorem}
A major difficulty in extending Haro and Puig's formula to the non-selfadjoint case (i.e.,  when $\varepsilon\neq 0$) is the lack of spectral theorems. Hence there is no density of states measure, which played a crucial role in Haro and Puig's original proof.
We overcome this problem by establishing the Aubry duality, by means of the characters on  cyclic groups, for finite scale operators with periodic boundary conditions and rational frequencies. 
These finite-scale operators naturally have (finite scale) density of states measures via the normalized counting measures of the eigenvalues.
We then prove Thouless formulas for the finite scale rational-frequency operators, which after combining with the duality, lead to \eqref{eq:non-sa_HP} with rational $\alpha$.
Finally, \eqref{eq:non-sa_HP} for irrational $\alpha$ follows from rational approximation and continuity of the Lyapunov exponents at irrational frequencies.

It was also addressed by Jitomirskaya in her ICM talk that Theorem \ref{announce} leads to the following characterization of Avila's acceleration. 
\begin{corollary}\label{cor:Avila_acc}
For irrational $\alpha$, and any $\varepsilon\ge0$
\begin{align}
\kappa(\alpha,A_{E,\varepsilon})=\#\{j: 0\leq \hat{L}_j(\alpha, \hat{A}_E^\varepsilon)\le 2\pi \varepsilon\}.
\end{align}
\end{corollary}
We give a proof of this in Section \ref{sec:HP}.
Clearly, Corollary \ref{cor:Avila_acc} and \eqref{eq:LA_eps=LA_0} lead to the following.
\begin{corollary}\label{cor:non-sa-HP}
If the acceleration of $H_{\alpha,\theta,\varepsilon}$ satisfies $\kappa(\alpha,A_{E,\varepsilon})\equiv 0$ for $|\varepsilon|\leq \varepsilon_1$, then 
\begin{align}
\hat{L}_1(\alpha,\hat{A}_E^0)\geq ...\geq \hat{L}_d(\alpha,\hat{A}_E^0)\geq 2\pi \varepsilon_1>0.
\end{align}
\end{corollary}
In the rest of the paper, we shall omit the dependence on $\varepsilon$ in various notations in the $\varepsilon=0$ case.

The paper is organized as follows: some preliminary knowledge is recalled in Sec.~\ref{sec:Pre}, the proofs of Theorem~\ref{thm:AL_cos} and almost localization are given in Sec.~\ref{sec:Loc} assuming the numerator bound of the Green's function established in Sec.~\ref{sec:Numerator} and the denominator bound in Sec.~\ref{sec:Denominator}. Finally we present the proof of Theorem~\ref{thm:non-sa-HP} in Sec.~\ref{sec:HP}.

\medskip

{\bf Notations:} for $x\in \R$, let $\|x\|_{\T}:=\mathrm{dist}(x,\Z)$ be the torus norm, and $\lfloor x\rfloor=\sup\{m\in \Z:\, m\leq x\}$ be the integer part of $x$. If $a\leq Cb$ for some constant $C>0$, we shall write $a=O(b)$.
For a matrix $M\in \C^{d_1\times d_2}$, and the sets of indicies $B_1\subset [1,d_1]\cap \Z, B_2\subset [1, d_2]\cap \Z$, let $M_{B_1, B_2}$ be the submatrix of $M$ with row numbers in $B_1$ and column numbers in $B_2$.

\section{Preliminaries}\label{sec:Pre}
\subsection{Continued fraction expansion}
Give $\alpha\in (0,1)$, let $[a_1,a_2,...]$ be the continued fraction expansion of $\alpha$. For $n\geq 1$, let $p_n/q_n=[a_1,a_2,...,a_n]$ be the continued fraction approximants. 
The following property is well-known for $n\geq 1$,
\begin{align}\label{eq:qn_alpha_min}
\|q_n\alpha\|_{\T}=\min_{1\leq k<q_{n+1}} \|k\alpha\|_{\T},
\end{align}
The $\beta(\alpha)$ exponent measures the exponential closeness of $\alpha$ to rational numbers:
\begin{align}\label{def:beta}
\beta(\alpha):=\limsup_{n\to\infty}\frac{\log q_{n+1}}{q_n}=\limsup_{n\to\infty}\left(-\frac{\log \|n\alpha\|_{\T}}{n}\right).
\end{align}

\subsection{Lyapunov exponents}
Let $(\alpha, A)\in (\T, C^{\omega}(\T, \mathrm{Mat}(k,\C)))$.
Let 
\begin{align}
A_n(\alpha,\theta)=A(\theta+(n-1)\alpha)\cdots A(\theta).
\end{align}
Let the finite-scale and infinite-scale Lyapunov exponents be defined as 
\begin{align}
L_{j,(n)}(\alpha, A):=\frac{1}{n}\int_{\T} \log \sigma_j(A_n(\alpha,\theta))\, \mathrm{d}\theta, \text{ for } 1\leq j\leq k,
\end{align}
where $\sigma_j(A)$ is the $j$-th singular value of $A$, and the $j$-th Lyapunov exponent
\begin{align}
L_j(\alpha, A)=\lim_{n\to\infty}L_{j,(n)}(\alpha, A).
\end{align}
It is easy to see that for $1\leq j\leq k$,
\begin{align}
L^j_{(n)}(\alpha, A):=\sum_{\ell=1}^j L_{\ell,(n)}(\alpha, A)=\frac{1}{n}\int_{\T}\log \|\wedge^j A_n(\alpha,\theta)\|\, \mathrm{d}\theta,
\end{align}
where $\wedge^j A$ is the $j$-th exterior power of $A$. Similarly $L^j(\alpha,A)=\sum_{\ell=1}^j L_{\ell}(\alpha, A)$.
If $\alpha=p/q\in \Q$, we define
\begin{align}\label{def:Lpq}
L^j(p/q,A,\theta)=\lim_{n\to\infty} \frac{1}{n}\log \|\wedge^j A_n(p/q,\theta)\|=\frac{1}{q}\log (\rho(\wedge^j A_q(p/q,\theta))),
\end{align}
where $\rho(A)$ is the spectral radius of $A$. 
Let $L_1(p/q,A,\theta)=L^1(p/q,A,\theta)$ and
\begin{align}
L_j(p/q,A,\theta)=L^j(p/q,A,\theta)-L^{j-1}(p/q,A,\theta), \text{ for } j\geq 2,
\end{align}
and for $1\leq j\leq k$,
\begin{align}
L^j(p/q,A)=\int_{\T} L^j(p/q,A,\theta)\, \mathrm{d}\theta, \text{ and } L_j(p/q,A)=\int_{\T} L_j(p/q,A,\theta)\, \mathrm{d}\theta.
\end{align}

\subsection{Transfer matrix of \texorpdfstring{$\hat{H}_{\alpha,\theta,\varepsilon}$}{Lg}}\label{sec:transfer_matrix}

As we mentioned in the introduction,  $\hat{H}_{\alpha,\theta,\varepsilon}$ can be viewed as a block-valued Jacobi matrix:
\begin{align}
(\hat{H}_{\alpha,\theta,\varepsilon}\Phi)_n=B_{\varepsilon} \Phi_{n+1}+\tilde{B}_{\varepsilon}\Phi_{n-1}+C_{\varepsilon}(\theta+nd\alpha) \Phi_n,
\end{align}
where $\Phi_{n}=(\phi_{nd+d-1},...,\phi_{nd+1},\phi_{nd})^T\in \C^d$, the upper triangular matrix
\begin{align}
B_{\varepsilon}
=\left(\begin{matrix}
\hv_{-d}e^{2\pi d\varepsilon}   &\hv_{-d+1} e^{2\pi (d-1)\varepsilon}\  &... \ &\hv_{-1} e^{2\pi\varepsilon}\\
         &\hv_{-d}e^{2\pi d\varepsilon}         &      &\hv_{-2}e^{4\pi \varepsilon}\\
         \\  &          &\ddots
         \\  &           &   \\      
         &               &       &\hv_{-d}e^{2\pi d\varepsilon}
\end{matrix}\right)
=e^{2\pi d\varepsilon} F_{\varepsilon} B_0 F_{\varepsilon}^{-1},
\end{align}
where $F_{\varepsilon}=\mathrm{diag}(e^{2\pi \varepsilon}, e^{4\pi \varepsilon}, ..., e^{2\pi d\varepsilon})$.
The lower triangular matrix
$\tilde{B}_{\varepsilon}=e^{-2\pi d\varepsilon} F_{\varepsilon} B_0^*F_{\varepsilon}^{-1}$,
and 
\begin{align}
C_{\varepsilon}(\theta)
=&\left(
\begin{matrix}
2\cos2\pi(\theta+(d-1)\alpha) & \hv_1e^{-2\pi \varepsilon} & \hv_2 e^{-4\pi\varepsilon} &... & & \hv_{d-1} e^{-2\pi (d-1)\varepsilon}\\
\hv_{-1}e^{2\pi\varepsilon} & 2\cos2\pi(\theta+(d-2)\alpha) & \hv_1e^{-2\pi\varepsilon} &... &  & \hv_{d-2} e^{-2\pi (d-2)\varepsilon}\\
\hv_{-2}e^{4\pi\varepsilon} & \hv_{-1}e^{2\pi\varepsilon} &\ddots &\ddots & & \vdots \\
\vdots & & \ddots\\
& & & & &\hv_1 e^{-2\pi\varepsilon}\\
\hv_{-d+1}e^{2\pi(d-1)\varepsilon} &... &  & &\hv_{-1}e^{2\pi\varepsilon} &2\cos2\pi\theta
\end{matrix}
\right)\\
&=F_{\varepsilon} C_0(\theta)F_{\varepsilon}^{-1}.
\end{align}
In the rest of the paper, we shall write $C_{\varepsilon=0}(\theta)$ and $B_{\varepsilon=0}$ as $C(\theta)$ and $B$, and $M_E^{\varepsilon=0}=M_E$, $\hat{A}_E^{\varepsilon=0}=\hat{A}_E$,  respectively.
 The eigenvalue equation $\hat{H}_{\alpha,\theta,\varepsilon}\phi=E\phi$ can be reformulated as:
\begin{align}\label{eq:hH_transfer_block2}
\left(\begin{matrix}
B_{\varepsilon}\Phi_{n+1}\\
\Phi_n
\end{matrix}\right)=
M_E^{\varepsilon}(\theta+nd\alpha)
\left(\begin{matrix}
B_{\varepsilon}\Phi_n\\
\Phi_{n-1}
\end{matrix}\right),
\end{align}
where
\begin{align}\label{def:ME^eps2}
M_E^{\varepsilon}(\theta)
=&\left(\begin{matrix}
(E-C_{\varepsilon}(\theta))B_{\varepsilon}^{-1} &-\tilde{B}_{\varepsilon}\\
B_{\varepsilon}^{-1} &0
\end{matrix}\right)\\
=&\left(\begin{matrix}
e^{-2\pi d\varepsilon} F_{\varepsilon} (E-C(\theta)) B^{-1} F_{\varepsilon}^{-1} & -e^{-2\pi d\varepsilon} F_{\varepsilon}B^* F_{\varepsilon}^{-1}\\
e^{-2\pi d\varepsilon} F_{\varepsilon} B^{-1} F_{\varepsilon}^{-1} &0
\end{matrix}\right)\\
=&e^{-2\pi d\varepsilon} \mathrm{diag}(F_{\varepsilon}, F_{\varepsilon}) \cdot M_E(\theta) \cdot \mathrm{diag}(F_{\varepsilon}^{-1}, F_{\varepsilon}^{-1}).
\end{align}
Clearly this implies that for any $1\leq j\leq 2d$,
\begin{align}\label{eq:LM_eps=LM-eps}
L_j(d\alpha, M_E^\varepsilon)=L_j(d\alpha, M_E)-2\pi d\varepsilon.
\end{align}
It is easy to see that 
\begin{align}\label{eq:symplectic}
M_E^*\cdot \Omega \cdot  M_E=\Omega,
\end{align}
where 
\begin{align}\label{def:Omega}
\Omega:=\left(\begin{matrix}
0 & I_d\\
 -I_d &0
\end{matrix}\right).
\end{align}
Hence $M_E^0$ is complex symplectic matrix.
This implies that for any $1\leq j\leq 2d$,
\begin{align}\label{eq:Lj=L2d-j}
\hat{L}_j(d\alpha, M_E)=-\hat{L}_{2d+1-j}(d\alpha, M_E).
\end{align}
Let 
\begin{align}
M_{n,E}(\theta)=\left(\begin{array}{c|c}
M_{n,E}^{UL}(\theta) & M_{n,E}^{UR}(\theta)\\
\hline
M_{n,E}^{LL}(\theta) & M_{n,E}^{LR}(\theta)
\end{array}\right),
\end{align}
where each $M_{n,E}^{\star}$ is a $d\times d$ block, $\star=UL, UR, LL, LR$.
The following recursive relations hold: for $n=1$,
\begin{align}\label{eq:transfer_inductive_1}
\begin{cases}
M_{1,E}^{UL}(\theta)=-(C(\theta)-E)B^{-1}\\
M_{1,E}^{UR}(\theta)=-B^*\\
M_{1,E}^{LL}(\theta)=B^{-1}\\
M_{1,E}^{LR}(\theta)=0
\end{cases}
\end{align}
and for $n\geq 2$,
\begin{align}\label{eq:transfer_inductive_2}
\begin{cases}
M_{n,E}^{UL}(\theta)=-M_{n-1,E}^{UL}(\theta+d\alpha) (C(\theta)-E)B^{-1}+M_{n-1,E}^{UR}(\theta+d\alpha) B^{-1}\\
M_{n,E}^{UR}(\theta)=-M_{n-1,E}^{UL}(\theta+d\alpha) B^*\\
M_{n,E}^{LL}(\theta)=-M_{n-1,E}^{LL}(\theta+d\alpha) (C(\theta)-E)B^{-1}+M_{n-1,E}^{LR}(\theta+d\alpha)B^{-1}\\
M_{n,E}^{LR}(\theta)=-M_{n-1,E}^{LL}(\theta+d\alpha) B^*
\end{cases}
\end{align}

The eigenvalue equation $\hat{H}_{\alpha,\theta,\varepsilon}\phi=E\phi$ (when viewed as a scalar-valued Hamiltonian) can also be rewritten as:
\begin{align}\label{eq:tranfer_hH_varepsilon_1step}
\left(\begin{matrix}
\phi_{n+d}\\
\vdots\\
\phi_{n+1}\\
\phi_n\\
\vdots\\
\phi_{n-d+1}
\end{matrix}\right)
=\hat{A}_{E}^{\varepsilon}(\theta+n \alpha)
\left(\begin{matrix}
\phi_{n+d-1}\\
\vdots\\
\phi_{n}\\
\phi_{n-1}\\
\vdots\\
\phi_{n-d}\end{matrix}\right),
\end{align}
where 
\begin{equation}
\begin{split}
    &\hat{A}_{E}^{\varepsilon}(\theta):=
\frac{1}{\hat{v}_{-d}e^{2\pi d\varepsilon}}\\
&\left(\begin{matrix}
-\hat{v}_{1-d}e^{2\pi (d-1)\varepsilon} &\cdots &-\hat{v}_{-1}e^{2\pi \varepsilon} &E-2\cos(2\pi\theta) &| &-\hat{v}_1e^{-2\pi \varepsilon} &\cdots &-\hat{v}_{d-1}e^{-2\pi (d-1)\varepsilon} &-\hat{v}_d e^{-2\pi d\varepsilon}\\
\hat{v}_{-d}e^{2\pi d\varepsilon} & & & &|\\
&\ddots & & &| & \\
& & \hat{v}_{-d} e^{2\pi d\varepsilon} & &| &\\
\hline
& & &\hat{v}_{-d} e^{2\pi d\varepsilon} &| &\\
& & & &| &\hat{v}_{-d} e^{2\pi d\varepsilon}\\
& & & &| & &\ddots\\
& & & &| & & &\hat{v}_{-d} e^{2\pi d\varepsilon}
\end{matrix}\right),
\end{split}\label{def:hAE}
\end{equation}
is the 1-step transfer matrix.
By \eqref{eq:tranfer_hH_varepsilon_1step} and \eqref{eq:hH_transfer_block2} we have
\begin{align}\label{eq:M=A2}
M_{E}^{\varepsilon}(\theta)
=\mathrm{diag}(B_{\varepsilon}, I_d)\cdot
\hat{A}_{d,E}^{\varepsilon}(\alpha,\theta)\cdot
\mathrm{diag}(B_{\varepsilon}^{-1}, I_d).
\end{align}
Recall that  $\hat{A}_d$ signifies a product of $d$ factors of $\hat{A}$ matrices.
Hence for any $1\leq m\leq 2d$,
\begin{align}\label{eq:LM=d_LA2}
\hat{L}_m(d\alpha, M_{E}^{\varepsilon})=d\cdot \hat{L}_m(\alpha,\hat{A}_{E}^{\varepsilon}).
\end{align}
Combining this with \eqref{eq:LM_eps=LM-eps} yields
\begin{align}\label{eq:LA_eps=LA_0}
\hat{L}_m(\alpha, \hat{A}_{E}^{\varepsilon})=\hat{L}_m(\alpha, \hat{A}_{E})-2\pi\varepsilon.
\end{align}

\subsection{Avila's acceleration}
Let $(\alpha, A)\in (\T, C^{\omega}(\mathrm{SL}(2, \R)))$, and $A_{\varepsilon}(\theta):=A(\theta+i\varepsilon)$.
The (top) Lyapunov exponent $L_1(\alpha, A_{\varepsilon})$ is a convex and even function in $\varepsilon$.
Avila defined the acceleration to be the right-derivative as follows:
\begin{align}
\kappa(\alpha, A_{\varepsilon}):=\lim_{\delta\to 0^+} \frac{L_1(\alpha, A_{\varepsilon+\delta})-L_1(\alpha, A_{\varepsilon})}{2\pi \delta}.
\end{align}
As a cornerstone of his Global theory \cite{Global}, he showed that for $A\in \mathrm{SL}(2,\R)$ and irrational $\alpha$, $\kappa(\alpha, A_{\varepsilon})\in \Z$ is always quantized.

\subsection{Uniform upper semi-continuity}
The following lemma is an easy corollary of the arguments, essentially a subadditivity argument, in the proof of Lemma 5.1 in \cite{AJS}. 
\begin{lemma}\label{lem:upper_semi_cont}
For any small $\varepsilon>0$, there exists $N=N(\varepsilon, \alpha, v, E,d)$ such that for any $1\leq j\leq 2d$ and $n\geq N$,
\begin{align}
\frac{1}{n} \log \left\|\wedge^j M_{n,E}(\theta)\right\|\leq \hat{L}_{1}(d\alpha,\wedge^j M_{E})+\varepsilon=\hat{L}^j(d\alpha,M_{E})+\varepsilon,
\end{align}
uniformly in $\theta\in \T$.
\end{lemma}

\subsection{Large deviation estimates}
The following large deviation estimates play a crucial role in Section \ref{sec:Denominator}.
For Diophantine $\alpha$, a stronger form (with a deviation of size $n^{-\sigma}$) of the following estimate for $\mathrm{SL}(2,\R)$ cocycles was first proved in Lemma 1.1 of \cite{BG} by Bourgain and Goldstein, and further developed in sharp forms in \cites{GS1,GS2} by Goldstein and Schlag. In this paper we need the following weaker form, but it holds for any $\alpha$ such that $\beta(\alpha)=0$.
\begin{lemma}\label{lem:LDT}\cite{HZ}*{Theorem 1.1}
Let
\begin{align}\label{def:Bn_LDT}
\mathcal{B}_{n,E,\varepsilon}^d:=\left\{\theta:\, \frac{1}{n}\log \|\wedge^d M_{n,E}(\theta)\|\leq (1-\varepsilon)\cdot \hat{L}^d(d\alpha,M_E)\right\}
\end{align}
be the large deviation set of $\|\wedge^d M_{n,E}(\theta)\|$.
For $\alpha$ such that $\beta(\alpha)=0$, for any small $\varepsilon>0$, there exists $N=N(\varepsilon, \alpha, \hat{L}_d(d\alpha, M_E),v,d)>0$ and $c_1=c_1(v,d)>0$, such that for $n\geq N$,
\begin{align}\label{eq:mes_B}
\mathrm{mes}(\mathcal{B}_{n,E,\varepsilon}^d)\leq e^{-c_1 n \varepsilon^2\cdot  \hat{L}^d(d\alpha,M_E)}.
\end{align}
\end{lemma}

\begin{remark}
The above large deviation lemma was stated and proved in \cite{HZ} for $\mathrm{SL}(2,\R)$ cocycles, but it is clear that the same proof   applies to $(d\alpha, \wedge^d M_E)$.
\end{remark}

\subsection{Avalanche Principle}
The Avalanche Principle was first introduced for $\mathrm{SL}(2,\R)$ cocycles by Goldstein and Schlag in \cite{GS1}.
Later it was extended to higher-dimensional cocycles in \cites{S1,DK}.
The following form is from \cite{DK}. 

\begin{theorem}\label{thm:AP}\cite{DK}*{Proposition 2.42}
There exists $c_0, \tilde{C}_0>0$ such that given $0<\varepsilon<1$, $0<\kappa<c_0\varepsilon^2$ and $g_0,g_1,...,g_{n-1}\in \mathrm{Mat}(m,\R)\setminus\{0\}$, if 
\begin{align}
\frac{\sigma_1(g_j)}{\sigma_2(g_j)}>\frac{1}{\kappa}, \text{ for all } 0\leq j\leq n-1\\
\frac{\|g_j g_{j-1}\|}{\|g_j\| \|g_{j-1}\|}>\varepsilon, \text{ for all } 1\leq j\leq n-1
\end{align}
then
\begin{align}
\left| \log \|g_{n-1}\cdots g_0\|+\sum_{j=1}^{n-2}\log \|g_j\|-\sum_{j=1}^{n-1} \log \|g_j g_{j-1}\|\right|\leq \tilde{C}_0n \frac{\kappa}{\varepsilon^2}.
\end{align}
\end{theorem}

\subsection{Integrated density of states}
Let 
\begin{align}\label{def:H_Dirichlet}
\hat{H}_{\alpha,n}(\theta):=\chi_{[0,nd-1]}\hat{H}_{\alpha,\theta,\varepsilon=0}\,\chi_{[0,nd-1]}
\end{align}
be the restriction of the operator to the interval $[0,nd-1]$ with Dirichlet boundary condition. 
Let $E_{j}(\theta)$, $j=0,1,...,nd-1$ be the eigenvalues of $\hat{H}_{\alpha,n}(\theta)$.
Consider the sequence of measures
\begin{align}\label{def:finite_DOS}
    \hat{\mu}_{n,\alpha,\theta}(\mathrm{d}E):=\frac{1}{nd}\sum_{j=0}^{nd-1} \delta_{E_j(\theta)}(\mathrm{d}E).
\end{align}
It is well-known the weak limit exists and is independent of $\theta$ (up to measure zero):
\begin{align}\label{def:IDS}
    \lim_{n\to\infty} \hat{\mu}_{n,\alpha,\theta}(\mathrm{d}E)=\widehat{\mathcal{N}}(\mathrm{d}E).
\end{align}
The distribution function $\widehat{\mathcal{N}}(E)$ is called the integrated density of states.

The integrated density of states $\mathcal{N}$ are defined analogously for $H_{\alpha,\theta,\varepsilon=0}$ (indeed as the $d=1$ special case of the above). It is well-known that $\mathcal{N}=\widehat{\mathcal{N}}$, see e.g. \cite{HP}

The Thouless formula relates the Lyapunov exponent to the integrated density of states:
\begin{align}\label{eq:Thouless}
\hat{L}^d(\alpha, \hat{A}_E)+\log |\hat{v}_d|=\int_{\R} \log |E'-E|\, \widehat{\mathcal{N}}(\mathrm{d} E').
\end{align}
The scalar case $d=1$ was proved in \cites{AS,CS1}, and the $d>1$ case was proved in \cites{CS2,KS,HP}.

\section{Anderson localization and almost localization}\label{sec:Loc}
\subsection{Green's function}
The key to prove Anderson localization is to estimate the decay of Green's functions.
We consider the finite volume Hamiltonian with periodic boundary conditions. To distinguish  $\hat{H}_{\alpha,n}(\theta)$ (see \eqref{def:H_Dirichlet}), which exhibits Dirichlet boundary conditions, from the operator with periodic ones, we denote the latter by $P_{\alpha,n}(\theta)$. Moreover, for simplicity we shall omit $\alpha$ in the notation in the rest of the section.  Thus, 
\begin{align}\label{def:Pn}
P_{n}(\theta)=
\left(\begin{matrix}
C(\theta+(n-1)d\alpha) & B^* & &  &B\\
B &C(\theta+(n-2)d\alpha) &\ddots \\
& \ddots &\ddots &\ddots \\
& &\ddots &\ddots &B^*\\
B^* & & &B &C(\theta)
\end{matrix}\right).
\end{align}
Let 
\begin{align}\label{def:fn}
f_{E,n}(\theta):=\det(P_{n}(\theta)-E),
\end{align}
and
\begin{align}\label{def:Green}
G_{E,n}(\theta):=(P_{n}(\theta)-E)^{-1}
\end{align}
be the resolvent, which is also called the finite volume Green's function. 
It follows from Cramer's rule that 
\begin{align}\label{eq:Green_ratio}
G_{E,n}(\theta;x,y)=\frac{\mu_{n,x,y}(\theta)}{f_{E,n}(\theta)},
\end{align}
where $\mu_{n,x,y}(\theta)$ is the determinant of the submatrix of $P_{n}(\theta)-E$ obtained by deleting the $x$-th row and $y$-th column.
Let $\hat{u}$ be a solution to the eigenvalue equation $\hat{H}_{\alpha,\theta}\hat{u}=E\hat{u}$. 
For any $k\in \Z$
the following Poisson formula holds for all  $k\leq m\leq k+nd-1$:
\begin{align}
\hat{u}_m=&\sum_{y_1=0}^{d-1} G_{E,n}(\theta+k\alpha; m-k, y_1) \left(B^*\cdot \left(\begin{matrix}\hat{u}_{k+nd-1}-\hat{u}_{k+d-1}\\ \vdots\\ \hat{u}_{k+(n-1)d}-\hat{u}_k\end{matrix}\right)\right)_{y_1}\\
&\qquad\qquad+\sum_{y_2=(n-1)d}^{nd-1}G_{E,n}(\theta+k\alpha;m-k,y_2)\left(B \cdot \left(\begin{matrix}\hat{u}_{k+d-1}-\hat{u}_{k+nd-1}\\ \vdots \\ \hat{u}_k-\hat{u}_{k+(n-1)d}\end{matrix}\right)\right)_{y_2-(n-1)d},
\end{align}
in which $(M)_y$ refers to the element of vector $M$ in row $y$.
This implies 
\begin{align}\label{eq:Poisson}
|\hat{u}_m|\leq &4d^{3/2}\|B\|\cdot \max_{y\in \{0,...,d-1\}\cup \{(n-1)d,...,nd-1\}} |G_{E,n}(\theta+k\alpha; m-k, y)|\cdot \\
&\qquad\qquad\qquad\qquad\qquad\qquad\qquad\qquad\cdot \max_{\ell\in \{0,...,d-1\}\cup \{(n-1)d,...,nd-1\}} |\hat{u}_{k+\ell}|
\end{align}
For the remainder of this section, we shall write $\hat{L}^k(d\alpha, M_E), \hat{L}_k(d\alpha, M_E)$ as $\hat{L}^k, \hat{L}_k$, respectively. 
We split the estimate of the Green's function into Lemma~\ref{lem:numerator} for the numerator, and Lemma~\ref{lem:deno} for the denominator.

\begin{lemma}\label{lem:numerator}
Let $3d\leq y\leq (n-1)d-1$ and $0\leq x\leq d-1$ or $(n-1)d\leq x\leq nd-1$. Then 
for any $\varepsilon>0$, and uniformly in $\theta\in \T$, 
\begin{align}
|\mu_{n,x,y}(\theta)|\leq C_d\|B^{-1}\|^2 \cdot |\det B|^{n} \cdot \max\left(e^{\ell(1+\varepsilon) \hat{L}^{d-1}+(n-\ell)(1+\varepsilon) \hat{L}^d}+e^{\ell(1+\varepsilon)\hat{L}^d+(n-\ell)(1+\varepsilon)\hat{L}^{d-1}}\right),
\end{align}
provided $n>N(\varepsilon,\alpha,E,d,v)$ is large enough. 
Here $C_d$ is a constant depending only on~$d$.
\end{lemma}
We defer the proof to Section \ref{sec:Numerator}.

\begin{lemma}\label{lem:deno}
Let $N_0=N_0(\varepsilon,\alpha,E,d)$ and $\kappa_0=\kappa_0(\varepsilon,\alpha,E,d,N_0)$ be defined as in Lemma \ref{lem:Mn2} and $\mathcal{N}:=\{n\geq 8N_0: \|nd\alpha\|_{\T}\leq \kappa_0\}$ be the {\it $\kappa_0$-admissible} sequence.
For any $E\in \R$, for any small $0<\varepsilon<\hat{L}_d/(11\hat{L}^d)$, and for $n\geq N(\varepsilon,\alpha,E,d,v)$ and $n\in \mathcal{N}$, we have
\begin{align}
\mathrm{mes}\left\{\theta\in \T: \frac{1}{n}\log |f_{E,n}(\theta)|<|\det B|+(1-8\varepsilon)\hat{L}^d\right\}<\frac{1}{10}.
\end{align}
\end{lemma}
We present the proof of Lemma \ref{lem:deno} in Section \ref{sec:Denominator}.
Although Lemma \ref{lem:deno} does not hold for arbitrary $n$, it holds for all $\kappa_0$-admissible $n$'s, which form a positive density subsequence. Indeed,  there exists $C_{\kappa_0,d,\alpha}$ such that for any $n\in \N$, there is a $\kappa_0$-admissible $\tilde{n}$ such that $|n-\tilde{n}|\leq C_{\kappa_0,d,\alpha}$.
Hence Lemma \ref{lem:deno} yields the following.
\begin{corollary}\label{cor:deno}
For any $E\in \R$, for any small $0<\varepsilon<\hat{L}_d/(11\hat{L}^d)$, let $N_0, \kappa_0$ be as in Lemma~\ref{lem:deno}.
There exists $C_{\kappa_0,d,\alpha}>0$ such that for any $n\geq N(\varepsilon,\alpha,E,d,v)$, there exists $\kappa_0$-admissible $\tilde{n}$, $|\tilde{n}-n|\leq C_{\kappa_0,d,\alpha}$, so that 
\begin{align}
\mathrm{mes}\left\{\theta\in \T: \frac{1}{\tilde{n}}\log |f_{E,\tilde{n}}(\theta)|<|\det B|+(1-8\varepsilon)\hat{L}^d\right\}<\frac{1}{10}.
\end{align}
\end{corollary}

\begin{remark}\label{rmk:Dirichlet_vs_periodic}
The reason why we consider $P_{\alpha,n}(\theta)$ with periodic boundary conditions, instead of the finite volume operator $\hat{H}_{\alpha,n}$ on $[0,nd-1]$ with Dirichlet boundary conditions, lies with the differences between the denominators of their respective Green's functions. In fact, 
as we will show in Lemma~\ref{lem:P=M-I} $$|f_{E,n}(\theta)|=|\det B|^{n}\cdot |\det(M_{n,E}(\theta)-I_{2d})|,$$ while 
$$|\det (\hat{H}_{\alpha,n}(\theta)-E)|=|\det B|^{n} \cdot |\det M_{n,E}^{UL}(\theta)|,$$
see~\cite{BGV}. 
To prove non-perturbative Anderson localization, it is crucial to obtain a lower bound \[|\det(M_{n,E}(\theta)-I_{2d})|\gtrsim \exp(n(\hat{L}^d-\varepsilon))\] for certain $\theta$'s, which we establish in Lemma~\ref{lem:lower} along the $\kappa_0$-admissible sequence. It appears to be more difficult to obtain a similar result for $|\det M_{n,E}^{UL}(\theta)|$. One may attempt to derive the inequality $$|\det M_{n,E}^{UL}(\theta)|\gtrsim \exp(n(\hat{L}^d-\varepsilon))$$ via Thouless' formula.
For $E\notin \R$, such a lower bound for $|\det M_{n,E}^{UL}(\theta)|$ can indeed be deduced directly from the Thouless formula (proved for scalar case $(d=1)$ in \cites{AS,CS1}, and operators on the strip in \cites{CS2,KS,HP}). 
In fact, if $E\notin \R$, then $g(E')=\log |E-E'|$ is a continuous function on $\R$. By the $\sigma^*$-convergence of the finite volume density of states measure to the (infinite volume) density states measure $\widehat{\mathcal{N}}$, it follows that 
\begin{align}
\lim_{n\to\infty} \frac{1}{nd} (\log \left|\det M_{n,E}^{UL}(\theta)\right|+n\log|\det B|)&=\lim_{n\to\infty}\frac{1}{nd}\log |\det (\hat{H}_{\alpha,n}(\theta)-E)|\\
    &=\int_{\T}\log |E-E'|  \, \widehat{\mathcal{N}}(\mathrm{d}E')= \frac{1}{d} \hat{L}^d+\log |\hat{v}_d|
\end{align}
 for a.e.\ $\theta$. 
However this simple proof does not apply to $E\in \R$, since $g$ is not continuous. 
See also~\eqref{eq:wrong}.
\end{remark}

We emphasize the robustness of the lower bound in Lemma \ref{lem:deno} for the denominator. It can be applied to operators $\hat{H}_{\alpha,\theta}$ on the strip with general analytic potentials on higher-dimensional tori. We will use it in an upcoming work to address non-perturbative Anderson localization in the general case. 
For the special case of cosine potential, a simpler lower bound can be obtained via Herman's subharmonic argument.

\begin{lemma}\label{lem:deno_ARC}
For any $n\geq 1$, we have
\begin{align}
    \frac{1}{n}\int_{\T}\log |f_{E,n}(\theta)|\, \mathrm{d}\theta\geq 0=\hat{L}^d+d\log |\hat{v}_d|.
\end{align}
\end{lemma}
We present the proof at the end of Section \ref{sec:Denominator}. Note the equality sign above is exactly Haro-Puig's theorem (see Theorem \ref{thm:HP}) combined with the subcritical assumption on the Schr\"odinger cocycle.

\subsection{Almost localization}
The proof of almost localization follows the scheme of \cite{AJ2}, combined with our Lemmas \ref{lem:numerator} and \ref{lem:deno}.
A minor adjustment is that \cite{AJ2} deals with $\alpha\in \mathrm{DC}$ while we assume a slightly weaker condition $\beta(\alpha)=0$. 
The almost localization for the $\beta(\alpha)=0$ case was previously addressed for the extended Harper's model in \cite{DryTen}.
It should be sufficient to assume that $0\leq \beta(\alpha)<c$ for some small constant $c>0$, but we do not pursue this improvement here.

Let $C_3>0$ be the universal constant in Lemma \ref{lem:uniform} below. 
Fix 
\begin{align}\label{def:epsilon_0}
0<\varepsilon_0<\frac{\hat{L}_d(\alpha, \hat{A}_E)}{30\hat{L}^d(\alpha, \hat{A}_E)+100C_3d}<\frac{\hat{L}_d}{11\hat{L}^d}.
\end{align}

\begin{definition}[$\varepsilon_0$-resonance of $\theta$]
Let $\alpha,\theta\in \T$ and constant $\varepsilon_0>0$. We say $k\in \Z$ is an $\varepsilon_0$-resonance of $\theta$ if $\|2\theta-k\alpha\|_{\T}\leq e^{-\varepsilon_0 |k|}$ and $\|2\theta-k\alpha\|_{\T}=\min_{|j|\leq k} \|2\theta-j\alpha\|_{\T}$.
\end{definition}
Since $\beta(\alpha)=0$, there exists $C_0=C_0(\alpha,\varepsilon_1)>0$ such that for $n\neq 0$, 
\begin{align}\label{eq:n_alpha_varepsilon_1}
\|n\alpha\|_{\T}\geq C_0 e^{-\varepsilon_1 |n|}.
\end{align}
In the following we shall fix $\varepsilon_1=\varepsilon_0/1000$.
This implies the following. 
\begin{lemma}\label{lem:resonant_smallest}
If $|k|$ is large enough, and $\|2\theta-k\alpha\|_{\T}\leq e^{-\varepsilon_0 |k|}$, then $\|2\theta-k\alpha\|_{\T}=\min_{|j|\leq k} \|2\theta-j\alpha\|_{\T}$.
\end{lemma}
\begin{proof}
In fact,  for $|j|\leq k$ and $j\neq k$,
\begin{align}\label{eq:res_suf}
\|2\theta-j\alpha\|_{\T}
\geq &\|(j-k)\alpha\|_{\T}-\|2\theta-k\alpha\|_{\T}\\
\geq &C_0 e^{-\varepsilon_0 |j-k|/1000}-e^{-\varepsilon_0 |k|}\\
\geq &C_0 e^{-\varepsilon_0 |k|/500}-e^{-\varepsilon_0 |k|}\\
\geq &e^{-\varepsilon_0 |k|}\geq \|2\theta-k\alpha\|_{\T},
\end{align}
in which we used \eqref{eq:n_alpha_varepsilon_1} to estimate $\|(j-k)\alpha\|_{\T}$.
\end{proof}

\begin{definition}
Let $0=|n_0|\leq |n_1|\leq |n_2|\leq ...$ be the $\varepsilon_0$-resonances of $\theta$. If this sequence is infinite, we
say $\theta$ is $\varepsilon_0$-resonant, otherwise we say it is $\varepsilon_0$-non-resonant. Furthermore, if $\theta$ is $\varepsilon_0$-non-resonant with a finite sequence of resonances $0=|n_0|\leq |n_1|\leq ...\leq |n_j|$, we let $n_{j+1}=\infty$.
\end{definition}

Due to $\beta(\alpha)=0$, we can show the following growth of the infinite sequence of $\theta$ resonances.
\begin{lemma}\label{eq:nj_jump}
Let $\beta(\alpha)=0$, then there exists $N_0=N_0(\alpha,\varepsilon_0)$, such that for $|n_j|>N_0$, 
\begin{align}
|n_{j+1}|>250 |n_j|
\end{align}
\end{lemma}
\begin{proof}
By \eqref{eq:n_alpha_varepsilon_1}, 
\begin{align}
e^{-\varepsilon_0 |n_j|}\geq \|2\theta-n_j\alpha\|_{\T}
\geq &\|(n_{j+1}-n_j)\alpha\|_{\T}-\|2\theta-n_{j+1}\alpha\|_{\T}\\
\geq &C_0 e^{-\varepsilon_0 |n_{j+1}-n_j|/1000}-e^{-\varepsilon_0 |n_{j+1}|}\\
\geq &C_0 e^{-\varepsilon_0 |n_{j+1}|/500}-e^{-\varepsilon_0 |n_{j+1}|}\\
\geq &\frac{C_0}{2} e^{-\varepsilon_0 |n_{j+1}|/500},
\end{align}
which implies the desired bound.
\end{proof}

\begin{definition}[Almost localization]\label{def:almost_AL}
We say the family $\{\hat{H}_{\alpha,\theta}\}_{\theta\in \T}$ is $(C_1, C_2,\varepsilon_2)$-{\em almost localized} for some constant $C_1, C_2, \varepsilon_2>0$ if for every solution $\hat{u}$ of $\hat{H}_{\alpha,\theta}\hat{u}=E\hat{u}$ for some energy $E\in \R$, satisfying $\hat{u}_0=1$ and $|\hat{u}_k|\leq 1+|k|$, and for every $C_1(1+|n_j|)<|k|<C_1^{-1}|n_{j+1}|$, the bound $|\hat{u}_k|\leq C_2 e^{-\varepsilon_2 |k|}$ holds, where the $n_j$'s are the $\varepsilon_0$-resonances of $\theta$.
\end{definition}

We will prove the following.
\begin{theorem}\label{thm:almost_AL}
There exists a constant $C_2=C_2(v,d)>0$ such that 
$\{\hat{H}_{\alpha,\theta}\}_{\theta\in \T}$ is $(3, C_2, \hat{L}_d/10)$-almost localized.
\end{theorem} 
\begin{proof}
Let $f_{E,N}(\theta)$ be defined as in \eqref{def:fn}, and we omit the dependence on $E$ for simplicity.
The following observation was made for the Harper operator (with Dirichlet boundary conditions) in \cite{J99}, and extended to the long-range setting in \cite{AJ2}.
The extension to the periodic boundary condition setting is immediate.
\begin{fact}\label{fact:Pn_polynomial}
$f_{N}(\theta)=Q_{dN}(\cos(2\pi (\theta+\frac{dN-1}{2}\alpha)))$, where $Q_k(x)$ denotes a real polynomial in $x$ of degree~$k$.
\end{fact}

Next we introduce of definition of uniformity, introduced in \cite{J99}, that measures the evenness of distribution of a set in $\T$.
\begin{definition}[Uniformity]\label{def:uniform}
For some $\kappa>0$, we say that a set $\{\theta_1,\theta_2,...,\theta_{m+1}\}$ is $\kappa$-uniform if 
\begin{align*}
\max_{z\in [0,1]}\max_{j=1,\cdots,m+1} \Bigg|\prod_{\substack{\ell=1\\ \ell \neq j}}^{m+1} \frac{z-\cos (2\pi \theta_{\ell})}{\cos (2\pi \theta_j)-\cos (2\pi \theta_{\ell})}\Bigg|<e^{\kappa m}.
\end{align*}
\end{definition}

Without loss of generality, assume $3(|n_j|+1)<y<|n_{j+1}|/3$ and $|n_j|>N_0(\alpha,\varepsilon_0)$. The case of negative $y$ can be handled similarly. Choose $n$ such that 
$dq_n\leq y/8< dq_{n+1}$ and let $s$ be the largest positive integer satisfying $sd q_n\leq y/8$.
Define intervals $I_1, I_2\subset \Z$ as follows:
\begin{align*}
I_1=[-2sd q_n, 0]\cap \Z\text{\ \ and\ \ }I_2=[y-2sd q_n+1, y+2sd q_n]\cap \Z,\,\, &\mathrm{if}\ n_j<0      \\
I_1=[-1, 2sd q_n-1]\cap \Z\text{\  \ and\ \ }I_2=[y-2sd q_n+1, y+2sd q_n]\cap \Z,\,\, &\mathrm{if}\ n_j\geq 0
\end{align*}
The key in choosing $I_1, I_2$ is to guarantee that $-n_j\notin (I_1\cup I_2)+(I_1\cup I_2)$.
Let $6sq_n=N$ whence $|I_1|+|I_2|=6dsq_n+1=dN+1$.

\begin{lemma}\label{lem:uniform}
The set 
$$\left\{\theta+j\alpha-\left\lfloor\frac{dN-1}{2}\right\rfloor \alpha+\frac{dN-1}{2}\alpha\right\}_{j\in I_1\cup I_2}\subset \T$$ 
is $(C_3d\, \varepsilon_0)$-uniform for some absolute  constant $C_3\geq 1$ and $y>y_0(\alpha, \varepsilon_0, C_{\kappa_0})$.
\end{lemma}
For $d=1$, when $N-1$ is even, this is Lemma 4.3 of \cite{DryTen}. In the general case, the proof requires a slight modification, which we present below.
\begin{proof}
We treat the $n_j>0$ case in detail and leave the other case to the reader. Note that for $y$ large, $n$ is also large.
Set $z=\cos(2\pi \xi)$, and denote 
\begin{align}\label{def:ttheta_j}
\tilde{\theta}_j:=\theta+j\alpha-\left\lfloor\frac{dN-1}{2}\right\rfloor\alpha+\frac{dN-1}{2}\alpha.
\end{align}
Note that since 
\begin{align}\label{eq:the_-1_0}
    &\max_{z\in [0,1]} \max_{\substack{j\in I_1\cup I_2\\ j\neq -1}} \left| \frac{z-\cos(2\pi\tilde{\theta}_{-1})}{\cos(2\pi\tilde{\theta}_j)-\cos(2\pi\tilde{\theta}_{-1})}\right|\\
    \leq &2\max_{\substack{j\in I_1\cup I_2\\ j\neq -1}} |\cos(2\pi\tilde{\theta}_j)-\cos(2\pi\tilde{\theta}_{-1})|^{-1}\\
    \leq &C\max_{\substack{j\in I_1\cup I_2\\ j\neq -1}} \left(\|\tilde{\theta}_j-\tilde{\theta}_{-1}\|_{\T}^{-1}\cdot \|\tilde{\theta}_j+\tilde{\theta}_{-1}\|_{\T}^{-1}\right)
\end{align}
We use \eqref{eq:n_alpha_varepsilon_1} with $\varepsilon_1=\varepsilon_0/10$ to estimate
\begin{align}\label{eq:the_-1_1}
    \|\tilde{\theta}_j-\tilde{\theta}_{-1}\|_{\T}=\|(j+1)\alpha\|_{\T}\geq C_{\varepsilon_0} e^{-\frac{\varepsilon_0}{10}|j+1|}\geq C_{\varepsilon_0} e^{-\frac{\varepsilon_0}{10}(y+2sdq_n+1)}\geq C_{\varepsilon_0} e^{-\frac{3}{10}\varepsilon_0 dN},
\end{align}
where we used $y+2sdq_n\leq 18sdq_n=3dN$.
Regarding the term
\begin{align}\label{eq:the_-1_2}
\|\tilde{\theta}_j+\tilde{\theta}_{-1}\|_{\T}
=&\|2\theta+(j-1)\alpha-2\left\lfloor (dN-1)/2 \right\rfloor \alpha+(dN-1)\alpha\|_{\T}
\end{align}
we distinguish into two cases.

\underline{Case 1.} 
Suppose
$\|2\theta+(j-1)\alpha-2\left\lfloor (dN-1)/2 \right\rfloor \alpha+(dN-1)\alpha\|_{\T}<e^{-\varepsilon_0 |j-1-2\left\lfloor (dN-1)/2\right\rfloor+(dN-1)|}$.
Then by Lemma \ref{lem:resonant_smallest}, we have $-n_j< -1\leq j-1-2\left\lfloor (dN-1)/2\right\rfloor+(dN-1)=-n_p$ for some $p\leq n-1$.
This implies
\begin{align}\label{eq:the_-1_3}
\|2\theta+(j-1)\alpha-2\left\lfloor (dN-1)/2 \right\rfloor \alpha+(dN-1)\alpha\|_{\T}
=&\|2\theta-n_p\alpha\|_{\T}\\
\geq &\|(n_j-n_p)\alpha\|_{\T}-\|2\theta-n_j\alpha\|_{\T}\\
\geq &C_{\varepsilon_0}e^{-\frac{\varepsilon_0}{10}|n_j-n_p|}-e^{-\varepsilon_0 |n_j|}\\
\geq &C_{\varepsilon_0}e^{-\frac{\varepsilon_0}{5}|n_j|}-e^{-\varepsilon_0 |n_j|}\\
\geq &C_{\varepsilon_0} e^{-\frac{\varepsilon_0}{4}|n_j|}\\
\geq &C_{\varepsilon_0} e^{-\frac{\varepsilon_0}{4}dN},
\end{align}
where in the last inequality we used $|n_j|<y/3\leq dN$.

\underline{Case 2.} 
Suppose $\|2\theta+(j-1)\alpha-2\left\lfloor (dN-1)/2 \right\rfloor \alpha+(dN-1)\alpha\|_{\T}\geq e^{-\varepsilon_0 |j-1-2\left\lfloor (dN-1)/2\right\rfloor+(dN-1)|}$. Then
we directly bound
\begin{align}\label{eq:the_-1_4}
\|2\theta+(j-1)\alpha-2\left\lfloor (dN-1)/2 \right\rfloor \alpha+(dN-1)\alpha\|_{\T}
\geq &e^{-\varepsilon_0 |j-1-2\left\lfloor (dN-1)/2\right\rfloor+(dN-1)|}\\
\geq &e^{-3\varepsilon_0 dN},
\end{align}
where we used $j-1-2\left\lfloor (dN-1)/2\right\rfloor+(dN-1)\leq y+2sdq_n\leq 3dN$.
Taking \eqref{eq:the_-1_0}, \eqref{eq:the_-1_1}, \eqref{eq:the_-1_2}, \eqref{eq:the_-1_3} and \eqref{eq:the_-1_4} into account, it is enough to prove the claimed result with $I_1=[-1,2sdq_n-1]$ replaced with $I_1=[0,2sdq_n-1]$.

\begin{lemma}\label{smallest}\cite{AJ1}
Let $\alpha\in \R\backslash\Q$, $x\in\R$ and $0\leq j_0\leq q_n-1$ be such that $|\sin\pi(x+j_0\alpha)|=\inf_{0\leq j\leq q_n-1}|\sin\pi(x+j\alpha)|$, then for some absolute constant $C_4>0$,
$$-C_4\log q_n\leq \sum_{0\leq j\leq q_n-1, j\neq j_0} \log|\sin\pi(x+j\alpha)|+(q_n-1)\log 2\leq C_4\log q_n$$
\end{lemma}
For any $j\in I_1\cup I_2$, it suffices to estimate
\begin{align*}
\sum_{\ell\in I_1\cup I_2,\ \ell\neq j}\left( \log{|\cos(2\pi \xi)-\cos(2\pi \tilde{\theta}_{\ell})|}-\log{|\cos(2\pi \tilde{\theta}_j)-\cos(2\pi \tilde{\theta}_{\ell})|}\right) =:\sum_1-\sum_2.
\end{align*}
Lemma \ref{smallest} reduces this problem to estimating the minimal terms. We first bound $\sum_1$:
\begin{align}\label{eq:sum1}
\sum_{1}
&=\sum_{\ell\in I_1\cup I_2, \ell\neq j} \log|\cos(2\pi \xi)-\cos(2\pi\tilde{\theta}_{\ell})| \notag\\
&=\sum_{\ell\in I_1\cup I_2, \ell\neq j} \log|\sin{(\pi(\xi+\tilde{\theta}_{\ell}))}|+\sum_{\ell\in I_1\cup I_2, \ell\neq j} \log|\sin(\pi(\xi-\tilde{\theta}_{\ell}))|+(|I_1|+|I_2|-1)\log 2 \notag\\
&=:\sum_{1,+}+\sum_{1,-}+dN\log 2.\\
\end{align}

Dividing $I_1\cup I_2=\bigcup_{k=1}^{6sd}T_k$ into $6sd$ small intervals of integers, each of length $q_n$, and denoting
\begin{align}
|\sin(\pi (\xi+\tilde{\theta}_{\ell_k}))|:=\min_{\ell\in T_k} |\sin(\pi(\xi+\tilde{\theta}_{\ell}))|,
\end{align}
one has 
\begin{align}\label{eq:sum1+_2}
    \sum_{\ell\in I_1\cup I_2, \ell\neq j} \log |\sin(\pi(\xi+\tilde{\theta}_{\ell}))|
    \leq &\sum_{k=1}^{6sd}\sum_{\ell\in T_k, \ell\neq \ell_k}\log |\sin(\pi(\xi+\tilde{\theta}_{\ell}))| \notag\\
    \leq &6sd(C_4\log q_n-(q_n-1)\log 2),
\end{align}
where we used Lemma \ref{smallest} in the last inequality.
Similarly, replacing $\xi$ with $-\xi$, 
\begin{align}\label{eq:sum1-_1}
\sum_{1,-}\leq 6sq(C_4\log q_n-(q_n-1)\log 2).
\end{align}
Combining \eqref{eq:sum1} with \eqref{eq:sum1+_2} and \eqref{eq:sum1-_1}, we arrive at
\begin{align}\label{eq:sum1_final}
    \sum_1
    \leq &12sd(C_4\log q_n-(q_n-1)\log 2)+dN \log 2 \notag\\
    \leq &-6sdq_n\log 2+\varepsilon_0 sdq_n,
\end{align}
for $n$ large.
Next, we estimate $\sum_2$:
\begin{align}\label{eq:sum2_1}
\sum_2=&\sum_{\ell\in I_1\cup I_2, \ell\neq j} \log |\cos(2\pi\tilde{\theta}_j)-\cos(2\pi\tilde{\theta}_{\ell})|\notag\\
=&\sum_{\ell\in I_1\cup I_2, \ell\neq j}\log |\sin(\pi(\tilde{\theta}_j+\tilde{\theta}_{\ell}))|+\sum_{\ell\in I_1\cup I_2, \ell\neq j}\log |\sin(\pi(\tilde{\theta}_j-\tilde{\theta}_{\ell}))|+dN\log 2\notag\\
=:&\sum_{2,+}+\sum_{2,-}+dN\log 2.
\end{align}

Let $T_1=[0, q_n-1]\cap \Z$ be the left-most small interval in $I_1\cup I_2=\cup_{k=1}^{6sd}T_k$, and for $k\geq 1$, let $T_{k+1}$ be on the right-hand side of $T_k$.
Then $\bigcup_{k=1}^{2sd}T_k=I_1$ and $\bigcup_{k=2sd+1}^{6sd}T_k=I_2$.
We also group $T_k$'s into intervals of length $sq_n$, for $1\leq h\leq 6d$, by defining
\begin{align}
    R_h=\bigcup_{k=(h-1)s+1}^{hs} T_k.
\end{align}
Denote the smallest term in each $T_k$, for $1\leq k\leq 6sd$, by:
\begin{align}
|\sin(\pi(\tilde{\theta}_j+\tilde{\theta}_{\ell_k^+}))|=\min_{\ell\in T_k}|\sin(\pi(\tilde{\theta}_j+\tilde{\theta}_{\ell}))|,
\end{align}
and
\begin{align}
|\sin(\pi(\tilde{\theta}_j-\tilde{\theta}_{\ell_k^-}))|=\min_{\ell\in T_k}|\sin(\pi(\tilde{\theta}_j-\tilde{\theta}_{\ell}))|.
\end{align}
Denote the smallest term in each $R_h$, for $1\leq h\leq 6d$, by:
\begin{align}
|\sin(\pi(\tilde{\theta}_j+\tilde{\theta}_{\tilde{\ell}_h^+}))|=\min_{\ell\in R_h}|\sin(\pi(\tilde{\theta}_j+\tilde{\theta}_{\ell}))|,
\end{align}
and
\begin{align}
|\sin(\pi(\tilde{\theta}_j-\tilde{\theta}_{\tilde{\ell}_h^-}))|=\min_{\ell\in R_h}|\sin(\pi(\tilde{\theta}_j-\tilde{\theta}_{\ell}))|.
\end{align}
We now establish the following estimate for $\sum_{2,+}$.
\begin{lemma}\label{lem:sum2+}
For $n$ large, 
\begin{align}
    \sum_{2,+}\geq -300\varepsilon_0 d^2 sq_n-6sdq_n\log 2.
\end{align}
\end{lemma}
\begin{proof}
We bound
\begin{align}\label{eq:sum2+}
    \sum_{2,+}
    =&\sum_{k=1}^{6sd}\sum_{\substack{\ell\in T_k\\ \ell\neq j}}\log |\sin(\pi(\tilde{\theta}_j+\tilde{\theta}_{\ell}))|\\
    \geq &\sum_{k=1}^{6sd}\sum_{\ell\in T_k}\log |\sin(\pi(\tilde{\theta}_j+\tilde{\theta}_{\ell}))|\\
    \geq &\sum_{k=1}^{6sd} \log |\sin(\pi(\tilde{\theta}_j+\tilde{\theta}_{\ell_k^+}))|-(6sd)(C_4\log q_n+(q_n-1)\log 2)\\
    \geq &\sum_{k=1}^{6sd} \log (\|\tilde{\theta}_j+\tilde{\theta}_{\ell_k^+}\|_{\T})-6sdq_n\log 2-\varepsilon_0 sdq_n,
\end{align}
for $n$ large, where we used Lemma \ref{smallest} in the second to last inequality, and $|\sin(\pi x)|\geq 2|x|$ in the final one.
As far as the minimal terms at $\ell_k^+$, $1\leq k\leq 6sd$ are concerned, we now establish the following result. 
\begin{lemma}\label{lem:sum2+_min}
    \begin{align}
        \min_{k=1,\ldots,6sd} \; \|\tilde{\theta}_j+\tilde{\theta}_{\ell_{k}^+}\|_{\T}\geq e^{-36\varepsilon_0 dsq_n}.
    \end{align}
\end{lemma}
\begin{proof}
Within the proof, for $1\leq k\leq 6sd$, let 
\begin{align}
m_k:=j+\ell_{k}^+-2\lfloor (dN-1)/2\rfloor+(dN-1)
\end{align}
for simplicity. Then $\tilde{\theta}_j+\tilde{\theta}_{\ell_{k}^+}=2\theta+m_k\alpha$.
Note that since $j, \ell_k^+\in I_1\cup I_2$, one has
\begin{align}\label{eq:mk<36sdqn}
    -n_j<-3\leq m_k
    \leq 2y+4sdq_n
    \leq 
    \begin{cases}
        \frac{5}{2}y<\frac{5}{6}|n_{j+1}|\\
        16(s+1)dq_n+4sdq_n\leq 36sdq_n
    \end{cases}
\end{align}
Next, we prove Lemma \ref{lem:sum2+_min} by contradiction. Suppose for some $k$ that
\begin{align}\label{eq:assume_j+k1_extreme}
    \|\tilde{\theta}_j+\tilde{\theta}_{\ell_{k}^+}\|_{\T}
    =\left\|2\theta+m_k\alpha\right\|_{\T}
    \leq e^{-36\varepsilon_0 dsq_n}
\end{align}
Then by \eqref{eq:mk<36sdqn}, we conclude that 
\begin{align}\label{eq:mr_resonant}
    \left\|2\theta+m_k\alpha\right\|_{\T}\leq e^{-36\varepsilon_0 dsq_n}\leq e^{-\varepsilon_0 |m_k|}.
\end{align}
Furthermore, for any $0\leq |m|\leq 6sdq_n$, 
 by \eqref{eq:mk<36sdqn} 
\begin{align}
    |m-m_k|\leq 42 sdq_n.
\end{align}
Hence by \eqref{eq:n_alpha_varepsilon_1} with $\varepsilon_1=\varepsilon_0/10$, one obtains 
\begin{align}\label{eq:malpha_large}
    \|2\theta+m\alpha\|_{\T}
    \geq &\left\|(m-m_k)\alpha\right\|_{\T}-\|2\theta+m_k\alpha\|_{\T}\\
    \geq &C_{\varepsilon_0} e^{-\frac{21}{5}\varepsilon_0 sdq_n}-e^{-36\varepsilon_0 dsq_n}\\
    >&e^{-36\varepsilon_0 dsq_n}\\
    \geq &\|2\theta+m_k\alpha\|_{\T},
\end{align}
for $n$ large enough.
Combining \eqref{eq:mr_resonant} with \eqref{eq:malpha_large}, we see that $(-m_k)$ must be an $\varepsilon_0$-resonance.

Since $3|n_j|<y<8(s+1)dq_n\leq 16sdq_n$, we obtain $|n_j|\leq 6sdq_n$. 
Hence by \eqref{eq:malpha_large} with $m=-n_j$, we infer that 
\begin{align}
    \|2\theta-n_j\alpha\|_{\T}>\|2\theta+m_k\alpha\|_{\T}.
\end{align}
This implies $-m_k=n_p$ for some $p\geq j+1$. However, this leads to a contradiction with \eqref{eq:mk<36sdqn} via $|m_k|<|n_{j+1}|$.
\end{proof}

With Lemma \ref{lem:sum2+_min}, we conclude the following.
\begin{lemma}\label{lem:sumSh}
    For any $1\leq h\leq 6d$, 
    \begin{align}
        \sum_{\ell_k^+\in R_h}\log (\|\tilde{\theta}_j+\tilde{\theta_{\ell_k^+}}\|)\geq -40\varepsilon_0 dsq_n.
    \end{align}
\end{lemma}
\begin{proof}
We divide the argument into two cases.

\underline{Case 1}. Suppose $2\|\tilde{\theta}_j+\tilde{\theta}_{\tilde{\ell}_h^+}\|_{\T}\leq |\sin(\pi(\tilde{\theta}_j+\tilde{\theta}_{\tilde{\ell}_h^+}))|\leq e^{-2\varepsilon_0 q_n}$.
For $\ell_k^+, \tilde{\ell}_h^+\in R_h$ and $\ell_k^+\neq \tilde{\ell}_h^+$, we have $|\ell_k^+-\tilde{\ell}_h^+|\leq sq_n$. 
Since $sdq_n\leq y/8<dq_{n+1}$, in particular $sq_n<q_{n+1}$, it follows 
by \eqref{eq:qn_alpha_min} and \eqref{eq:n_alpha_varepsilon_1} with $\varepsilon_1=\varepsilon_0$ that
\begin{align}
    \|(\ell_k^+-\tilde{\ell}_h^+)\alpha\|_{\T}\geq \|q_n\alpha\|_{\T}\geq C_{\varepsilon_0} e^{-\varepsilon_0 q_n}.
\end{align}
This implies, by the triangle inequality, that 
\begin{align}\label{eq:sumSh_1}
\|\tilde{\theta}_j+\tilde{\theta}_{\ell_k^+}\|_{\T}
\geq &\|(\ell_k^+-\tilde{\ell}_h^+)\alpha\|_{\T}-\|\tilde{\theta}_j+\tilde{\theta}_{\tilde{\ell}_h^+}\|_{\T}\\
\geq &C_{\varepsilon_0} e^{-\varepsilon_0 q_n}-\frac{1}{2}e^{-2\varepsilon_0 q_n}\\
\geq &e^{-2\varepsilon_0 q_n},
\end{align}
for $n$ large. 
Combining \eqref{eq:sumSh_1} with Lemma \ref{lem:sum2+_min}, yields
\begin{align}\label{eq:sumSh_2}
    \sum_{\ell_k^+\in R_h}\log(\|\tilde{\theta}_j+\tilde{\theta}_{\ell_k^+}\|_{\T})
    \geq &s\log 2+(s-1) \log (e^{-2\varepsilon_0 q_n})+\log (e^{-36\varepsilon_0 dsq_n})\\
    \geq &-40\varepsilon_0 dsq_n.
\end{align}

\underline{Case 2}. Assume $\pi\|\tilde{\theta}_j+\tilde{\theta}_{\tilde{\ell}_h^+}\|_{\T}\geq |\sin(\pi(\tilde{\theta}_j+\tilde{\theta}_{\tilde{\ell}_h^+}))|> e^{-2\varepsilon_0 q_n}$.
Then for any $\ell_k^+\in R_h$, 
\[\pi\|\tilde{\theta}_j+\tilde{\theta}_{\ell_k^+}\|_{\T}\geq |\sin(\pi(\tilde{\theta}_j+\tilde{\theta}_{\ell_k^+}))|\geq e^{-2\varepsilon_0 q_n}.\]
This implies 
\begin{align}\label{eq:sumSh_3}
    \sum_{\ell_k^+\in R_h}\log(\|\tilde{\theta}_j+\tilde{\theta}_{\ell_k^+}\|_{\T})
    \geq &s\log 2+s\log (e^{-2\varepsilon_0 q_n})\geq -40\varepsilon_0 dsq_n.
\end{align}
The claimed result follows by combining \eqref{eq:sumSh_2} with \eqref{eq:sumSh_3}.
\end{proof}
Lemma \ref{lem:sumSh} implies
\begin{align}\label{eq:sumSh_sum}
\sum_{k=1}^{6sd}\log (\|\tilde{\theta}_j+\tilde{\theta}_{\ell_k^+}\|_{\T})\geq -240\varepsilon_0 d^2sq_n.
\end{align}
Lemma \ref{lem:sum2+} follows by plugging \eqref{eq:sumSh_sum} into \eqref{eq:sum2+}, whence
\begin{align}
    \sum_{2,+}\geq -300\varepsilon_0 d^2sq_n-6dsq_n \log 2,
\end{align}
 as claimed.
\end{proof}

It remains to estimate $\sum_{2,-}$.
Similar to \eqref{eq:sum2+}, one has
\begin{align}\label{eq:sum2-}
\sum_{2,-}
\geq 
\sum_{k=1}^{6sd}\log (\|\tilde{\theta}_j-\tilde{\theta}_{\ell_k^-}\|_{\T})-6sdq_n\log 2-\varepsilon_0 sdq_n.
\end{align}
We will prove the following.
\begin{lemma}\label{lem:sum2-}
For $n$ large, 
\begin{align}
    \sum_{2,-}\geq -300\varepsilon_0 d^2sq_n-6sdq_n\log 2.
\end{align}
\end{lemma}
\begin{proof}
We have
\begin{lemma}\label{lem:sum-_Sh}
\begin{align}
    \sum_{\ell_k^-\in R_h, \ell_k^-\neq j} \log (\|\tilde{\theta}_j-\tilde{\theta}_{\ell_k^-}\|_{\T})\geq -40\varepsilon_0 sdq_n.
\end{align}
\end{lemma}
\begin{proof}
    We split the argument into two different cases depending on whether or not $j\in R_h$. 

    \underline{Case 1}. If $j\in R_h$, then $\tilde{\ell}_h^-=j$.

For $\ell_k^-\in R_h$ and $\ell_k^-\neq j$, we note that 
$0\leq |\ell_k^--j|\leq sq_n<q_{n+1}$.
Hence by \eqref{eq:qn_alpha_min} and \eqref{eq:n_alpha_varepsilon_1}, 
\begin{align}\label{eq:aaaaa}
    \|\tilde{\theta}_{j}-\tilde{\theta}_{\ell_k^-}\|_{\T}=\|(\ell_k^--j)\alpha\|_{\T}\geq \|q_n \alpha\|_{\T}\geq C_{\varepsilon_0} e^{-\varepsilon_0 q_n}.
\end{align}
This implies that 
\begin{align}\label{eq:sum2-_11}
    \sum_{\ell_k^-\in R_h, \ell_k^-\neq j}\log (\|\tilde{\theta}_j-\tilde{\theta}_{\ell_k^-}\|_{\T})\geq (s-1)\log (C_{\varepsilon_0} e^{-\varepsilon_0 q_n})\geq -40dsq_n.
\end{align}

    \underline{Case 2}. Assume $j\notin R_h$. 

    \underline{Case 2.1}. Assume further that 
    $2\|\tilde{\theta}_j-\tilde{\theta}_{\tilde{\ell}_h^-}\|\leq |\sin(\pi(\tilde{\theta}_j-\tilde{\theta}_{\tilde{\ell}_h^-}))|\leq e^{-2\varepsilon_0 q_n}$.

    For $\ell_k^-\in R_h$ and $\ell_k^-\neq \tilde{\ell}_h^-$, 
    analogously to \eqref{eq:aaaaa} one has 
    \begin{align}
        \|\tilde{\theta}_j-\tilde{\theta}_{\ell_k^-}\|_{\T}
         \geq &\|(\ell_k^--\tilde{\ell}_h^-)\alpha\|_{\T}-e^{-2\varepsilon_0 q_n}\\
        \geq &C_{\varepsilon_0} e^{-\varepsilon_0 q_n}-
        e^{-2\varepsilon_0 q_n}\\
        \geq &e^{-2\varepsilon_0 q_n},
    \end{align}
    for $n$ large.
    Also since $|j-\tilde{\ell}_h^-|\leq y+2sdq_n+1\leq 20dsq_n$, we infer by \eqref{eq:n_alpha_varepsilon_1} that
    \begin{align}
        \|\tilde{\theta}_j-\tilde{\theta}_{\tilde{\ell}_h^-}\|_{\T}\geq \|(j-\tilde{\ell}_h^-)\alpha\|_{\T}\geq C_{\varepsilon_0} e^{-20\varepsilon_0 dsq_n}.
    \end{align}
This implies that 
\begin{align}
    \sum_{\ell_k^-\in R_h} \log(\|\tilde{\theta}_j-\tilde{\theta}_{\ell_k^-}\|_{\T})\geq (s-1)\log(e^{-2\varepsilon_0 q_n})+\log(e^{-20\varepsilon_0 dsq_n})\geq -40\varepsilon_0 dsq_n.
\end{align}    
\underline{Case 2.2}. Suppose $\pi \|\tilde{\theta}_j-\tilde{\theta}_{\tilde{\ell}_h^-}\|\geq |\sin(\pi(\tilde{\theta}_j-\tilde{\theta}_{\tilde{\ell}_h^-}))|>e^{-2\varepsilon_0 q_n}$.
Thus, for any $\ell_k^-\in R_h$,
\[\pi \|\tilde{\theta}_j-\tilde{\theta}_{\ell_k^-}\|\geq |\sin(\pi(\tilde{\theta}_j-\tilde{\theta}_{\ell_k^-}))|\geq e^{-2\varepsilon_0 q_n}.\]
This implies 
\begin{align}
    \sum_{\ell_k^-\in R_h} \|\tilde{\theta}_j-\tilde{\theta}_{\ell_k^-}\|_{\T}\geq s\log(\pi^{-1} e^{-2\varepsilon_0 q_n})>-40dsq_n.
\end{align}
Combining the cases above, we have proved Lemma \ref{lem:sum-_Sh}.
\end{proof}
Lemma \ref{lem:sum2-} follows from plugging Lemma \ref{lem:sum-_Sh} into~\eqref{eq:sum2-}.
\end{proof}
Finally, combining \eqref{eq:sum1_final}, \eqref{eq:sum2_1} with Lemmas \ref{lem:sum2+} and \ref{lem:sum2-}, we obtain
\begin{align}
    \sum_{\ell\in I_1\cup I_2,\ \ell\neq j}\left( \log{|\cos(2\pi \xi)-\cos(2\pi \tilde{\theta}_{\ell})|}-\log{|\cos(2\pi \tilde{\theta}_j)-\cos(2\pi \tilde{\theta}_{\ell})|}\right)\leq 650\varepsilon_0 d^2sq_n.
\end{align}
This establishes Lemma \ref{lem:uniform} with $C_3=110$.
\end{proof}

Lemma \ref{lem:uniform} implies the following. Recall that $\hat{L}^d=\hat{L}(d\alpha,M_E)=d\cdot \hat{L}^d(\alpha,\hat{A}_E)$.
\begin{lemma}\label{lem:I1_I2_exist_big}
If $y>y_1(\alpha,\varepsilon_0,v,E)$, then there exists $j\in I_1\cup I_2$ such that
\begin{align}\label{eq:I1_I2_exist_big}
\frac{1}{dN} \log \left|f_N\left(\theta+j\alpha-
\left\lfloor\frac{dN-1}{2}\right\rfloor\alpha\right)\right|\geq\hat{L}^d(\alpha,\hat{A}_E)+\log |\hat{v}_d|-2C_3d\, \varepsilon_0.
\end{align}
\end{lemma}
\begin{proof}
Assuming the lemma is false, 
let $\tilde{\theta}_j$ be defined as in \eqref{def:ttheta_j}.
If \eqref{eq:I1_I2_exist_big} fails for all $j\in I_1\cup I_2$, then 
\begin{align}\label{assume:I1_I2_small}
&\frac{1}{dN} \log \left|f_{N}\left(\theta+j\alpha-\left\lfloor\frac{dN-1}{2}\right\rfloor\alpha\right)\right|\\
&\qquad\qquad =\frac{1}{dN}\log \left|Q_{dN} (\cos(2\pi \tilde{\theta}_j))\right| 
<\hat{L}^d(\alpha,\hat{A}_E)+\log |\hat{v}_d|-2 C_3d\, \varepsilon_0.
\end{align}
Note that $Q_{dN}(x)$ is a polynomial in $x$ of degree $dN\leq dN+1=|I_1|+|I_2|$.
Hence by the Lagrange interpolation theorem, for any $x\in [0,1]$,
\begin{align}\label{eq:Lag_Q}
Q_{dN}(x)
=\sum_{j\in I_1\cup I_2} Q_{dN}(\cos(2\pi \tilde{\theta}_j))
\cdot \prod_{\substack{m\in I_1\cup I_2\\ m\neq j}} \frac{x-\cos(2\pi \tilde{\theta}_m)}{\cos(2\pi \tilde{\theta}_j)-\cos(2\pi \tilde{\theta}_m)}.
\end{align}
By Lemma \ref{lem:uniform}, we infer that 
\begin{align}\label{eq:uniform_Q}
\max_{x\in [0,1]} \max_{j\in I_1\cup I_2} \left|\prod_{\substack{m\in I_1\cup I_2\\ m\neq j}} \frac{x-\cos(2\pi\tilde{\theta}_m)}{\cos(2\pi\tilde{\theta}_j)-\cos(2\pi\tilde{\theta}_m)}\right|\leq e^{C_3\varepsilon_0 d^2 N}.
\end{align}
Combining \eqref{eq:Lag_Q} with \eqref{assume:I1_I2_small} and \eqref{eq:uniform_Q}, we conclude that 
\begin{align}
\max_{\theta\in\T} |f_{N}(\theta)|=\max_{x\in [0,1]} |Q_{dN}(x)|\leq e^{dN (\hat{L}^d(\alpha,\hat{A}_E)+\log |\hat{v}_d|-C_3 d\, \varepsilon_0)}.
\end{align}
However, combining Lemma \ref{lem:deno_ARC} with \eqref{eq:LM=d_LA2}, there exists $\theta$ such that
\begin{align}
|f_N(\theta)|
\geq |\hat{v}_d|^{dN} e^{dN (\hat{L}^d(\alpha,\hat{A}_E)-\frac{C_3d\,\varepsilon_0}{2})},
\end{align}
which is a contradiction. 
\end{proof}

Lemma \ref{lem:I1_I2_exist_big} implies the following result.
\begin{corollary}\label{cor:I2_exist_big}
If $y>y_1(\alpha,\varepsilon_0,v,E)$, then there exists $j\in I_2$ such that
\begin{align}\label{eq:I2_exist_big}
\frac{1}{dN} \log \left|f_{N}\left(\theta+j\alpha-
\left\lfloor\frac{dN-1}{2}\right\rfloor\alpha\right)\right|\geq \hat{L}^d(\alpha,\hat{A}_E)+\log |\hat{v}_d|-2C_3d\,\varepsilon_0.
\end{align}
\end{corollary}
\begin{proof}
It suffices to show that \eqref{eq:I2_exist_big} does not hold for any $j\in I_1$. Thus, assume that~\eqref{eq:I2_exist_big} does hold for some $j_1\in I_1$.
Let 
\begin{align}
\begin{cases}
x_1:=j_1-\left\lfloor\frac{dN-1}{2}\right\rfloor\\
x_2:=x_1+dN-1
\end{cases}
\end{align}
We estimate the Green's function, 
\begin{align}
G_{N}\left(\theta+x_1\alpha; -x_1, m\right)=\frac{\mu_{N,-x_1, m}(\theta+x_1\alpha)}{f_{N}(\theta+x_1\alpha)},
\end{align}
for $0\leq m\leq d-1$.
By Lemma \ref{lem:numerator}, and using $|\det B|=|\hat{v}_d|^{dN}$, one has 
\begin{align}
&|\mu_{N,-x_1,m}(\theta+x_1\alpha)|\leq C_d\|B^{-1}\|\cdot |\hat{v}_d|^{dN}\cdot\\
&\qquad\cdot \max\left(e^{|x_1|(1+\varepsilon_0)\hat{L}^{d-1}(\alpha,\hat{A}_E)+|x_2|(1+\varepsilon_0)\hat{L}^d(\alpha,\hat{A}_E)}, e^{|x_1|(1+\varepsilon_0)\hat{L}^d(\alpha,\hat{A}_E)+|x_2|(1+\varepsilon_0)\hat{L}^{d-1}(\alpha,\hat{A}_E)}\right)
\end{align}
Combining this with the lower bound in \eqref{eq:I2_exist_big} for $j=j_1$ yields
\begin{align}\label{eq:G_I1_L}
&|G_{N}(\theta+x_1\alpha; -x_1, m)|\\
&\qquad \leq C_d\|B^{-1}\| \max\left(e^{-|x_1|\hat{L}_d(\alpha,\hat{A}_E)}, e^{-|x_2|\hat{L}_d(\alpha,\hat{A}_E)}\right)\cdot e^{\varepsilon_0(\hat{L}^d(\alpha,\hat{A}_E)+2C_3d)dN}.
\end{align}
Similarly, for any $0\leq m\leq d-1$,
\begin{align}\label{eq:G_I1_R}
&|G_{N}(\theta+x_1\alpha; -x_1, dN-1-m)|\\
&\qquad\leq C_d\|B^{-1}\| \max\left(e^{-|x_1|\hat{L}_d(\alpha,\hat{A}_E)}, e^{-|x_2|\hat{L}_d(\alpha,\hat{A}_E)}\right)\cdot e^{\varepsilon_0(\hat{L}^d(\alpha,\hat{A}_E)+2C_3d)dN}.
\end{align}
Using the Poisson formula \eqref{eq:Poisson}, and bounding $|\hat{u}_k|\leq 1+|k|$ and $\|B\|\leq C_v$, $\|B^{-1}\|\leq C_v |\hat{v}_d|^{-1}$,  we infer that 
\begin{align}\label{eq:hat_u_0}
|\hat{u}_0|
\leq &C_{d,v}|\hat{v}_d|^{-2} \max\left(e^{-|x_1|\hat{L}_d(\alpha,\hat{A}_E)}, e^{-|x_2|\hat{L}_d(\alpha,\hat{A}_E)}\right)\cdot\\
&\qquad\qquad\qquad \cdot \max(|x_1|, |x_2|)\cdot e^{\varepsilon_0(\hat{L}^d(\alpha,\hat{A}_E)+2C_3d)dN}
\end{align}
Clearly, 
$$dsq_n\leq \min(|x_1|,|x_2|)\leq 5d sq_n.$$
Hence, we can further bound \eqref{eq:hat_u_0} by
\begin{align}\label{eq:hat_u_0_1}
|\hat{u}_0|&\leq C_{d,v}\, |\hat{v}_d|^{-2} e^{-dsq_n(\hat{L}_d(\alpha,\hat{A}_E)-6\varepsilon_0 \hat{L}^d(\alpha,\hat{A}_E)-20 C_3d\, \varepsilon_0)}.
\end{align}
With the choice of $\varepsilon_0$ in \eqref{def:epsilon_0}, we arrive a contradiction with the assumption that $\hat{u}_0=1$ by  sending $n\to \infty$.
\end{proof}
Corollary \ref{cor:I2_exist_big} implies the following bound: \begin{align}\label{eq:u_y_decay}
|\hat{u}_y|\leq C_{d,v}\, |\hat{v}_d|^{-2} e^{-\frac{y}{8}(\hat{L}_d(\alpha,\hat{A}_E)-6\varepsilon_0\hat{L}^d(\alpha,\hat{A}_E)-20C_3d\, \varepsilon_0)}\leq C_{d,v}|\hat{v}_d|^{-2} e^{-\frac{y}{10}\hat{L}_d(\alpha,\hat{A}_E)},
\end{align}
which we used \eqref{def:epsilon_0} in the last inequality.
The proof is similar to arguments in the proof of Corollary~\ref{cor:I2_exist_big}, after replacing the role of $\hat{u}_0$ with $\hat{u}_y$ and using 
$y\geq 8dsq_n$.
This finishes the proof of almost localization.
\end{proof}

\subsection{Localization}
Theorem \ref{thm:AL_cos} follows from Theorem \ref{thm:almost_AL} by noting that any $\theta$ with $\gamma(\alpha,\theta)=0$ is  $\varepsilon_0$-nonresonant for any $\varepsilon_0>0$.
\begin{remark}
    For $\alpha$ satisfying a stronger Diophantine condition 
    \begin{align}\label{def:DC}
    \alpha\in \mathrm{DC}:=\bigcup_{c>0}\bigcup_{b>1}\left\{\alpha\in \T: \|n\alpha\|_{\T}\geq \frac{c}{|n|^b}\right\}
    \end{align}
    one can give another proof of localization by combining the large deviation theorem with the zero count of $\det(P_{\alpha,n}(\theta)-E)$ as in~\cite{HS2}. This approach will allow us to prove a non-perturbative localization for a.e.~$\alpha\in \mathrm{DC}$ for a general class of operators obtained by replacing the cosine potential in $\hat{H}_{\alpha,\theta}$ with arbitrary analytic potentials.
    We will address this problem in an upcoming work.
\end{remark}

\section{Numerator upper bound for periodic boundary conditions}\label{sec:Numerator}

This section is devoted to the proof of Lemma \ref{lem:numerator}.
We consider the finite volume Hamiltonian with periodic boundary conditions. To distinguish it from $\hat{H}_{\alpha,n}(\theta)$, which satisfies Dirichlet boundary conditions, we denote the periodic Hamiltonian by $P_{\alpha,n}(\theta)$. We shall omit $\alpha$ from the notation for simplicity, and denote $C(\theta+jd\alpha)-E$ by $C_j$.  Thus, 
\begin{align}
P_n(\theta)-E=
\left(\begin{matrix}
C_{n-1} & B^* & &  &B\\
B &C_{n-2} &\ddots \\
& \ddots &\ddots &\ddots \\
& &\ddots &\ddots &B^*\\
B^* & & &B &C_0
\end{matrix}\right).
\end{align}
Let $P_{n}(\theta;x,y)$ be the submatrix of $P_n(\theta)-E$ obtained by deleting the $x-$th row and $y-$th column.
We treat the case $0\leq x\leq d-1$ in detail, and briefly go through the other case at the end of this section.

\underline{Case 1}. Suppose $0\leq x\leq d-1$ and $3d\leq y\leq (n-1)d-1$. 
We write $y=\ell d+r$ for some $\ell\in [3,n-2]$ and $r\in [0,d-1]$.
Let 
\begin{align}
R_{x,y}:=\left(\begin{array}{c|c}
P_n(\theta)-E & {\bf e}_{dn,x}\\
\hline
{\bf e}_{dn,y}^{*} & 0
\end{array}\right),
\end{align}
in which ${\bf e}_{m,j}^*=(\delta_j(m-1),...,\delta_j(1),\delta_j(0))$ and $\delta_{j}(k)$ is the Kronecker delta.
Clearly 
\begin{align}\label{eq:muxy=Rxy}
|\mu_{n,x,y}(\theta)|=|\det(P_n(\theta;x,y)-E)|=|\det R_{x,y}|
\end{align}
Next, we carry out elementary row operations to reduce $R_{x,y}$.
These operations do not alter $|\det R_{x,y}|$.
We have
\begin{align}\label{eq:Rnxy}
\qquad
R_{x,y}
=
&\left(\begin{array}{c|c|c|c|c|c|c|c|c|c}
C_{n-1} & B^* & & &  & & & &B &\\
\hline
B &\ddots &\ddots & & & & & & & \\
\hline
&\ddots &\ddots &\ddots & & & & & &\\
\hline
& &\ddots &\ddots &B^* & & & & &\\
\hline
& & & B &C_{\ell} &B^* & & & &\\
\hline
& & & &B &C_{\ell-1} &\ddots & & &\\
\hline
& & & & &\ddots &\ddots &\ddots & &\\
\hline
& & & & & &\ddots &\ddots &B^* &\\
\hline
B^*& & & & & & &B &C_0 &{\bf e}_{d,x}\\
\hline
& & & &{\bf e}_{d,r}^* & & & & &
\end{array}\right)\\
=:
&
\left(\begin{matrix}
\text{Row}_1\\
\text{Row}_2\\
\vdots\\
\text{Row}_{n+1}
\end{matrix}\right)
\end{align}

In the following we also denote $M_{E,k}(\theta+jd\alpha)$ by $M_k(j)$ for simplicity.
We use the following row operations.

\subsection*{Row operations on \texorpdfstring{$\text{Row}_1$}{Lg}}
Replacing the block-valued $\text{Row}_1$ in 
\eqref{eq:Rnxy}, 
$$ \text{Row}_1 \longrightarrow \text{Row}_1-C_{n-1}\cdot B^{-1}\cdot  \text{Row}_2,$$
yields a new row as follows
$$\text{Row}_1^{(1)}=(0, B^*-C_{n-1} B^{-1}C_{n-2}, -C_{n-1} B^{-1}B^*,0,...,0,B,0).$$
Appealing to \eqref{eq:transfer_inductive_1} and \eqref{eq:transfer_inductive_2}, it is easy to see that
$$\text{Row}_1^{(1)}=(0,-M_2^{UL}(n-2)B,\, -M_{2}^{UR}(n-2),0,...,0,B,0).$$
Replacing 
$$\text{Row}_1^{(1)} \longrightarrow \text{Row}_1^{(1)}+M_{2}^{UL}(n-2) \cdot \text{Row}_3,$$
yields the new row 
$$\text{Row}_1^{(2)}=(0, 0, -M_{2}^{UR}(n-2)+M_{2}^{UL}(n-2) C_{n-3},\, M_{2}^{UL}(n-2)B^*,0,...,0,B,0).$$
Appealing to \eqref{eq:transfer_inductive_1} and \eqref{eq:transfer_inductive_2}, it is easy to see that
$$\text{Row}_1^{(2)}=(0, 0, -M_{3}^{UL}(n-3)B,\, -M_{3}^{UR}(n-3),0,...,0,B,0).$$
One can iterate this process, and after replacing 
$$\text{Row}_1^{(n-3)}\longrightarrow \text{Row}_1^{(n-3)}+M_{n-2}^{UL}(2)\cdot \text{Row}_{n-1},$$ we arrive at
$$\text{Row}_1^{(n-2)}=(0,...,0,-M_{n-1}^{UL}(1)B,\, B-M_{n-1}^{UR}(1),0).$$

\subsection*{Row operations on \texorpdfstring{$\text{Row}_2$}{Lg}}
Replacing 
$$\text{Row}_2\longrightarrow \text{Row}_2-C_{n-2}\cdot B^{-1}\cdot \text{Row}_3,$$
yields the following new  row 
\begin{align}
    \text{Row}_2^{(1)}=
    &(B, 0,  B^*-C_{n-2}B^{-1}C_{n-3}, -C_{n-2}B^{-1}B^*,0,...,0)\\
   =&(B, 0, -M_2^{UL}(n-3)B, -M_2^{UR}(n-3), 0,...,0). 
\end{align}
Repeating this process, and after replacing 
$$\text{Row}_2^{(n-4)}\longrightarrow \text{Row}_2^{(n-4)}+M_{n-3}^{UL}(2)\cdot  \text{Row}_{n-1},$$
we arrive at
\begin{align}
    \text{Row}_2^{(n-3)}=(B,0,...,0,-M_{n-2}^{UL}(1)B,-M_{n-2}^{UR}(1),0).
\end{align}
Next, we perform row reduction on $\text{Row}_{n-\ell+1}$. Recall that $3\leq \ell\leq n-2$.

\subsection*{Row operations on \texorpdfstring{$\text{Row}_{n-\ell+1}$ (the row containing $C_{\ell-1}$)}{Lg}.}

Replacing 
$$\text{Row}_{n-\ell+1}\longrightarrow \text{Row}_{n-\ell+1}-C_{\ell-1} B^{-1}\cdot \text{Row}_{n-\ell+2},$$
yields a new row as follows
\begin{align}
\text{Row}_{n-\ell+1}^{(1)}
=&(0,...,0,B,0,B^*-C_{\ell-1} B^{-1}C_{\ell-2}, -C_{\ell-1}B^{-1}B^*,0,...,0)\\
=&(0,...,0,B,0,-M_{2}^{UL}(\ell-2)B,\, -M_{2}^{UR}(\ell-2),0,...,0).
\end{align}
where the matrix $B$ is the $(n-\ell)-$th entry (from the left) of the block-valued vector above, and we have used \eqref{eq:transfer_inductive_1} and \eqref{eq:transfer_inductive_2} to obtain the second line.
Repeating this process, and after replacing
$$\text{Row}_{n-\ell+1}^{(\ell-3)}\longrightarrow \text{Row}_{n-\ell+1}^{(\ell-3)}+M_{\ell-2}^{UL}(2) \cdot \text{Row}_{n-1},$$ 
we arrive at
$$\text{Row}_{n-\ell+1}^{(\ell-2)}=(0,...,0,B,0,...,0,-M_{\ell-1}^{UL}(1)B,-M_{\ell-1}^{UR}(1),0).$$
At this point, $R_{x,y}$ becomes
\begin{align}\label{eq:R2xy_big_matrix}
R^{(1)}_{x,y}
=
\left(\begin{array}{c|c|c|c|c|c|c|c|c|c|c|c|c}
0 & 0 &0 &\cdots & &  & & &\cdots &0 &-M_{n-1}^{UL}(1)B &B-M_{n-1}^{UR}(1) & 0 \\
\hline
B &0 &0 &\cdots & & & & &\cdots &0 &-M_{n-2}^{UL}(1)B &-M_{n-2}^{UR}(1) &0 \\
\hline 
0&B &C_{n-3} &B^* & & & & & & & & &\\
\hline
\vdots& &\ddots &\ddots &\ddots & & & & & & & &\\
\hline
& & & \ddots &\ddots &\ddots & & & & & & &\\
\hline
\vdots & & & &B &C_{\ell} &B^* & & & & & &\\
\hline
0&\cdots & &\cdots &0 &B &0 &0 &\cdots &0 &-M_{\ell-1}^{UL}(1)B &-M_{\ell-1}^{UR}(1) &0\\
\hline
\vdots & & & & &0 &B &C_{\ell-2} &B^* & & &0 &0\\
\hline
& & & & & & &\ddots &\ddots &\ddots & &\vdots &\vdots\\
\hline
\vdots& & & & & & & &\ddots &\ddots &\ddots &0 &\vdots\\
\hline
0& & & & & & & & &\ddots &\ddots &B^* &0\\
\hline
B^* &0 &\cdots & & & & & &\cdots &0 &B &C_0 &{\bf e}_{d,x}\\
\hline
0 &\cdots & &\cdots & 0 &{\bf e}_{d,r}^* & 0 &\cdots & &\cdots &  0 & 0 & 0 
\end{array}\right)
\end{align}
The first, second, $(n-\ell+1)$-th, $n$-th, and $(n+1)$-th rows are
\begin{align}
\left(\begin{matrix}
\text{Row}_1^{(n-2)}\\
\text{Row}_2^{(n-3)}\\
\text{Row}_{n-\ell+1}^{(\ell-2)}\\
\text{Row}_{n}\\
\text{Row}_{n+1}
\end{matrix}
\right)
=\left(\begin{array}{ccccccccccc}
0 &0 &\cdots &0 &0 &0 &\cdots &0 &-M_{n-1}^{UL}(1)B &B-M_{n-1}^{UR}(1) &0\\
B &0 &\cdots &0 &0 &0 &\cdots &0 &-M_{n-2}^{UL}(1)B &-M_{n-2}^{UR}(1) &0\\
0 &0 &\cdots &0 &B &0 &\cdots &0 &-M_{\ell-1}^{UL}(1)B &-M_{\ell-1}^{UR}(1) &0\\
B^* &0 &\cdots &0 &0 &0 &\cdots &0 &B &C_0 &{\bf e}_{d,x}\\
0 &0 &\cdots &0 &{\bf e}_{d,r}^* &0 &\cdots &0 &0 &0 &0
\end{array}\right),
\end{align}
in which only columns $1$, $n-\ell$, $n-1$, $n$, $n+1$ are non-vanishing.
Let
\begin{align}
    S_1:=\left(\begin{array}{ccccc}
0 &0 &-M_{n-1}^{UL}(1)B &B-M_{n-1}^{UR}(1) &0\\
B &0 &-M_{n-2}^{UL}(1)B &-M_{n-2}^{UR}(1) &0\\
0 &B &-M_{\ell-1}^{UL}(1)B &-M_{\ell-1}^{UR}(1) &0\\
B^* &0 &B &C_0 &{\bf e}_{d,x}\\
0 &{\bf e}^*_{d,r} & 0 &0 &0
    \end{array}\right)
\end{align}
be the non-vanishing $(4d+1)\times (4d+1)$ submatrix of rows $1$, $2$, $n-\ell+1$, $n$ and $n+1$.
Let
\begin{align}
    S_2:=
    \left(\begin{array}{cccccccccccccccc}
       B &C_{n-3} &B^* & & & & & & & & & & & & &\\
          &B       &C_{n-4} &\ddots & & & & & & & & & & & &\\
          &        &\ddots  &\ddots & & & & & & & & & & & &\\
                   \\
          & & & & &\ddots &\ddots  & & & & & & & & &\\
          & & & & &\ddots &C_{\ell+2} &B^* & & & & & & & &\\
          & & & & &  &B &C_{\ell+1} &0 & & & & & & &\\
          & & & & &  & &B &B^* &0 & & & & & &\\
          & & & & &  & & &B &C_{\ell-2} &B^* & & & & &\\
          & & & & &  & & & &B &\ddots &\ddots & & & &\\
          & & & & &  & & & & &\ddots & & & & &\\
          \\
          & & & & &  & & & & & & &\ddots &\ddots &\ddots &\\
          & & & & &  & & & & & & & &B &C_3 &B^*\\
          & & & & &  & & & & & & & & &B &C_2   \\ 
          & & & & &  & & & & & & & & & &B          
    \end{array}\right),
\end{align}
which is the submatrix of $R^{(1)}_{x,y}$ after deleting rows $1$, $2$, $n-\ell+1$, $n$, $n+1$ and columns $1$, $n-\ell$, $n-1$, $n$, $n+1$.
Hence 
\begin{align}\label{eq:num1}
|\det R^{(1)}_{x,y}|=
|\det S_1|\cdot |\det S_2|=|\det B|^{n-4}\cdot |\det S_1|,
\end{align}
where we used $\det S_2=(\det B)^{n-4}$. 
It suffices to compute $\det S_1$.
Expanding determinants and using that $M_{k-2}^{UL}(1)=BM_{k-1}^{LL}(1)$, $M_{k-2}^{UR}(1)=BM_{k-1}^{LR}(1)$, we have
\begin{align}\notag
|\det S_1|=
&\left|\det\left(\begin{array}{ccccc}
0 &0 &-M_{n-1}^{UL}(1)B &B-M_{n-1}^{UR}(1) &0\\
B &0 &-BM_{n-1}^{LL}(1)B &-BM_{n-1}^{LR}(1) &0\\
0 &B &-M_{\ell-1}^{UL}(1)B &-M_{\ell-1}^{UR}(1) &0\\
B^* &0 &B &C_0 &{\bf e}_{d,x}\\
0 &{\bf e}^*_{d,r} & 0 &0 &0
    \end{array}\right)\right| \notag\\
=
&|\det B|\cdot \left|\det\left(\begin{array}{ccccc}
0 &0 &-M_{n-1}^{UL}(1) &B-M_{n-1}^{UR}(1) &0\\
B &0 &-BM_{n-1}^{LL}(1) &-BM_{n-1}^{LR}(1) &0\\
0 &B &-M_{\ell-1}^{UL}(1) &-M_{\ell-1}^{UR}(1) &0\\
B^* &0 &I_d &C_0 &{\bf e}_{d,x}\\
0 &{\bf e}^*_{d,r} & 0 &0 &0
    \end{array}\right)\right|\notag
    \end{align}
    Expanding along the first column this simplifies further to
    \begin{align}\label{eq:num2}
|\det S_1|
=&|\det B|^2\cdot |\det B^*| \cdot
\left|\det\left(\begin{array}{ccccc}
0 &0 &-M_{n-1}^{UL}(1) &B-M_{n-1}^{UR}(1) &0\\
I_d &0 &-M_{n-1}^{LL}(1) &-M_{n-1}^{LR}(1) &0\\
0 &B &-M_{\ell-1}^{UL}(1) &-M_{\ell-1}^{UR}(1) &0\\
I_d &0 &(B^*)^{-1} &(B^*)^{-1} C_0 &(B^*)^{-1} {\bf e}_{d,x}\\
0 &{\bf e}^*_{d,r} & 0 &0 &0
    \end{array}\right)\right|\notag\\
=&|\det B|^3 \cdot 
\left|\det\left(\begin{array}{cccccc}
0& 0 &0 &-M_{n-1}^{UL}(1) &B-M_{n-1}^{UR}(1) &0\\
0& I_d &0 &-M_{n-1}^{LL}(1) &-M_{n-1}^{LR}(1) &0\\
0 &0 &B &-M_{\ell-1}^{UL}(1) &-M_{\ell-1}^{UR}(1) &0\\
I_d &0 &0 &0 &-B &0\\
0& I_d &0 &(B^*)^{-1} &(B^*)^{-1} C_0 &(B^*)^{-1} {\bf e}_{d,x}\\
0& 0 &{\bf e}^*_{d,r} & 0 &0 &0
    \end{array}\right)\right|\notag\\
=&|\det B|^3\cdot 
\left|\det\left(\begin{array}{c|c|c|c}
\left(\begin{matrix}0& 0\\
0 &I_d\end{matrix}\right)& &\left(\begin{matrix}0 & B\\ 0 & 0\end{matrix}\right)-M_{n-1}(1) &\\
\hline
 & B &-M_{\ell-1}^U(1) &  \\
 \hline
I_{2d} & &-M_1^{-1}(0) &\left(\begin{matrix}0\\ I_d\end{matrix}\right) (B^*)^{-1}{\bf e}_{d,x}\\
\hline
 &{\bf e}^*_{d,r} & &
    \end{array}\right)\right|\notag\\
=:&|\det B|^3\cdot |\det S_3|,
\end{align}
in which $M_{\ell-1}^U(1)=(M_{\ell-1}^{UL}(1), M_{\ell-1}^{UR}(1))=(I_d, 0)\cdot M_{\ell-1}(1)$. Empty spaces are to be filled in by zeros.  
We carry out the following row operations on $S_3$ that do not alter the determinant of $S_3$. Below $\text{row}_j$ refers to the $j$-th (block-valued) row of $S_3$. Replacing 
$$\text{row}_1\longrightarrow \text{row}_1-\left(\begin{matrix} 0 &0\\ 0 &I_d\end{matrix}\right) \cdot \text{row}_3,$$
and
$$\text{row}_4\longrightarrow \text{row}_4-{\bf e}_{d,r}^* B^{-1}\cdot \text{row}_2,$$
we arrive at
\begin{align}
S_3^{(1)}=\left(\begin{array}{c|c|c|c}
& &\left(\begin{matrix}0 & B\\ 0 & 0\end{matrix}\right)-M_{n-1}(1)+\left(\begin{matrix}0 &0 \\ 0 &I_d\end{matrix}\right)M_1^{-1}(0) &-\left(\begin{matrix}0\\ I_d\end{matrix}\right) (B^*)^{-1}{\bf e}_{d,x}\\
\hline
 & B &-(I_d, 0) \cdot M_{\ell-1}(1) &  \\
 \hline
I_{2d} & &-M_1^{-1}(0) &\left(\begin{matrix}0\\ I_d\end{matrix}\right) (B^*)^{-1}{\bf e}_{d,x}\\
\hline
 & &{\bf e}_{d,r}^*B^{-1}(I_d, 0)\cdot M_{\ell-1}(1) &
    \end{array}\right).
\end{align}
Next, multiplying $\text{col}_3$ of $S_3^{(1)}$ on the right by $M_1(0)$ (for which $\det M_1(0)=1$, so $\det S_3^{(1)}$ is invariant),  yields
\begin{align}
S_3^{(2)}=\left(\begin{array}{c|c|c|c}
& &\left(\begin{matrix}0 & B\\ 0 & 0\end{matrix}\right)M_1(0)-M_{n}(0)+\left(\begin{matrix}0 &0 \\ 0 &I_d\end{matrix}\right)&-\left(\begin{matrix}0\\ I_d\end{matrix}\right) (B^*)^{-1}{\bf e}_{d,x}\\
\hline
 & B &-(I_d, 0)\cdot M_{\ell}(0) &  \\
 \hline
I_{2d} & &-I_d &\left(\begin{matrix}0\\ I_d\end{matrix}\right) (B^*)^{-1}{\bf e}_{d,x}\\
\hline
 & &{\bf e}_{d,r}^*B^{-1}(I_d, 0)\cdot M_{\ell}(0) &
    \end{array}\right).
\end{align}
Expanding the determinant along the first two columns, and noting that $\left(\begin{matrix} 0 &B\\ 0 &0\end{matrix}\right) M_1(0)=\left(\begin{matrix}I_d &0\\ 0 &0\end{matrix}\right)$, yields
\begin{align}\label{eq:num3}
|\det S_3^{(2)}|&=|\det B|\cdot  
\left|\det\left(\begin{array}{c|c}
    -M_n(0)+I_{2d} &-\left(\begin{matrix}0\\ I_d\end{matrix}\right) (B^*)^{-1}{\bf e}_{d,x}\\
    \hline
    {\bf e}_{d,r}^*B^{-1}(I_d, 0)\cdot M_{\ell}(0) &0
\end{array}\right)\right| \notag\\
&=:|\det B|\cdot |\det S_3^{(3)}|.
\end{align}
We expand $S_3^{(3)}$ as the following product: 
\begin{align}
    S_3^{(3)}
    =&\left(\begin{array}{c|c} I_{2d} & \\ \hline &{\bf e}_{d,r}^*B^{-1}(I_d, 0)\end{array}\right)_{(2d+1)\times 4d}
    \cdot 
    \left(\begin{array}{c|c}-M_{n-\ell}(\ell) & I_{2d}\\ \hline I_{2d} & \end{array}\right)
    \cdot 
    \left(\begin{array}{c|c}M_{\ell}(0) & \\ \hline &I_{2d}\end{array}\right)
    \cdot\\
    &\qquad\qquad\qquad\qquad\qquad\qquad\qquad\qquad
    \cdot 
    \left(\begin{array}{c|c} I_{2d} & \\ \hline I_{2d} &-\left(\begin{matrix} 0\\ I_d\end{matrix}\right)(B^*)^{-1}{\bf e}_{d,x}\end{array}\right)_{4d\times (2d+1)}
\end{align}
We bound
\begin{align}\label{eq:num4}
    |\det S_3^{(3)}|
    =&\left|\left(\wedge^{2d+1}\left(\begin{array}{c|c} I_{2d} & \\ \hline &{\bf e}_{d,r}^*B^{-1}(I_d, 0)\end{array}\right)\right)
    \cdot
    \left(\wedge^{2d+1} \left(\begin{array}{c|c}-M_{n-\ell}(\ell) & I_{2d}\\ \hline I_{2d} & \end{array}\right)\right)
    \cdot \right. \\
    &\left. \qquad \cdot
    \left(\wedge^{2d+1} \left(\begin{array}{c|c}M_{\ell}(0) & \\ \hline &I_{2d}\end{array}\right)\right)
    \cdot
    \left(\wedge^{2d+1}
    \left(\begin{array}{c|c} I_{2d} & \\ \hline I_{2d} &-\left(\begin{matrix} 0\\ I_d\end{matrix}\right)(B^*)^{-1}{\bf e}_{d,x}\end{array}\right)\right)\right|\\
    \leq &\sum_{y=2d+1}^{3d} \sum_{\substack{1\leq j_1<...<j_{2d+1}\leq 4d\\ 1\leq k_1<...<k_{2d+1}\leq 4d\\ 1\leq \ell_1<...<\ell_{2d+1}\leq 4d}} \left|\det\left(\begin{array}{c|c} I_{2d} & \\ \hline &{\bf e}_{d,r}^*B^{-1}(I_d, 0)\end{array}\right)_{(1,...,2d,2d+1),(1,...,2d,y)}\right|\\
    &\qquad\qquad\qquad\cdot 
    \det \left|\left(\begin{array}{c|c}-M_{n-\ell}(\ell) & I_{2d}\\ \hline I_{2d} & \end{array}\right)_{(1,...,2d,y),(j_1,...,j_{2d+1})}\right|\\
    &\qquad\qquad\qquad\cdot 
    \left|\det \left(\begin{array}{c|c}M_{\ell}(0) & \\ \hline &I_{2d}\end{array}\right)_{(j_1,...,j_{2d+1}),(k_1,...,k_{2d+1})}\right|\\
    &\qquad\qquad\qquad \cdot \left|\det \left(\begin{array}{c|c} I_{2d} & \\ \hline I_{2d} &-\left(\begin{matrix} 0\\ I_d\end{matrix}\right)(B^*)^{-1}{\bf e}_{d,x}\end{array}\right)_{(k_1,...,k_{2d+1}),(\ell_1,...,\ell_{2d+1})}\right|.
\end{align}
Suppose that
$$k_1<...<k_{m_0}\leq 2d<k_{m_0+1}<...<k_{2d+1}.$$
In order for 
\begin{align}
    \det \left(\begin{array}{c|c}M_{\ell}(0) & \\  \hline &I_{2d}\end{array}\right)_{(j_1,...,j_{2d+1}),(k_1,...,k_{2d+1})}, 
\end{align}
to be non-vanishing, necessarily  
\begin{align}\label{eq:jm0}
j_1<...<j_{m_0}\leq 2d \text{\ \ and\ \ } j_m=k_m \text{ for } m_0+1\leq m\leq 2d+1,
\end{align}
If it is non-vanishing, then 
\begin{align}\label{eq:numerator_1}
        &\left|\det \left(\begin{array}{c|c}M_{\ell}(0) & \\ \hline &I_{2d}\end{array}\right)_{(j_1,...,j_{2d+1}),(k_1,...,k_{2d+1})}\right|\\
        =&\left|\det(M_{\ell}(0))_{\{j_1,...,j_{m_0}\},\{k_1,...,k_{m_0}\}}\right|\\
        \leq &\|\wedge^{m_0} M_{\ell}(0)\|\leq e^{\ell(1+\varepsilon) \cdot \hat{L}^{m_0}},
\end{align}
where we used Lemma \ref{lem:upper_semi_cont} to bound $\|\wedge^{m_0}M_{\ell}(0)\|$. 
Taking \eqref{eq:jm0} into account, one has
\begin{align}\label{eq:numerator_2}
    &\left|\det \left(\begin{array}{c|c}-M_{n-\ell}(\ell) & I_{2d}\\ \hline I_{2d} & \end{array}\right)_{(1,...,2d,y),(j_1,...,j_{2d+1})}\right|\\
    =&
    \left|\det (M_{n-\ell}(\ell))_{\{1,...,2d\}\setminus \{j_{m_0+1}-2d,...,j_{2d+1}-2d\}, \{j_1,...,j_{m_0}\}\setminus\{y-2d\}}\right|\\
    \leq &\|\wedge^{m_0-1} M_{n-\ell}(\ell)\|\leq e^{(n-\ell)(1+\varepsilon)\cdot \hat{L}^{m_0-1}},
\end{align}
where we used Lemma \ref{lem:upper_semi_cont} again.
The following estimates are easy to verify:
\begin{align}\label{eq:numerator_3}
\left|\det\left(\begin{array}{c|c} I_{2d} & \\ \hline &{\bf e}_{d,r}^*B^{-1}(I_d, 0)\end{array}\right)_{(1,...,2d,2d+1),(1,...,2d,y)}\right|\leq C\|B^{-1}\|, \text{ and }\\
\left|\det \left(\begin{array}{c|c} I_{2d} & \\ \hline I_{2d} &-\left(\begin{matrix} 0\\ I_d\end{matrix}\right)(B^*)^{-1}{\bf e}_{d,x}\end{array}\right)_{(k_1,...,k_{2d+1}),(\ell_1,...,\ell_{2d+1})}\right|\leq C\|B^{-1}\|.
\end{align}
Combining \eqref{eq:numerator_1}, \eqref{eq:numerator_2} with \eqref{eq:numerator_3}, one has 
\begin{align}
    |\det S_3^{(3)}|
    \leq &C_d \|B^{-1}\|^2 \sum_{m_0} e^{\ell (1+\varepsilon)\cdot\hat{L}^{m_0}+(n-\ell)(1+\varepsilon)\cdot \hat{L}^{m_0-1}}\\
    \leq &C_d \|B^{-1}\|^2 \max \left(e^{\ell (1+\varepsilon)\cdot \hat{L}^{d}+(n-\ell)(1+\varepsilon)\cdot \hat{L}^{d-1}}, e^{\ell(1+\varepsilon)\cdot \hat{L}^{d+1}+(n-\ell)(1+\varepsilon)\cdot \hat{L}^d}\right),
\end{align}
where we used that $\hat{L}^1<...<\hat{L}^{d-1}<\hat{L}^d$ and $\hat{L}^d>\hat{L}^{d+1}>...>\hat{L}^{2d}$, due to $\hat{L}_d>0>\hat{L}_{d+1}$.
Combining the above with \eqref{eq:num1},\eqref{eq:num2} and \eqref{eq:num3} yields the claimed result for Case~1.

\underline{Case 2}. Suppose $(n-1)d\leq x\leq nd-1$. 
Again, we write $y=\ell d+r$ with $3\leq \ell\leq n-2$ and $0\leq r\leq d-1$. We also denote $x-(n-1)d=\tilde{x}$.
In analogy  to \eqref{eq:Rnxy}, we have $|\mu_{n,x,y}(\theta)|=|\det \tilde{R}_{x,y}|$, where
\begin{align} 
\qquad
\tilde{R}_{x,y}
=
&\left(\begin{array}{c|c|c|c|c|c|c|c|c|c}
C_{n-1} & B^* & & &  & & & &B &{\bf e}_{d,\tilde{x}}\\
\hline
B &\ddots &\ddots & & & & & & & \\
\hline
&\ddots &\ddots &\ddots & & & & & &\\
\hline
& &\ddots &\ddots &B^* & & & & &\\
\hline
& & & B &C_{\ell} &B^* & & & &\\
\hline
& & & &B &C_{\ell-1} &\ddots & & &\\
\hline
& & & & &\ddots &\ddots &\ddots & &\\
\hline
& & & & & &\ddots &\ddots &B^* &\\
\hline
B^*& & & & & & &B &C_0 &\\
\hline
& & & &{\bf e}_{d,r}^* & & & & &
\end{array}\right)
\end{align}
Applying exactly the same row operations on $\text{Row}_1$, $\text{Row}_2$ and $\text{Row}_{n-\ell+1}$ as in Case 1, reduces $\tilde{R}_{x,y}$ to the following shape:
\begin{align}\label{eq:tR2xy_big_matrix}
\tilde{R}^{(1)}_{x,y}
=
\left(\begin{array}{c|c|c|c|c|c|c|c|c|c|c|c|c}
0 & 0 &0 &\cdots & &  & & &\cdots &0 &-M_{n-1}^{UL}(1)B &B-M_{n-1}^{UR}(1) & {\bf e}_{d,\tilde{x}} \\
\hline
B &0 &0 &\cdots & & & & &\cdots &0 &-M_{n-2}^{UL}(1)B &-M_{n-2}^{UR}(1) &0 \\
\hline 
0&B &C_{n-3} &B^* & & & & & & & & &\vdots\\
\hline
\vdots& &\ddots &\ddots &\ddots & & & & & & & &\\
\hline
& & & \ddots &\ddots &\ddots & & & & & & &\\
\hline
\vdots & & & &B &C_{\ell} &B^* & & & & & &\vdots\\
\hline
0&\cdots & &\cdots &0 &B &0 &0 &\cdots &0 &-M_{\ell-1}^{UL}(1)B &-M_{\ell-1}^{UR}(1) &0\\
\hline
\vdots & & & & &0 &B &C_{\ell-2} &B^* & & &0 &0\\
\hline
& & & & & & &\ddots &\ddots &\ddots & &\vdots &\vdots\\
\hline
\vdots& & & & & & & &\ddots &\ddots &\ddots &0 &\\
\hline
0& & & & & & & & &\ddots &\ddots &B^* &\vdots\\
\hline
B^* &0 &\cdots & & & & & &\cdots &0 &B &C_0 &0\\
\hline
0 &\cdots & &\cdots & 0 &{\bf e}_{d,r}^* & 0 &\cdots & &\cdots &  0 & 0 & 0 
\end{array}\right)
\end{align}
This implies, similarly to \eqref{eq:num1}, that
\begin{align}
    |\det \tilde{R}_{x,y}|=|\det \tilde{R}^{(1)}_{x,y}|=|\det B|^{n-4}\cdot |\det \tilde{S}_1|,
\end{align}
where in analogy with \eqref{eq:num2} and \eqref{eq:num3} one has 
\begin{align}\notag 
|\det \tilde{S}_1|=
&\left|\det\left(\begin{array}{ccccc}
0 &0 &-M_{n-1}^{UL}(1)B &B-M_{n-1}^{UR}(1) &{\bf e}_{d,\tilde{x}}\\
B &0 &-BM_{n-1}^{LL}(1)B &-BM_{n-1}^{LR}(1) &0\\
0 &B &-M_{\ell-1}^{UL}(1)B &-M_{\ell-1}^{UR}(1) &0\\
B^* &0 &B &C_0 &0\\
0 &{\bf e}^*_{d,r} & 0 &0 &0
    \end{array}\right)\right| \notag\\
=
&|\det B|\cdot \left|\det\left(\begin{array}{ccccc}
0 &0 &-M_{n-1}^{UL}(1) &B-M_{n-1}^{UR}(1) &{\bf e}_{d,\tilde{x}}\\
B &0 &-BM_{n-1}^{LL}(1) &-BM_{n-1}^{LR}(1) &0\\
0 &B &-M_{\ell-1}^{UL}(1) &-M_{\ell-1}^{UR}(1) &0\\
B^* &0 &I_d &C_0 &0\\
0 &{\bf e}^*_{d,r} & 0 &0 &0
    \end{array}\right)\right|\notag\\
=&|\det B|^2\cdot |\det B^*| \cdot
\left|\det\left(\begin{array}{ccccc}
0 &0 &-M_{n-1}^{UL}(1) &B-M_{n-1}^{UR}(1) &{\bf e}_{d,\tilde{x}}\\
I_d &0 &-M_{n-1}^{LL}(1) &-M_{n-1}^{LR}(1) &0\\
0 &B &-M_{\ell-1}^{UL}(1) &-M_{\ell-1}^{UR}(1) &0\\
I_d &0 &(B^*)^{-1} &(B^*)^{-1} C_0 &0\\
0 &{\bf e}^*_{d,r} & 0 &0 &0
    \end{array}\right)\right|\notag
    \end{align}
    Rewriting this in block form yields 
   \begin{align}\label{eq:num2*}
|\det \tilde{S}_1| 
=&|\det B|^3\cdot 
\left|\det\left(\begin{array}{c|c|c|c}
\left(\begin{matrix}0& 0\\
0 &I_d\end{matrix}\right)& &\left(\begin{matrix}0 & B\\ 0 & 0\end{matrix}\right)-M_{n-1}(1) & \left(\begin{matrix} I_d\\ 0\end{matrix}\right) {\bf e}_{d,\tilde{x}}\\
\hline
 & B &-M_{\ell-1}^U(1) &  \\
 \hline
I_{2d} & &-M_1^{-1}(0) &\\
\hline
 &{\bf e}^*_{d,r} & &
    \end{array}\right)\right|\notag\\
=&|\det B|^3\cdot 
\left|\det\left(\begin{array}{c|c|c|c}
\left(\begin{matrix}0& 0\\
0 &I_d\end{matrix}\right)& &\left(\begin{matrix}0 & B\\ 0 & 0\end{matrix}\right)M_1(0)-M_{n}(0) & \left(\begin{matrix} I_d\\ 0\end{matrix}\right) {\bf e}_{d,\tilde{x}}\\
\hline
 & B &-(I_d, 0)\, M_{\ell}(0) &  \\
 \hline
I_{2d} & &-I_{2d} &\\
\hline
 &{\bf e}^*_{d,r} & &
    \end{array}\right)\right|\notag\\
=&|\det B|^3\cdot 
\left|\det\left(\begin{array}{c|c|c|c}
& &I_{2d}-M_{n}(0) & \left(\begin{matrix} I_d\\ 0\end{matrix}\right) {\bf e}_{d,\tilde{x}}\\
\hline
 & B &-(I_d, 0)\, M_{\ell}(0) &  \\
 \hline
I_{2d} & &-I_{2d} &\\
\hline
 &{\bf e}^*_{d,r} & &
    \end{array}\right)\right|\notag\\
=&|\det B|^3\cdot 
\left|\det\left(\begin{array}{c|c|c|c}
& &I_{2d}-M_{n}(0) & \left(\begin{matrix} I_d\\ 0\end{matrix}\right) {\bf e}_{d,\tilde{x}}\\
\hline
 & B &-(I_d, 0)\, M_{\ell}(0) &  \\
 \hline
I_{2d} & &-I_{2d} &\\
\hline
 & &{\bf e}_{d,r}^*B^{-1} (I_d, 0)\, M_{\ell}(0) &
    \end{array}\right)\right|\notag\\
=&|\det B|^3\cdot 
\left|\det\left(\begin{array}{c|c|c|c}
& &I_{2d}-M_{n}(0) & \left(\begin{matrix} I_d\\ 0\end{matrix}\right) {\bf e}_{d,\tilde{x}}\\
\hline
 & B &-(I_d, 0)\, M_{\ell}(0) &  \\
 \hline
I_{2d} & &-I_{2d} &\\
\hline
 &{\bf e}^*_{d,r} & &
    \end{array}\right)\right|\notag\\
=&|\det B|^3\cdot 
\left|\det\left(\begin{array}{c|c}
I_{2d}-M_{n}(0) & \left(\begin{matrix} I_d\\ 0\end{matrix}\right) {\bf e}_{d,\tilde{x}}\\
\hline
{\bf e}_{d,r}^*B^{-1} (I_d, 0)\, M_{\ell}(0) &
    \end{array}\right)\right|\notag\\
=:&|\det B|^4\cdot |\det \tilde{S}_{(3)}|,
\end{align}
where
\begin{align}
    \tilde{S}_3^{(3)}
    =&\left(\begin{array}{c|c} I_{2d} & \\ \hline &{\bf e}_{d,r}^*B^{-1}(I_d, 0)\end{array}\right)_{(2d+1)\times 4d}
    \cdot 
    \left(\begin{array}{c|c}-M_{n-\ell}(\ell) & I_{2d}\\ \hline I_{2d} & \end{array}\right)
    \cdot 
    \left(\begin{array}{c|c}M_{\ell}(0) & \\ \hline &I_{2d}\end{array}\right)
    \cdot\\
    &\qquad\qquad\qquad\qquad\qquad\qquad\qquad\qquad
    \cdot 
    \left(\begin{array}{c|c} I_{2d} & \\ \hline I_{2d} &-\left(\begin{matrix} I_d\\ 0\end{matrix}\right){\bf e}_{d,\tilde{x}}\end{array}\right)_{4d\times (2d+1)}
\end{align}
The remainder of the argument is exactly the same as in Case 1. \qed

\section{Denominator lower bound for periodic boundary condition}\label{sec:Denominator}
In this section we prove Lemma \ref{lem:deno}.
Throughout, we denote $B_0$ by $B$, and $C_0(\theta+jd\alpha)-E$ by $C_j$.
First, we relate the determinant $f_{E,n}(\theta)=\det(P_{\alpha,n}(\theta)-E)$ to the transfer matrix $M_{n,E}$. 

\begin{lemma}\label{lem:P=M-I}
The following identity holds pointwise in $\theta$
\begin{align}
|f_{E,n}(\theta)|=|\det B|^n \cdot |\det (M_{n,E}(\theta)-I_{2d})|.
\end{align}
\end{lemma}
\begin{proof}
The proof is similar to, and is in fact  simpler than, the one of Lemma \ref{lem:numerator}.    
Indeed, following exactly the same row operations (that do not alter $f_{E,n}$) applied to $\text{Row}_1$, $\text{Row}_2$ and $\text{Row}_{n-\ell+1}$ of $P_{\alpha,n}(\theta)-E$, instead of those on $R_{x,y}$ in Section~\ref{sec:Numerator} (note these operations do not involve the terms ${\bf e}_{d,x}$ or ${\bf e}_{d,r}^*$),  in analogy  to \eqref{eq:R2xy_big_matrix} we arrive at the equality 
\begin{align}
P^{(1)}_n(\theta)
=
\left(\begin{array}{c|c|c|c|c|c|c|c|c|c|c|c}
0 & 0 &0 &\cdots & &  & & &\cdots &0 &-M_{n-1}^{UL}(1)B &B-M_{n-1}^{UR}(1) \\
\hline
B &0 &0 &\cdots & & & & &\cdots &0 &-M_{n-2}^{UL}(1)B &-M_{n-2}^{UR}(1) \\
\hline 
0&B &C_{n-3} &B^* & & & & & & & & \\
\hline
\vdots& &\ddots &\ddots &\ddots & & & & & & & \\
\hline
& & & \ddots &\ddots &\ddots & & & & & & \\
\hline
\vdots & & & &B &C_{\ell} &B^* & & & & & \\
\hline
0&\cdots & &\cdots &0 &B &0 &0 &\cdots &0 &-M_{\ell-1}^{UL}(1)B &-M_{\ell-1}^{UR}(1) \\
\hline
\vdots & & & & &0 &B &C_{\ell-2} &B^* & & &0 \\
\hline
& & & & & & &\ddots &\ddots &\ddots & &\vdots \\
\hline
\vdots& & & & & & & &\ddots &\ddots &\ddots &0\\
\hline
0& & & & & & & & &\ddots &\ddots &B^* \\
\hline
B^* &0 &\cdots & & & & & &\cdots &0 &B &C_0
\end{array}\right)
\end{align}
This implies, similarly to \eqref{eq:num1}, that
\begin{align}\label{eq:fn=tildeS1}
    |f_{E,n}(\theta)|=|\det P_n^(1)(\theta)|=|\det B|^{n-4} \cdot |\det \tilde{S}_1|,
\end{align}
where
\begin{align}\label{eq:tildeS1}
    \tilde{S}_1:=\left(\begin{array}{cccc}
0 &0 &-M_{n-1}^{UL}(1)B &B-M_{n-1}^{UR}(1) \\
B &0 &-M_{n-2}^{UL}(1)B &-M_{n-2}^{UR}(1) \\
0 &B &-M_{\ell-1}^{UL}(1)B &-M_{\ell-1}^{UR}(1) \\
B^* &0 &B &C_0
    \end{array}\right).
\end{align}
Furthermore,  in place of \eqref{eq:num2}, we now have
\begin{align}\notag 
    |\det \tilde{S}_1|
    =&|\det B|^2 \cdot 
    \left|\det \left(\begin{array}{cccc}
0 &-M_{n-1}^{UL}(1) &B-M_{n-1}^{UR}(1) \\
B &-M_{n-2}^{UL}(1) &-M_{n-2}^{UR}(1) \\
B^* &I_d &C_0
    \end{array}\right)\right|\\
    =&|\det B|^2\cdot |\det B^*|\cdot
    \left|\det \left(\begin{array}{cccc}
0 &-M_{n-1}^{UL}(1) &B-M_{n-1}^{UR}(1) \\
B &-M_{n-2}^{UL}(1) &-M_{n-2}^{UR}(1) \\
I_d &(B^*)^{-1} &(B^*)^{-1}C_0
    \end{array}\right)\right|\\
    =&|\det B|^4 \cdot 
\left|\det\left(\begin{array}{cccc}
0& 0  &-M_{n-1}^{UL}(1) &B-M_{n-1}^{UR}(1)  \\
0& I_d  &-M_{n-1}^{LL}(1) &-M_{n-1}^{LR}(1) \\
I_d &0  &0 &-B \\
0& I_d  &(B^*)^{-1} &(B^*)^{-1} C_0 
    \end{array}\right)\right|\notag\\
=&|\det B|^4\cdot 
\left|\det\left(\begin{array}{c|c}
\left(\begin{matrix}0& 0\\
0 &I_d\end{matrix}\right) &\left(\begin{matrix}0 & B\\ 0 & 0\end{matrix}\right)-M_{n-1}(1) \\
 \hline
I_{2d} &-M_1^{-1}(0)
    \end{array}\right)\right|\notag
    \end{align}
 Rewriting this in block form yields 
 \begin{align}\label{eq:dem2}
    |\det \tilde{S}_1|
=&|\det B|^4\cdot 
\left|\det\left(\begin{array}{c|c}
\left(\begin{matrix}0& 0\\
0 &I_d\end{matrix}\right) &\left(\begin{matrix}0 & B\\ 0 & 0\end{matrix}\right)M_1(0)-M_{n}(0) \\
 \hline
I_{2d} &-I_{2d}
    \end{array}\right)\right|\notag\\
=&|\det B|^4\cdot 
\left|\det\left(\begin{array}{c|c}
\left(\begin{matrix}0& 0\\
0 &I_d\end{matrix}\right) &I_{2d}-M_{n}(0) \\
 \hline
I_{2d} &0
    \end{array}\right)\right|\notag\\
=:&|\det B|^4\cdot |\det M_n(0)-I_{2d}|.
\end{align}
The claimed result follows by combining \eqref{eq:fn=tildeS1} with \eqref{eq:dem2}.
\end{proof}

To estimate $\det(M_{n,E}(\theta)-I_{2d})$, we need the following elementary lemma concerning the vector spaces associated with the singular value decomposition of a symplectic matrix. We include the proof for the sake of completeness.

\begin{lemma}\label{lem:Omega_v}
Let $M\in \C^{2d\times 2d}$ be a complex symplectic matrix, i.e., $M^*\Omega M=\Omega$. Let $\sigma_1\geq ...\geq \sigma_{2d}$ be the singular values of~$M$. Let $\mu_1>...>\mu_k$ be the non-repeated singular values and $V_j$, $1\leq j\leq k$, be the singular vector space of $M$ associated to $\mu_j$. Similarly, let $W_j$ be the singular vector space of $M^*$ associated to $\mu_j$. Then 
\begin{align}
\Omega V_j=V_{k+1-j}, \text{ and } \Omega W_j=W_{k+1-j}.
\end{align}
\end{lemma}
\begin{proof}
Let $v_j\in V_j$ be a unit vector, and $w_j$ be the unit vector such that 
\begin{align}\label{eq:Mw=sv1}
M v_j=\mu_j w_j.
\end{align}
This implies 
\begin{align}\label{eq:Mw=sv2}
\begin{cases}
    M^* w_j=\mu_j v_j\\
    (M^*)^{-1}v_j=\mu_j^{-1} w_j\\
    M^{-1}w_j=\mu_j^{-1}v_j
\end{cases}
\end{align}
Since $M$ is symplectic, we also have 
\begin{align}\label{eq:sym2}
\begin{cases}
    M^*\Omega M=\Omega\\
    \Omega M^*\Omega=-M^{-1}\\
    \Omega M\Omega =-(M^*)^{-1}
\end{cases}
\end{align}
Combining \eqref{eq:Mw=sv1}, \eqref{eq:Mw=sv2} with \eqref{eq:sym2}, yields
\begin{align}
    M^*M\Omega v_j=-M^* \Omega^2 M \Omega v_j=M^* \Omega (M^*)^{-1} v_j
=M^*\Omega \mu_j^{-1} w_j
=&-\mu_j^{-1} \Omega^2 M^* \Omega w_j\\
=&\mu_j^{-1} \Omega M^{-1}w_j\\
=&\mu_j^{-2} \Omega v_j.
\end{align}
This implies that $\Omega v_j$ is an eigenvector of $M^*M$ associated to the eigenvalue $\mu_j^{-2}=\mu_{k+1-j}^2$, whence $\Omega v_j\in V_{k+1-j}$.
Therefore, $\Omega V_j\subset V_{k+1-j}$ for any $1\leq j\leq k$, which implies that $\Omega V_j=V_{k+1-j}$. The claim concerning $W_j$ can be proved analogously.
\end{proof}

Next, we will apply the previous lemma to $M_{n,E}(\theta)$. For simplicity, we shall write $\hat{L}^j(d\alpha, M_E)$, $\hat{L}_j(d\alpha, M_E)$ as $\hat{L}^j$ and $\hat{L}_j$, respectively.
Recall that the large deviation sets $\mathcal{B}_{n,E,\varepsilon}^d$ was defined in~\eqref{def:Bn_LDT}.
By Lemma \ref{lem:LDT}, this set is exponentially small for large enough $n$. 
We fix $0<\varepsilon<\frac{\hat{L}_d}{11\hat{L}^d}$.
Note that for such $\varepsilon$, 
\begin{align}\label{def:epsilon}
    (1-10\varepsilon)\hat{L}^d>(1+\varepsilon)\hat{L}^{d-1}
\end{align}
Let $\theta\in (\mathcal{B}^d_{n,E,\varepsilon})^c$, 
and $v_{j,n}(\theta)$ be the normalized singular vector of $M_{n,E}(\theta)$ corresponding to the singular value $\sigma_j(M_{n,E}(\theta))$. Suppose $w_{j,n}(\theta)$ is the normalized singular vector of $M_{n,E}^*(\theta)$ such that
$$M_{n,E}(\theta)v_{j,n}(\theta)=\sigma_j(M_{n,E}(\theta)) w_{j,n}(\theta).$$
According to Lemmas \ref{lem:upper_semi_cont} and \ref{lem:LDT}, for $\theta\in (\mathcal{B}_{n,E,\varepsilon}^d)^c$, one has 
\begin{align}
    \prod_{j=1}^d \sigma_j(M_{n,E}(\theta))=\|\wedge^d M_{n,E}(\theta)\|\geq (1-\varepsilon)\cdot \hat{L}^d\geq (1+\varepsilon)\cdot \hat{L}^{d+1}\geq \prod_{j=1}^{d+1}\sigma_j(M_{n,E}(\theta)).
\end{align}
It follows that $$\sigma_d(M_{n,E}(\theta))>1>\sigma_{d+1}(M_{n,E}(\theta))=(\sigma_d(M_{n,E}(\theta)))^{-1}.$$
Let $V_{j,n}(\theta)$ and $W_{j,n}(\theta)$ be as defined in Lemma~\ref{lem:Omega_v} with $M$ replaced with $M_{n,E}(\theta)$. Clearly $k$ is even. By that lemma, for any $1\leq j\leq k/2$,
\begin{align}\label{eq:Omega_V_Mn}
    \Omega V_{j,n}(\theta)=V_{k+1-j,n}(\theta), \text{ and } \Omega W_{j,n}(\theta)=W_{k+1-j,n}(\theta).
\end{align}
This implies the following relations between the most expanding and most contracting $d$-planes:
\begin{align}
    &(\wedge^d\Omega) (v_{1,n}(\theta)\wedge\cdots \wedge v_{d,n}(\theta))=\pm v_{d+1,n}(\theta)\wedge\cdots \wedge v_{2d,n}(\theta), \text{ and }\\
    &(\wedge^d\Omega) (w_{1,n}(\theta)\wedge\cdots \wedge w_{d,n}(\theta))=\pm w_{d+1,n}(\theta)\wedge\cdots \wedge w_{2d,n}(\theta),
\end{align}
and in particular
\begin{align}\label{eq:Omega_v_w}
    &\left|\langle v_{1,n}(\theta)\wedge\cdots \wedge v_{d,n}(\theta), w_{1,n}(\theta)\wedge \cdots \wedge w_{d,n}(\theta)\rangle\right|\\
    &\qquad\qquad\qquad=
    \left|\langle v_{d+1,n}(\theta)\wedge\cdots \wedge v_{2d,n}(\theta), w_{d+1,n}(\theta)\wedge \cdots \wedge w_{2d,n}(\theta)\rangle\right|
\end{align}
The following result will play a crucial role in our argument, cf.~\cite{GS3}*{Lemma 3.2}. It will enable us to show that for a subsequence of $n$, the most expanding $d$-planes of $M_{n,E}(\theta)$ and $M_{n,E}^*(\theta)$ are not entirely orthogonal.
\begin{lemma}\label{lem:Mn2}
    Assume that  $\beta(\alpha)=0$. Then for any small $\varepsilon>0$, there exists $N_0=N_0(\varepsilon,\alpha,E,d)$ and $\kappa_0=\kappa_0(\varepsilon,\alpha,E,d,N_0)>0$ such that for any $\kappa_0$-admissible $N_1$ (those satisfying $\|N_1d\alpha\|_{\T}\leq \kappa_0$), and $N_1\geq 8N_0$, we have
    \begin{align}
\mes(\tilde{\mathcal{B}}^d_{N_1,E,\varepsilon}):=\mes\Big(\left\{\theta:\,  \left|\frac{1}{2N_1}\log\|\wedge^d \left(M_{N_1,E}^2(\theta)\right)\|- \hat{L}^d\right|>\varepsilon\cdot \hat{L}^d\right\}\Big)\leq \frac{1}{20}
    \end{align}
\end{lemma}
\begin{proof}
By Lemma \ref{lem:LDT}, we can choose $N_0:=N_0(\varepsilon,\alpha,E,d)$ so that for $n\geq N_0$, one has $\mathrm{mes}(\mathcal{B}_{n,E,\varepsilon}^d)\leq e^{-n\varepsilon^2 \cdot \hat{L}^d(d\alpha,M_E)}$ and, furthermore, 
\begin{align}\label{eq:N0<1/100}
    \sum_{j=0}^{\infty}e^{-2^j c_1\varepsilon^2 N_0\cdot \hat{L}^d(d\alpha,M_E)}<1/100,
\end{align}
where $c_1$ is the constant in Lemma \ref{lem:LDT},
as well as 
\begin{align}\label{eq:N0<c0}
\sum_{j=0}^{\infty} e^{-2^{j+1}N_0(1-\varepsilon)\hat{L}^d(d\alpha,M_E)}\cdot e^{2^{j+5}N_0\varepsilon\hat{L}^d(d\alpha, M_E)}\cdot e^{8\tilde{C}_0c_0}<c_0,
\end{align}
where $\tilde{C}_0, c_0>0$ are the constants in Theorem~\ref{thm:AP}.
For the purposes of this proof, we let $\mathcal{G}_{n}:=(\mathcal{B}_{n,E,\varepsilon}^d)^c$, omitting the fixed parameters $E$ and $\varepsilon$ for simplicity.
We also denote $\hat{L}^k(d\alpha, M_E)$ by $\hat{L}^k$.
First, we need the following estimate at the fixed scale $N_0$.
\begin{lemma}\label{lem:kappa_0}
    There exists $\kappa_0=\kappa_0(\varepsilon,\alpha, E, d, N_0)$ such that for any $\|\kappa\|_{\T}\leq \kappa_0$, and for any $\theta_0$ satisfying
    \begin{align}\label{eq:kappa_0_assume}
        \|\wedge^d M_{2N_0,E}(\theta_0)\|\geq e^{2N_0(1-\varepsilon)\hat{L}^d}
    \end{align}
    one has the bound
    \begin{align}
        \|\wedge^d \left(M_{N_0,E}(\theta_0+N_0d\alpha+\kappa) M_{N_0,E}(\theta_0)\right)\|\geq e^{2N_0(1-2\varepsilon)\hat{L}^d}.
    \end{align}
\end{lemma}
\begin{proof}
    We use the following estimate from \cite{DK}:
    \begin{lemma}\cite{DK}*{Lemma 2.8}\label{lem:2_8_DK}
        For $k\times k$ complex matrices $h_1, h_2$, for any $1\leq j\leq k$, 
        \begin{align}
            \|\wedge^j h_1-\wedge^j h_2\|\leq j\max\{1, \|h_1\|, \|h_2\|\}^{j-1}\, \|h_1-h_2\|. 
        \end{align}
    \end{lemma}
    We apply Lemma \ref{lem:2_8_DK} to $k=2d$, $j=d$, and
    \begin{align}
        h_1(\theta)=&M_{N_0,E}(\theta+N_0d\alpha+\kappa) M_{N_0,E}(\theta) \\
        h_2(\theta)=&M_{N_0,E}(\theta+N_0d\alpha)M_{N_0,E}(\theta)
    \end{align}
    By Lemma~\ref{lem:upper_semi_cont}, uniformly in $\theta$,
    \begin{align}
        \max(\|h_1(\theta)\|, \|h_2(\theta)\|)\leq e^{2N_0(1+\varepsilon)\hat{L}^d}, \text{ and } \|M_{N_0,E}(\theta)\|\leq e^{N_0(1+\varepsilon)\hat{L}^d}.
    \end{align}
    Hence, Lemma \ref{lem:2_8_DK} implies furthermore that uniformly in $\theta$,
    \begin{align}
        \|\wedge^d h_1(\theta)-\wedge^d h_2(\theta)\|
        \leq &de^{2dN_0(1+\varepsilon)\hat{L}^d}\cdot \|h_1(\theta)-h_2(\theta)\|\\
        \leq &de^{3dN_0(1+\varepsilon)\hat{L}^d}\cdot \|M_{N_0,E}(\theta+N_0d\alpha+\kappa)-M_{N_0,E}(\theta+N_0d\alpha)\|\\
        \leq &C_{\alpha, d, E, N_0} e^{4dN_0(1+\varepsilon)\hat{L}^d}\cdot \|\kappa\|_{\T},
    \end{align}
    where in the last step we used a telescoping argument.
    Hence for $\|\kappa\|_{\T}\leq \kappa_0(\varepsilon, \alpha,E,d,N_0)$, one has  uniformly in $\theta$,
    \begin{align}
        \|\wedge^d h_1(\theta)-\wedge^d h_2(\theta)\|\leq e^{2N_0(1-2\varepsilon)\hat{L}^d}.
    \end{align}
    This, combined with \eqref{eq:kappa_0_assume} at $\theta=\theta_0$, implies that
    \begin{align}
        \|\wedge^d (M_{N_0,E}(\theta_0+N_0d\alpha+\kappa) M_{N_0,E}(\theta_0))\|\geq e^{2N_0(1-\varepsilon)\hat{L}^d}-e^{2N_0(1-2\varepsilon)\hat{L}^d}\geq e^{2N_0(1-2\varepsilon)\hat{L}^d},
    \end{align}
    which is the claimed result.
\end{proof}

Next, we pick an arbitrary $\kappa_0$-admissible $N_1$ (i.e., $\|N_1d\alpha\|_{\T}\leq \kappa_0$), satisfying $N_1\geq 8N_0$.
Let $k\geq 2$ be such that $2^{k+1}N_0\leq N_1<2^{k+2}N_0$, and $n_0=N_1-2^k N_0=N_1-N_0-\sum_{\ell=0}^{k-1}2^{\ell}N_0$ be the remainder. 
Clearly
\begin{align}\label{eq:n0_N0}
2^k N_0\leq n_0\leq (3\cdot 2^k-1) N_0.
\end{align}
Our goal is to estimate 
\begin{align}
    \log \|\wedge^d (M_{N_1,E}^2(\theta))\|=\log \|\wedge^d (M_{N_1,E}(\theta+N_1d\alpha+\kappa)M_{N_1,E}(\theta))\|,
\end{align}
where $\kappa=-N_1d\alpha$ (modulo $\Z$) satisfies $\|\kappa\|_{\T}\leq \kappa_0$.
In order to accomplish this, we decompose
\begin{align}
    &M_{N_1,E}(\theta+N_1d\alpha+\kappa)M_{N_1,E}(\theta)= M_{N_1,E}(\theta)^2\\
    =&
    M_{n_0,E}(\theta+(2N_1-n_0)d\alpha+\kappa)
    M_{2^{k-1}N_0,E}(\theta+(N_1+2^{k-1}N_0)d\alpha+\kappa)
    \cdots\\
    &\cdots 
    M_{4N_0,E}(\theta+(N_1+2^2 N_0)d\alpha+\kappa)M_{2N_0,E}(\theta+(N_1+2N_0)d\alpha+\kappa)\\
    &\,\,\,\,\,\cdot
    M_{N_0,E}(\theta+(N_1+N_0)d\alpha+\kappa)M_{N_0,E}(\theta+N_1d\alpha+\kappa)\cdot \\
    &\,\,\,\,\,\cdot M_{N_0,E}(\theta+(N_1-N_0)d\alpha) M_{N_0,E}(\theta+(N_1-2N_0)d\alpha)\\
    &\,\,\,\,\,\cdot M_{2N_0,E}(\theta+(N_1-2^2 N_0)d\alpha) M_{4N_0,E}(\theta+(N_1-2^3 N_0)d\alpha)\\
    &\cdots 
    M_{2^{k-1}N_0,E}(\theta+n_0d\alpha) M_{n_0,E}(\theta)
\end{align}

\underline{Step 1}. 
Let 
\begin{align}
\mathcal{G}_{N_0}^{(1)}:=&(\mathcal{G}_{N_0}-(N_1-2N_0)d\alpha)\cap (\mathcal{G}_{N_0}-(N_1-N_0)d\alpha)\cap (\mathcal{G}_{N_0}-N_1d\alpha-\kappa)\cap\\ 
&\cap (\mathcal{G}_{N_0}-(N_1+N_0)d\alpha-\kappa)\cap (\mathcal{G}_{2N_0}-(N_1-2N_0)d\alpha)\cap  (\mathcal{G}_{2N_0}-N_1d\alpha-\kappa)\cap \\
&\cap (\mathcal{G}_{2N_0}-(N_1-N_0)d\alpha)
\end{align}
We restrict to $\theta\in \mathcal{G}_{N_0}^{(1)}$ for the remainder of the proof.
Let 
\begin{align}\label{def:gN0}
    g_{N_0,1}(\theta)=&\wedge^d M_{N_0,E}(\theta+(N_1-2N_0)d\alpha)\\
    g_{N_0,2}(\theta)=&\wedge^d M_{N_0,E}(\theta+(N_1-N_0)d\alpha)\\
    g_{N_0,3}(\theta)=&\wedge^d M_{N_0,E}(\theta+N_1d\alpha+\kappa)\\
    g_{N_0,4}(\theta)=&\wedge^d M_{N_0,E}(\theta+(N_1+N_0)d\alpha+\kappa)
\end{align}
The goal of Step 1 is to show the following.
\begin{lemma}\label{lem:AP_1}
For any $\theta\in \mathcal{G}_{N_0}^{(1)}$, one has 
\begin{align}
    \log \|g_{N_0,4}(\theta)\cdots g_{N_0,1}(\theta)\|\geq 4N_0(1-3\varepsilon)\hat{L}^d-4\tilde{C}_0 \tilde{\kappa}_{N_0}\cdot \tilde{\varepsilon}_{N_0}^{-2},
\end{align}
for the constant $\tilde{C}_0>0$ in Theorem \ref{thm:AP} and with 
\begin{align}\label{def:tepsilon_N0}
\tilde{\kappa}_{N_0}=e^{-2N_0(1-\varepsilon)\hat{L}_d+2^2 N_0\varepsilon\hat{L}^{d-1}}, \text{ and }  \tilde{\varepsilon}_{N_0}=e^{-2^3 N_0\varepsilon\hat{L}^d}.
\end{align}
Furthermore, 
\begin{align}\label{eq:mes_GN0_1}
\mathrm{mes}(\mathcal{G}_{N_0}^{(1)})\geq 1-4e^{-c_1\varepsilon^2N_0\hat{L}^d}-3e^{-2c_1\varepsilon^2 N_0\hat{L}^d},
\end{align}
where $c_1>0$ is the constant in Lemma \ref{lem:LDT}.
\end{lemma}
\begin{proof}
By Lemma \ref{lem:LDT}, for all $\theta\in \mathcal{G}_{N_0}^{(1)}$, we have
\begin{align}\label{eq:gN_0_j_sigma1}
    e^{N_0(1-\varepsilon)\hat{L}^d}\leq \sigma_1(g_{N_0,j}(\theta))\leq e^{N_0(1+\varepsilon)\hat{L}^d}, \text{ for } j=1,2,3,4,
\end{align}
and
\begin{align}\label{eq:gN_0_43_21}
    \|g_{N_0,4}(\theta)g_{N_0,3}(\theta)\|\geq e^{2N_0(1-\varepsilon)\hat{L}^d}, \text{ and } \|g_{N_0,2}(\theta)g_{N_0,1}(\theta)\|\geq e^{2N_0(1-\varepsilon)\hat{L}^d}.
\end{align}
Note that by definition $\theta\in \mathcal{G}_{N_0}^{(1)}\subset  \mathcal{G}_{2N_0}-(N_1-N_0)d\alpha$, hence
\begin{align}
    \|\wedge^d (M_{N_0,E}(\theta+N_1d\alpha) M_{N_0,E}(\theta+(N_1-N_0)d\alpha))\|\geq e^{2N_0(1-\varepsilon)\hat{L}^d}.
\end{align}
Combining the above (which verifies the assumptions of Lemma~\ref{lem:kappa_0} for $\theta_0=\theta+(N_1-N_0)d\alpha$) with Lemma~\ref{lem:kappa_0} and $\|\kappa\|_{\T}\leq \kappa_0$, yields 
\begin{align}\label{eq:gN_0_32}
    \|g_{N_0,3}(\theta)g_{N_0,2}(\theta)\|
    =&\|\wedge^d (M_{N_0,E}(\theta+N_1d\alpha+\kappa) M_{N_0,E}(\theta+(N_1-N_0)d\alpha))\|\\
    \geq &e^{2N_0(1-2\varepsilon)\hat{L}^d}.
\end{align}
By Lemma~\ref{lem:upper_semi_cont} applied to the matrix $\wedge^2 (\wedge^d M_{N_0,E})$ and 
\begin{align}\label{eq:L_wedge2_wedged}
 {L}^1(d\alpha, \wedge^2 (\wedge^d M_{E}))
=& {L}^1(d\alpha, \wedge^d M_{E})+ {L}^2(d\alpha, \wedge^d M_{E})\\
=& {L}^{d}+ {L}^{d-1}+ {L}_{d+1}=2 {L}^{d-1},
\end{align} 
one has uniformly in $\theta\in \T$,
\begin{align}
\sigma_1(\wedge^d M_{N_0,E}(\theta))\cdot \sigma_2(\wedge^d M_{N_0,E}(\theta))=
    \sigma_1(\wedge^2(\wedge^d M_{N_0,E}(\theta)))\leq e^{2N_0(1+\varepsilon)\hat{L}^{d-1}}.
\end{align}
Combining the above with the lower bounds in \eqref{eq:gN_0_j_sigma1}, we conclude that 
\begin{align}\label{eq:gN_0_j_sigma2}
    \sigma_2(g_{N_0,j}(\theta))\leq e^{N_0(1-\varepsilon)(\hat{L}^{d-1}+\hat{L}_{d+1})}\cdot e^{4N_0\varepsilon \hat{L}^{d-1}}.
\end{align}
Note that \eqref{eq:gN_0_j_sigma1}, \eqref{eq:gN_0_j_sigma2}, \eqref{eq:gN_0_43_21} and \eqref{eq:gN_0_32} verify the assumptions of Theorem~\ref{thm:AP}, the Avalanche Principle. In fact, for $j=1,2,3,4$,
\begin{align}\label{eq:AP_N0_1}
\frac{\sigma_1(g_{N_0,j}(\theta))}{\sigma_2(g_{N_0,j}(\theta))}\geq e^{2N_0(1-\varepsilon) \hat{L}_d}e^{-4N_0\varepsilon \hat{L}^{d-1}}=:\tilde{\kappa}_{N_0}^{-1},
\end{align}
and for $j=1,2,3$,
\begin{align}\label{eq:AP_N0_2}
    \frac{\|g_{N_0,j+1}(\theta)g_{N_0,j}(\theta)\|}{\sigma_1(g_{N_0,j+1}(\theta))\cdot \sigma_1(g_{N_0,j}(\theta))}\geq e^{-2^3 N_0\varepsilon\hat{L}^d}=:\tilde{\varepsilon}_{N_0}.
\end{align}
Clearly, 
\begin{align}\label{eq:AP_N0_c0}
    \tilde{\kappa}_{N_0}\cdot \tilde{\varepsilon}_{N_0}^{-2}\leq e^{-2N_0(1-\varepsilon) \hat{L}_d} e^{(2^4+2^2) N_0\varepsilon\hat{L}^d}<c_0,
\end{align}
due to \eqref{eq:N0<c0}.
Therefore, by applying Theorem \ref{thm:AP} to $g_{N_0,4}(\theta)\cdots g_{N_0,1}(\theta)$, we obtain with the constant $\tilde{C}_0>0$ in Theorem \ref{thm:AP},
\begin{align}\label{eq:AP_conclusion_N0}
\log \|g_{N_0,4}(\theta)\cdots g_{N_0,1}(\theta)\|
\geq &\sum_{j=1}^3 \log \|g_{N_0,j+1}(\theta)g_{N_0,j}(\theta)\|-\sum_{j=2}^3 \log \|g_{N_0,j}(\theta)\|-4\tilde{C}_0\tilde{\kappa}_{N_0}\cdot \tilde{\varepsilon}_{N_0}^{-2}
\\
\geq &4N_0(1-3\varepsilon)\hat{L}^d-4\tilde{C}_0 \tilde{\kappa}_{N_0}\cdot \tilde{\varepsilon}_{N_0}^{-2}.
\end{align}
Regarding the measure estimate, we note that 
\begin{align}
    \mathrm{mes}(\mathcal{G}_{N_0}^{(1)})\geq 1-4e^{-c_1\varepsilon^2 N_0\hat{L}^d}-3e^{-2c_1\varepsilon^2 N_0\hat{L}^d},
\end{align}
where we used that for any $n\geq N_0$, $\mathrm{mes}(\mathcal{G}_{n})=1-\mathrm{mes}(\mathcal{B}_{n,E,\varepsilon}^d)\geq 1-e^{-c_1\varepsilon^2 n\hat{L}^d}$, due to Lemma~\ref{lem:LDT}.
\end{proof}

To better illustrate the inductive procedure, we present the second step in detail as well.

\medskip

\underline{Step 2}.  Let
\begin{align}\label{def:g2N0}
    g_{2N_0,1}(\theta)=&\wedge^d M_{2N_0,E}(\theta+(N_1-2^2 N_0)d\alpha)\\
    g_{2N_0,2}(\theta)=&\wedge^d M_{2N_0,E}(\theta+(N_1-2N_0)d\alpha)\\
    g_{2N_0,3}(\theta)=&\wedge^d M_{2N_0,E}(\theta+N_1d\alpha+\kappa)\\
    g_{2N_0,4}(\theta)=&\wedge^d M_{2N_0,E}(\theta+(N_1+2N_0)d\alpha+\kappa)
\end{align}
Note that 
\[g_{2N_0,3}(\theta)g_{2N_0,2}(\theta)=g_{N_0,4}(\theta)\cdots g_{N_0,1}(\theta).\]
Hence by Lemma \ref{lem:AP_1}, for $\theta\in \mathcal{G}_{N_0}^{(1)}$, 
\begin{align}\label{eq:g2N_0_32}
    \log \|g_{2N_0,3}(\theta)g_{2N_0,2}(\theta)\|\geq 4N_0(1-3\varepsilon)\hat{L}^d-4\tilde{C}_0\tilde{\kappa}_{N_0}\cdot \tilde{\varepsilon}_{N_0}^{-2}.
\end{align}
We further restrict $\theta$ to the following subset of $\mathcal{G}_{N_0}^{(1)}$:
\begin{align}
    \theta\in \mathcal{G}_{2N_0}^{(1)}=&\mathcal{G}_{N_0}^{(1)}\cap (\mathcal{G}_{2N_0}-(N_1-2^2N_0)d\alpha)\cap (\mathcal{G}_{2N_0}-(N_1+2N_0)d\alpha-\kappa)\cap\\
    &\qquad \cap (\mathcal{G}_{4N_0}-(N_1-2^2N_0)d\alpha)\cap (\mathcal{G}_{4N_0}-N_1d\alpha-\kappa).
\end{align}
Our goal of Step 2 is to prove the following:
\begin{lemma}\label{lem:AP_2}
For $\theta\in \mathcal{G}_{2N_0}^{(1)}$, 
\begin{align}
    \log \|g_{2N_0,4}(\theta)\cdots g_{2N_0,1}(\theta)\|\geq 2^3N_0(1-3\varepsilon)\hat{L}^d-4\tilde{C}_0\sum_{\ell=0}^1 \tilde{\kappa}_{2^{\ell} N_0}\cdot \tilde{\varepsilon}_{2^{\ell}N_0}^{-2},
\end{align}
where 
\begin{align}\label{def:tepsilon_2N0}
\tilde{\kappa}_{2N_0}=e^{-2^2 N_0(1-\varepsilon)\hat{L}_d}\cdot e^{2^3N_0\varepsilon\hat{L}^{d-1}}\, \text{ and } \tilde{\varepsilon}_{2N_0}=e^{-2^4 N_0\varepsilon\hat{L}^d}\cdot e^{-4\tilde{C}_0\tilde{\kappa}_{N_0}\cdot \tilde{\varepsilon}_{N_0}^{-2}}.
\end{align}
Furthermore,  
\begin{align}
    \mathrm{mes}(G_{2N_0}^{(1)})\geq  1-4e^{-c_1\varepsilon^2 N_0\hat{L}^d}-5e^{-2c_1\varepsilon^2 N_0\hat{L}^d}-2e^{-4c_1\varepsilon^2 N_0\hat{L}^d}.
\end{align}
\end{lemma}
\begin{proof}
    Since $\theta\in \mathcal{G}_{2N_0}^{(1)}$, by  Lemma~\ref{lem:LDT} one has for $j=1,2,3,4$ that 
    \begin{align}\label{eq:g2N_0_j_sigma1}
        e^{2N_0(1-\varepsilon)\hat{L}^d}\leq \sigma_1(g_{2N_0,j}(\theta))\leq e^{2N_0(1-\varepsilon)\hat{L}^d},
    \end{align}
    and for $j=1,3$ that
    \begin{align}\label{eq:g2N_0_43_21}
        \|g_{2N_0,j+1}(\theta)g_{2N_0,j}(\theta)\|\geq e^{4N_0(1-\varepsilon)\hat{L}^d}.
    \end{align}
    In analogy to \eqref{eq:gN_0_j_sigma2}, we have for $j=1,2,3,4$ that
    \begin{align}\label{eq:g2N_0_sigma2}
        \sigma_2(g_{2N_0,j}(\theta))\leq e^{2N_0(1-\varepsilon)(\hat{L}^{d-1}+\hat{L}_{d+1})}\cdot e^{8N_0\varepsilon \hat{L}^{d-1}}.
    \end{align}
    We apply the avalanche principle (AP) Theorem \ref{thm:AP} to $g_{2N_0,4}(\theta),...,g_{2N_0,1}(\theta)$. The hypotheses of the AP hold since:
    for $j=1,2,3,4$, by \eqref{eq:g2N_0_j_sigma1} and \eqref{eq:g2N_0_sigma2},
    \begin{align}
        \frac{\sigma_1(g_{2N_0,j}(\theta))}{\sigma_2(g_{2N_0,j}(\theta))}\geq e^{2^2 N_0(1-\varepsilon)\hat{L}_d}e^{-2^3 N_0\varepsilon\hat{L}^{d-1}}=:\tilde{\kappa}_{2N_0}^{-1},
    \end{align}
    and for $j=1,2,3$, by further taking \eqref{eq:g2N_0_43_21} and \eqref{eq:g2N_0_32} into account,
    \begin{align}
        \frac{\|g_{2N_0,j+1}(\theta)g_{2N_0,j}(\theta)\|}{\sigma_1(g_{2N_0,j+1}(\theta))\cdot \sigma_1(g_{2N_0,j}(\theta))}\geq e^{-2^4 N_0\varepsilon\hat{L}^d} e^{-4\tilde{C}_0\tilde{\kappa}_{N_0}\cdot \tilde{\varepsilon}_{N_0}^{-2}}=:\tilde{\varepsilon}_{2N_0}. 
    \end{align}
    Clearly, 
    \begin{align}\label{eq:AP2_c0}
\tilde{\kappa}_{2N_0} \cdot \tilde{\varepsilon}_{2N_0}^{-2}\leq e^{-2^2 N_0(1-\varepsilon)\hat{L}_d}e^{(2^5+2^3)N_0\varepsilon\hat{L}^d}e^{8\tilde{C}_0c_0}<c_0,
    \end{align}
where we used \eqref{eq:AP_N0_c0} to bound $\tilde{\kappa}_{N_0}\cdot \tilde{\varepsilon}_{N_0}^{-2}<c_0$.
Therefore, Theorem \ref{thm:AP} implies that 
\begin{align}
    \log \|g_{2N_0,4}(\theta)\cdots g_{2N_0,1}(\theta)\|
    \geq &\sum_{j=1}^3 \log \|g_{2N_0,j+1}(\theta)g_{2N_0,j}(\theta)\|-\sum_{j=2}^3\log \|g_{2N_0,j}(\theta)\|-4\tilde{C}_0\tilde{\kappa}_{2N_0}\cdot \tilde{\varepsilon}_{2N_0}^{-2}\\
    \geq &2^3N_0(1-3\varepsilon)\hat{L}^d-4\tilde{C}_0 \sum_{\ell=0}^1 \tilde{\kappa}_{2^{\ell} N_0}\cdot \tilde{\varepsilon}_{2^{\ell} N_0}^{-2}.
\end{align}
Regarding the measure estimate,  Lemma \ref{lem:LDT} implies that
\begin{align}
    \mathrm{mes}(\mathcal{G}_{2N_0}^{(1)})
    \geq &\mathrm{mes}(\mathcal{G}_{N_0}^{(1)})-2e^{-2c_1\varepsilon^2 N_0}-2e^{-2c_2\varepsilon^2N_0}\\
    \geq &1-4e^{-c_1\varepsilon^2 N_0\hat{L}^d}-5e^{-2c_1\varepsilon^2 N_0\hat{L}^d}-2e^{-4c_1\varepsilon^2 N_0\hat{L}^d}.
\end{align}
This is the claimed result.
\end{proof}
Next, we carry out the following inductive procedure.
For any $\ell\geq 0$, let 
\begin{align}\label{def:g2ell_N0}
        g_{2^{\ell}N_0,1}(\theta)=&\wedge^d M_{2^{\ell}N_0,E}(\theta+(N_1-2^{\ell+1}N_0)d\alpha)\\
        g_{2^{\ell}N_0,2}(\theta)=&\wedge^d M_{2^{\ell}N_0,E}(\theta+(N_1-2^{\ell}N_0)d\alpha)\\
        g_{2^{\ell}N_0,3}(\theta)=&\wedge^d M_{2^{\ell}N_0,E}(\theta+N_1d\alpha+\kappa)\\
        g_{2^{\ell}N_0,4}(\theta)=&\wedge^d M_{2^{\ell}N_0,E}(\theta+(N_1+2^{\ell}N_0)d\alpha+\kappa)
\end{align}
Note this is consistent with the definitions of $g_{N_0,j}(\theta)$ and $g_{2N_0,j}(\theta)$ in \eqref{def:gN0} and \eqref{def:g2N0}.
For any $\ell\geq 0$, let
\begin{align}
    \tilde{\kappa}_{2^{\ell}N_0}=&e^{-2^{\ell+1}N_0(1-\varepsilon)\hat{L}_d}\cdot  e^{2^{\ell+2}N_0\varepsilon\hat{L}^{d-1}},
\end{align}
and define $\tilde{\varepsilon}_{2^{\ell}N_0}$ inductively;   $\tilde{\varepsilon}_{N_0}$ is defined in~\eqref{def:tepsilon_N0} and for $\ell\geq 1$,
\begin{align}
    \tilde{\varepsilon}_{2^{\ell}N_0}=&e^{- 2^{\ell+3}N_0\varepsilon\hat{L}^d}\cdot  e^{-4\tilde{C}_0\sum_{r=0}^{\ell-1}\tilde{\kappa}_{2^r N_0}\cdot \tilde{\varepsilon}_{2^r N_0}^{-2}}.
\end{align}
Note that these definitions are consistent with those in \eqref{def:tepsilon_N0} and \eqref{def:tepsilon_2N0} for $\ell=0,1$.
Clearly, for any $n\geq 1$,
\begin{align}\label{eq:sum_kp_eps^2}
\sum_{\ell=0}^{n} \tilde{\kappa}_{2^{\ell}N_0}\cdot \tilde{\varepsilon}_{2^{\ell}N_0}^{-2}\leq 
&\sum_{\ell=0}^{n} e^{-2^{\ell+1}N_0(1-\varepsilon)\hat{L}_d}\cdot e^{(2^{\ell+4}+2^{\ell+2})N_0\varepsilon\hat{L}^d}\cdot e^{8\tilde{C}_0\sum_{r=0}^{n-1}\tilde{\kappa}_{2^rN_0}\cdot \tilde{\varepsilon}_{2^rN_0}^{-2}}\\
\leq &\sum_{\ell=0}^{n} e^{-2^{\ell+1}N_0(1-\varepsilon)\hat{L}_d}\cdot e^{(2^{\ell+4}+2^{\ell+2})N_0\varepsilon\hat{L}^d}\cdot e^{8\tilde{C}_0c_0}<c_0,
\end{align}
by \eqref{eq:N0<c0}, \eqref{eq:AP_N0_c0}, and induction.

\medskip 

\underline{Step $m$ for $m=3,...,k$}.
\begin{lemma}
Suppose $\mathcal{G}_{2^{m-2} N_0}^{(1)}$ is defined and satisfies the following two conditions:
\begin{enumerate}
    \item 
     \begin{align}
        \mathcal{G}_{2^{m-2}N_0}^{(1)}\subset (\mathcal{G}_{2^{m-1}N_0}-(N_1-2^{m-1}N_0)d\alpha)\cap (\mathcal{G}_{2^{m-1}N_0}-N_1d\alpha-\kappa)
     \end{align}
    \item 
    \begin{align}
        \mathrm{mes}(G_{2^{m-2}N_0}^{(1)})\geq 1-5\sum_{\ell=0}^{m-2}e^{-2^{\ell} c_1\varepsilon^2 N_0 \hat{L}^d}-2e^{-2^{m-1}c_1\varepsilon^2N_0\hat{L}^d}.
    \end{align}
\end{enumerate}
    Assume further that for any $\theta\in \mathcal{G}_{2^{m-2}N_0}^{(1)}$ the following holds:
    \begin{align}\label{eq:induct_assume_g1-4}
        \log \|g_{2^{m-2}N_0,4}(\theta)\cdots g_{2^{m-2}N_0,1}(\theta)\|\geq 2^m N_0(1-3\varepsilon) \hat{L}^d-4\tilde{C}_0\sum_{\ell=0}^{m-2}\tilde{\kappa}_{2^{\ell}N_0}\cdot \tilde{\varepsilon}_{2^{\ell}N_0}^{-2}.
    \end{align}
Let
\begin{align}\label{def:G2m-1_N0}
    \mathcal{G}_{2^{m-1}N_0}^{(1)}:
    =&\mathcal{G}_{2^{m-2}N_0}^{(1)}\cap (\mathcal{G}_{2^{m-1}N_0}-(N_1-2^m N_0)d\alpha)\cap (\mathcal{G}_{2^{m-1}N_0}-(N_1+2^{m-1}N_0)d\alpha-\kappa)\cap\\
    &\qquad\cap (\mathcal{G}_{2^m N_0}-(N_1-2^m N_0)d\alpha)\cap (\mathcal{G}_{2^m N_0}-N_1d\alpha-\kappa).
\end{align}
Then for any $\theta\in \mathcal{G}_{2^{m-1}N_0}^{(1)}$, one  has
\begin{align}\label{eq:g2m-1_1-4}
        \log \|g_{2^{m-1}N_0,4}(\theta)\cdots g_{2^{m-1}N_0,1}(\theta)\|\geq 2^{m+1} N_0(1-3\varepsilon) \hat{L}^d-4\tilde{C}_0\sum_{\ell=0}^{m-1}\tilde{\kappa}_{2^{\ell}N_0}\cdot \tilde{\varepsilon}_{2^{\ell}N_0}^{-2}.
\end{align}
Furthermore, $\mathcal{G}_{2^{m-1}N_0}^{(1)}$ satisfies the following two conditions:
\begin{enumerate}
    \item     
    \begin{align}\label{eq:G2m-1_1}
        \mathcal{G}_{2^{m-1}N_0}^{(1)}\subset (\mathcal{G}_{2^{m}N_0}-(N_1-2^{m}N_0)d\alpha)\cap (\mathcal{G}_{2^{m}N_0}-N_1d\alpha-\kappa)
    \end{align}
    \item 
    \begin{align}\label{eq:G2m-1_2}
        \mathrm{mes}(G_{2^{m-1}N_0}^{(1)})\geq 1-5\sum_{\ell=0}^{m-1}e^{-2^{\ell} c_1\varepsilon^2 N_0\hat{L}^d}-2e^{-2^{m}c_1\varepsilon^2N_0\hat{L}^d}.
    \end{align}
\end{enumerate}
\end{lemma}
\begin{proof}
    For any $\theta\in \mathcal{G}_{2^{m-1}N_0}^{(1)}$, Lemma \ref{lem:LDT} implies for $j=1,2,3,4$ that 
    \begin{align}\label{eq:induct_1}
        e^{2^{m-1}N_0(1-\varepsilon)\hat{L}^d}\leq \sigma_1(g_{2^{m-1}N_0,j}(\theta))\leq e^{2^{m-1}N_0(1+\varepsilon)\hat{L}^d},
    \end{align}
    and for $j=1,3$ that 
    \begin{align}\label{eq:induct_2}
        \|g_{2^{m-1}N_0,j+1}(\theta)g_{2^{m-1}N_0,j}(\theta)\|\geq e^{2^m N_0(1-\varepsilon)\hat{L}^d}.
    \end{align}
By assumption \eqref{eq:induct_assume_g1-4}, one has 
\begin{align}\label{eq:induct_3}
    \log \|g_{2^{m-1}N_0,3}(\theta)g_{2^{m-1}N_0,2}(\theta)\|
    =&\log \|g_{2^{m-2}N_0,4}(\theta)\cdots g_{2^{m-2}N_0,1}(\theta)\|\\
    \geq &2^m N_0(1-3\varepsilon)\hat{L}^d-4\tilde{C}_0\sum_{\ell=0}^{m-2}\tilde{\kappa}_{2^{\ell}N_0}\cdot \tilde{\varepsilon}_{2^{\ell}N_0}^{-2}.
\end{align}
In analogy to \eqref{eq:gN_0_j_sigma2}, we infer for $j=1,2,3,4$ that
\begin{align}\label{eq:induct_sigma2}
\sigma_2(g_{2^{m-1}N_0,j}(\theta))\leq e^{2^{m-1}N_0(1-\varepsilon)(\hat{L}^{d-1}+\hat{L}_{d+1})}\cdot e^{2^{m+1}N_0\varepsilon\hat{L}^{d-1}}.
\end{align}
In order to apply Theorem \ref{thm:AP} to $g_{2^{m-1}N_0,4}(\theta),...,g_{2^{m-1}N_0,1}(\theta)$, we first check that the assumptions are satisfied.
By \eqref{eq:induct_1} and \eqref{eq:induct_sigma2}, one has for $j=1,2,3,4$ that 
\begin{align}
    \frac{\sigma_1(g_{2^{m-1}N_0,j}(\theta))}{\sigma_2(g_{2^{m-1}N_0,j}(\theta))}\geq e^{2^m N_0(1-\varepsilon)\hat{L}_d}\cdot e^{-2^{m+1}N_0\varepsilon\hat{L}^{d-1}}=\tilde{\kappa}_{2^{m-1}N_0}^{-1}.
\end{align}
By \eqref{eq:induct_1}, \eqref{eq:induct_2} and \eqref{eq:induct_3} we have that for $j=1,2,3$,
\begin{align}
    \frac{\|g_{2^{m-1}N_0,j+1}(\theta) g_{2^{m-1}N_0,j}(\theta)\|}{\sigma_1(g_{2^{m-1}N_0,j+1}(\theta))\cdot \sigma_1(g_{2^{m-1}N_0,j}(\theta))}\geq e^{-2^{m+2} N_0\varepsilon\hat{L}^d} \cdot e^{-4\tilde{C}_0\sum_{\ell=0}^{m-2}\tilde{\kappa}_{2^{\ell}N_0}\cdot \tilde{\varepsilon}_{2^{\ell}N_0}^{-2}}=\tilde{\varepsilon}_{2^{m-1}N_0}.
\end{align}
Clearly,
\begin{align}
    \tilde{\kappa}_{2^{m-1}N_0}\cdot \tilde{\varepsilon}_{2^{m-1}N_0}^{-2}
    \leq &e^{-2^m N_0(1-\varepsilon)\hat{L}_d}\cdot e^{(2^{m+3}+2^{m+1})N_0\varepsilon\hat{L}^d} \cdot e^{8\tilde{C}_0\sum_{\ell=0}^{m-2}\tilde{\kappa}_{2^{\ell}N_0}\cdot \tilde{\varepsilon}_{2^{\ell}N_0}^{-2}}\\
    \leq &e^{-2^m N_0(1-\varepsilon)\hat{L}_d}\cdot e^{(2^{m+3}+2^{m+1})N_0\varepsilon\hat{L}^d}\cdot e^{8\tilde{C}_0c_0}<c_0,
\end{align}
where we used \eqref{eq:sum_kp_eps^2}. Therefore, Theorem~\ref{thm:AP} implies that 
\begin{align}
    \log \|g_{2^{m-1}N_0,4}(\theta)\cdots g_{2^{m-1}N_0,1}(\theta)\|
    \geq &\sum_{j=1}^3 \log \|g_{2^{m-1}N_0,j+1}(\theta)g_{2^{m-1}N_0,j}(\theta)\|-\sum_{j=2}^3 \log \|g_{2^{m-1}N_0,j}(\theta)\|\\
    &\qquad\qquad-4\tilde{C}_0 \tilde{\kappa}_{2^{m-1}N_0} \cdot \tilde{\varepsilon}_{2^{m-1}N_0}^{-2}\\
    \geq &2^{m+1}N_0(1-3\varepsilon)\hat{L}^d-4\tilde{C}_0\sum_{\ell=0}^{m-1}\tilde{\kappa}_{2^{\ell}N_0} \cdot \tilde{\varepsilon}_{2^{\ell}N_0}^{-2}.
\end{align}
It suffices to prove the measure estimate for $\mathcal{G}_{2^{m-1}N_0}^{(1)}$.
Clearly,  Lemma \ref{lem:LDT} implies that 
\begin{align}
    \mathrm{mes}(\mathcal{G}_{2^{m-1}N_0}^{(1)})
    \geq \mathrm{mes}(\mathcal{G}_{2^{m-2}N_0}^{(1)})-2e^{-2^{m-1}c_1\varepsilon^2 N_0\hat{L}^d}-2e^{-2^m c_1\varepsilon^2 N_0\hat{L}^d}\\
    \geq 1-5\sum_{\ell=0}^{m-1}e^{-2^{\ell}c_1\varepsilon^2 N_0\hat{L}^d}-2e^{-2^mc_1\varepsilon^2 N_0\hat{L}^d},
\end{align}
as claimed.
\end{proof}

We now return to the main argument, i.e., the proof of Lemma~\ref{lem:Mn2}. 
Note that the assumptions on $\mathcal{G}_{2^{m-2}N_0}^{(1)}$ are verified in the base case $m=3$ by Lemma~\ref{lem:AP_2}.
Hence, the inductive scheme applies and therefore \eqref{eq:g2m-1_1-4}, \eqref{eq:G2m-1_1} and \eqref{eq:G2m-1_2} hold for $m=k$.

\medskip

\underline{Step $k+1$}.
Recall that $n_0=N_1-2^k N_0$.
Let 
\begin{align}
    \tilde{g}_{n_0,1}(\theta)=&\wedge^d M_{n_0,E}(\theta)\\
    \tilde{g}_{n_0,4}(\theta)=&\wedge^d M_{n_0,E}(\theta+(2N_1-n_0)d\alpha+\kappa).
\end{align}
By \eqref{eq:g2m-1_1-4} with $m=k$, and recalling the definition of $g_{2^kN_0,j}$ in \eqref{def:g2ell_N0}, we note that 
\begin{align}\label{eq:g2kN0_32}
    \log \|g_{2^kN_0,3}(\theta)g_{2^kN_0,2}(\theta)\|
    =&\log \|g_{2^{k-1}N_0,4}(\theta)\cdots g_{2^{k-1}N_0,1}(\theta)\|\\
    \geq &2^{k+1}N_0(1-3\varepsilon)\hat{L}^d-4\tilde{C}_0c_0,
\end{align}
where we used \eqref{eq:sum_kp_eps^2} in the last inequality.
Recall also that
\begin{align}\label{eq:MN_1=gggg}
    M_{N_1,E}(\theta+N_1d\alpha+\kappa)M_{N_1,E}(\theta)=\tilde{g}_{n_0,4}(\theta)g_{2^k N_0,3}(\theta) g_{2^kN_0,2}(\theta)\tilde{g}_{n_0,1}(\theta).
\end{align}
Let 
\begin{align}
    \mathcal{G}_{2^kN_0}^{(1)}
    :=&\mathcal{G}_{2^{k-1}N_0}^{(1)}\cap \mathcal{G}_{n_0} \cap (\mathcal{G}_{n_0}-(2N_1-n_0)d\alpha-\kappa)\cap\\
    &\qquad\cap \mathcal{G}_{n_0+2^kN_0}\cap (\mathcal{G}_{n_0+2^kN_0}-N_1d\alpha-\kappa)
\end{align}
In this last step, we are going to apply Theorem~\ref{thm:AP} to the four matrices on the right-hand side of \eqref{eq:MN_1=gggg}.
We now verify that the hypotheses are satisfied for $\theta\in \mathcal{G}_{2^kN_0}^{(1)}$.
By Lemma \ref{lem:LDT}, one has for $j=2,3$ that 
\begin{align}\label{eq:g2kN0_1}
    e^{2^k N_0(1-\varepsilon)\hat{L}^d} \leq \sigma_1(g_{2^kN_0,j}(\theta))\leq e^{2^k N_0(1+\varepsilon)\hat{L}^d},
\end{align}
and for $j=1,4$ that
\begin{align}\label{eq:g2kN0_2}
    e^{n_0(1-\varepsilon)\hat{L}^d}\leq \sigma_1(\tilde{g}_{n_0,j}(\theta))\leq e^{n_0(1+\varepsilon)\hat{L}^d}.
\end{align}
We also have
\begin{align}\label{eq:g2kN0_3}
    \|g_{2^kN_0,2}(\theta)\tilde{g}_{n_0,1}(\theta)\|\geq e^{(2^kN_0+n_0)(1-\varepsilon)\hat{L}^d},\\
    \|\tilde{g}_{n_0,4}(\theta)g_{2^kN_0,3}(\theta)\|\geq e^{(2^kN_0+n_0)(1-\varepsilon)\hat{L}^d}.
\end{align}
Similar to \eqref{eq:gN_0_j_sigma2},  for $j=2,3$,
\begin{align}\label{eq:g2kN0_4}
    \sigma_2(g_{2^k N_0,j}(\theta))\leq e^{2^k N_0(1-\varepsilon)(\hat{L}^{d-1}+\hat{L}_{d+1})}\cdot e^{2^{k+2}N_0\varepsilon\hat{L}^{d-1}},
\end{align}
and for $j=1,4$,
\begin{align}\label{eq:g2kN0_5}
    \sigma_2(\tilde{g}_{n_0,j}(\theta))\leq e^{n_0(1-\varepsilon)(\hat{L}^{d-1}+\hat{L}_{d+1})}\cdot e^{2^2 n_0\varepsilon\hat{L}^{d-1}}.
\end{align}
Combining \eqref{eq:g2kN0_32}, \eqref{eq:g2kN0_1}, \eqref{eq:g2kN0_2}, \eqref{eq:g2kN0_3}, \eqref{eq:g2kN0_4} and \eqref{eq:g2kN0_5}, we conclude that 
for $j=2,3$ 
\begin{align}
    \frac{\sigma_1(g_{2^kN_0,j}(\theta))}{\sigma_2(g_{2^kN_0,j}(\theta))}\geq e^{2^{k+1}N_0(1-\varepsilon)\hat{L}_d} e^{-2^{k+2}N_0\varepsilon\hat{L}^{d-1}}=:(\tilde{\kappa}_{2^kN_0}^{(1)})^{-1},
\end{align}
while for $j=1,4$, the following holds:
\begin{align}
    \frac{\sigma_1(\tilde{g}_{n_0,j}(\theta))}{\sigma_2(\tilde{g}_{n_0,j}(\theta))}\geq e^{2n_0(1-\varepsilon)\hat{L}_d} e^{-4n_0\varepsilon\hat{L}^{d-1}}\geq e^{2^{k+1}N_0(1-\varepsilon)\hat{L}_d} e^{-2^{k+2}N_0\varepsilon\hat{L}^{d-1}}=(\tilde{\kappa}_{2^kN_0}^{(1)})^{-1},
\end{align}
where we used $n_0\geq 2^kN_0$ as in \eqref{eq:n0_N0}.
Moreover, 
\begin{align} 
\frac{\|g_{2^kN_0,3}(\theta) g_{2^kN_0,2}(\theta)\|}{\sigma_1(g_{2^kN_0,3}(\theta))\cdot \sigma_1(g_{2^kN_0,2}(\theta))}&\geq e^{-2^{k+3}N_0\varepsilon\hat{L}^d} e^{-4\tilde{C}_0c_0}\geq e^{-4N_1\varepsilon\hat{L}^d}e^{-4\tilde{C}_0c_0}=:\tilde{\varepsilon}_{2^kN_0}^{(1)}\\
\frac{\|g_{2^kN_0,2}(\theta) \tilde{g}_{n_0,1}(\theta)\|}{\sigma_1(g_{2^kN_0,2}(\theta))\cdot \sigma_1(\tilde{g}_{n_0,1}(\theta))}&\geq e^{-2N_1\varepsilon \hat{L}^d}\geq \tilde{\varepsilon}_{2^kN_0}^{(1)}\\
\frac{\|\tilde{g}_{n_0,4}(\theta) g_{2^kN_0,3}(\theta)\|}{\sigma_1(\tilde{g}_{n_0,4}(\theta))\cdot \sigma_1(g_{2^kN_0,3}(\theta))}&\geq e^{-2N_1\varepsilon \hat{L}^d}\geq \tilde{\varepsilon}_{2^kN_0}^{(1)},
\end{align}
where we used $2^{k+1}N_0\leq N_1$ in the first inequality.
Clearly, 
\begin{align}
    \tilde{\kappa}_{2^kN_0}^{(1)}\cdot (\tilde{\varepsilon}_{2^kN_0}^{(1)})^{-2}<c_0.
\end{align}
Therefore, Theorem \ref{thm:AP} implies that 
\begin{align}
&    \log \|\tilde{g}_{n_0,4}(\theta)g_{2^kN_0,3}(\theta)g_{2^kN_0,2}(\theta)\tilde{g}_{n_0,1}(\theta)\|\\
&\qquad\geq \log \|\tilde{g}_{n_0,4}(\theta)g_{2^kN_0,3}(\theta)\|+\log\|g_{2^kN_0,3}(\theta)g_{2^kN_0,2}(\theta)\|+\log\|g_{2^kN_0,2}(\theta)\tilde{g}_{n_0,1}(\theta)\|\\
&\qquad-\log\|g_{2^kN_0,3}(\theta)\|-\log\|g_{2^kN_0,2}(\theta)\|-4\tilde{C}_0 \tilde{\kappa}_{2^kN_0}^{(1)}\cdot (\tilde{\varepsilon}_{2^kN_0}^{(1)})^{-2}\\
&\qquad\geq 2N_1(1-\varepsilon)\hat{L}^d-2^{k+3}N_0\varepsilon\hat{L}^d-8\tilde{C}_0c_0\\
&\qquad\geq 2N_1(1-3\varepsilon)\hat{L}^d-8\tilde{C}_0c_0.
\end{align}
Finally, to estimate the measure, recall $n_0\geq 2^k N_0$ and $N_1\geq 2^{k+1}N_0$ whence 
\begin{align}
    \mathrm{mes}(\mathcal{G}_{2^kN_0}^{(1)})
    \geq &\mathrm{mes}(\mathcal{G}_{2^{k-1}N_0}^{(1)})-2e^{-c_1\varepsilon^2 n_0}-2e^{-c_1\varepsilon^2 N_1}\\
    \geq &1-5\sum_{\ell=0}^{k-1}e^{-2^{\ell}c_1\varepsilon^2 N_0\hat{L}^d}-2e^{-2^k c_1\varepsilon^2 N_0\hat{L}^d}-2e^{-c_1\varepsilon^2 n_0\hat{L}^d}-2e^{-c_1\varepsilon^2 N_1\hat{L}^d}\\
    \geq &1-5\sum_{\ell=0}^{\infty} e^{-2^{\ell}c_1\varepsilon^2N_0\hat{L}^d}>\frac{19}{20},
\end{align}
where we used \eqref{eq:G2m-1_2} with $m=k$ and \eqref{eq:N0<1/100}.
This proves the claimed result.
\end{proof}

As a corollary of Lemma \ref{lem:Mn2}, we now establish 
\begin{corollary}\label{cor:Mn2}
Under the same conditions as in Lemma \ref{lem:Mn2}, for large enough $\kappa_0$-admissible $n$, with $n\geq 8N_0$, and for any $\theta\in (\tilde{\mathcal{B}}^d_{n,E,\varepsilon})^c$, we have that 
    \begin{align}
        |\langle v_{1,n}(\theta)\wedge\cdots \wedge v_{d,n}(\theta), w_{1,n}(\theta)\wedge\cdots \wedge w_{d,n}(\theta)\rangle|=|\det (\langle v_{j,n}(\theta), w_{k,n}(\theta)\rangle)_{1\leq j,k\leq d}|\geq e^{-6n\varepsilon\cdot \hat{L}^d}.
    \end{align}
\end{corollary}
\begin{proof}
Towards a contradiction, suppose that for some $\theta\in (\tilde{\mathcal{B}}^d_{n,E,\varepsilon})^c$ one has
\begin{align}\label{eq:assume_a1-d_small}
&|\langle v_{1,n}(\theta)\wedge\cdots \wedge v_{d,n}(\theta), w_{1,n}(\theta)\wedge\cdots \wedge w_{d,n}(\theta)\rangle|\\
&\qquad\qquad=|\det (\langle v_{j,n}(\theta), w_{k,n}(\theta)\rangle)_{1\leq j,k\leq d}|\leq e^{-6 n\varepsilon \cdot \hat{L}^d}.
\end{align}
We expand
\begin{align}\label{eq:wedge_w=wedge_v}
    w_{1,n}(\theta)\wedge\cdots\wedge w_{d,n}(\theta)=\sum_{1\leq j_1<...<j_d\leq 2d} a_{j_1,...,j_d}(\theta) \cdot v_{j_1,n}(\theta)\wedge\cdots\wedge v_{j_d,n}(\theta),
\end{align}
in which
\begin{align}
    a_{j_1,...,j_d}(\theta)=\langle w_{1,n}(\theta)\wedge\cdots\wedge  w_{d,n}(\theta), v_{j_1,n}(\theta)\wedge\cdots\wedge  v_{j_d,n}(\theta)\rangle.
\end{align}
For arbitrary $1\leq j_1<...<j_d\leq 2d$, there is  the trivial bound
\begin{align}\label{eq:aj1-jd_trivial}
    |a_{j_1,...,j_d}(\theta)|\leq 1.
\end{align}
By assumption \eqref{eq:assume_a1-d_small}, one has
\begin{align}
    |a_{1,...,d}(\theta)|\leq e^{-6n\varepsilon\cdot \hat{L}^d}.
\end{align}
We also require the following two uniform upper bounds.

\begin{lemma}\label{lem:prod_sigma_j_k_neq_d}
For $n$ large enough, and for {\it any } $1\leq k\leq 2d-1$, $k\neq d$, one has that uniformly in $\theta\in\T$,
\begin{align}
    \prod_{\ell=1}^k \sigma_{j_{\ell}}(M_{n,E}(\theta))\leq e^{n\cdot(1+\varepsilon) \hat{L}^{d-1}}.
\end{align}
Here $\varepsilon>0$ is arbitrary but fixed.
\end{lemma}
\begin{proof}
Clearly, by Lemma \ref{lem:upper_semi_cont},
\begin{align}
    \prod_{\ell=1}^k \sigma_{j_{\ell}}(M_{n,E}(\theta))\leq \prod_{j=1}^k \sigma_j(M_{n,E}(\theta))\leq e^{n(1+\varepsilon)\cdot \hat{L}^k}\leq e^{n(1+\varepsilon)\cdot \hat{L}^{d-1}},
\end{align}
where we used that for any $k\neq d$, $\hat{L}^k\leq \hat{L}^{d-1}$ due to $\hat{L}_d>0$ and $\hat{L}_{d+1}<0$.
\end{proof}

\begin{lemma}\label{lem:prod_sigma_j_k=d}
For $(j_1,...,j_d)\neq (1,...,d)$, and for $n$ large enough,
\begin{align}
    \prod_{\ell=1}^d \sigma_{j_{\ell}}(M_{n,E}(\theta))\leq e^{n(1+\varepsilon)\cdot \hat{L}^{d-1}}.
\end{align}
 uniformly in $\theta\in\T$.
\end{lemma}
\begin{proof}
    If $\sigma_{j_d}(M_{n,E}(\theta))<1$, then
    \begin{align}
        \prod_{\ell=1}^d \sigma_{j_{\ell}}(M_{n,E}(\theta))<\prod_{\ell=1}^{d-1}\sigma_{j_{\ell}}(M_{n,E}(\theta))\leq e^{n(1+\varepsilon)\cdot \hat{L}^{d-1}},
    \end{align}
    by Lemma \ref{lem:prod_sigma_j_k_neq_d}. On the other hand, 
    if $\sigma_{j_d}(M_{n,E}(\theta))\geq 1$, then
    \begin{align}
        \prod_{\ell=1}^d \sigma_{j_{\ell}}(M_{n,E}(\theta))\leq \prod_{\ell=1}^{j_d}\sigma_{\ell}(M_{n,E}(\theta))\leq e^{n(1+\varepsilon)\cdot \hat{L}^{d-1}},
    \end{align}
    where we used $j_d>d$ and Lemma \ref{lem:prod_sigma_j_k_neq_d}.
\end{proof}

By Lemmas \ref{lem:upper_semi_cont}, \ref{lem:LDT}, \ref{lem:prod_sigma_j_k=d} and estimates \eqref{eq:aj1-jd_trivial} and \eqref{eq:assume_a1-d_small}, one has  for $n$ large enough,
\begin{align}
    &\|(\wedge^d M_{n,E}^2(\theta)) (v_{1,n}(\theta)\wedge \cdots \wedge v_{d,n}(\theta))\|\\
=&\prod_{j=1}^d \sigma_j(M_{n,E}(\theta)) \|(\wedge^d M_{n,E}(\theta) (w_{1,n}(\theta))\wedge\cdots \wedge w_{d,n}(\theta))\|\\
\leq &e^{n(1+\varepsilon)\cdot \hat{L}^d} \sum_{1\leq j_1<...<j_d\leq 2d} |a_{j_1,...,j_d}(\theta)|\prod_{\ell=1}^d \sigma_{j_{\ell}}(M_{n,E}(\theta))\\
=&e^{n(1+\varepsilon)\cdot \hat{L}^d} \cdot \left(|a_{1,...,d}(\theta)|\cdot \prod_{j=1}^d\sigma_j(M_{n,E}(\theta))+\sum_{(j_1,...,j_d)\neq (1,...,d)} |a_{j_1,...,j_d}(\theta)|\cdot \prod_{\ell=1}^d\sigma_{j_{\ell}}(M_{n,E}(\theta))\right)\\
\leq &e^{n(1+\varepsilon)\cdot \hat{L}^d} \cdot \left(e^{n(1-5\varepsilon)\cdot \hat{L}^d}+C_d \cdot e^{n(1+\varepsilon)\cdot \hat{L}^{d-1}}\right)\\
\leq &e^{2n(1-2\varepsilon)\cdot \hat{L}^d},
\end{align}
where we used \eqref{def:epsilon}.
This leads to a contradiction with $\theta\in (\tilde{\mathcal{B}}^d_{n,E,\varepsilon})^c$.
\end{proof}
Combining Corollary \ref{cor:Mn2} with \eqref{eq:Omega_v_w}, we infer the following result. 

\begin{corollary}\label{cor:Mn2_2}
Under the same conditions as in Lemma \ref{lem:Mn2}, for large enough $\kappa_0$-admissible $n$, with $n\geq 8N_0$, and for any $\theta\in (\tilde{\mathcal{B}}^d_{n,E,\varepsilon})^c\cap (\mathcal{B}_{n,E,\varepsilon}^d)^c$, one has 
\begin{align}
    &|\langle v_{d+1,n}(\theta)\wedge\cdots \wedge v_{2d,n}(\theta), w_{d+1,n}(\theta)\wedge\cdots \wedge w_{2d,n}(\theta)\rangle|\\
    &\qquad\qquad=|\det (\langle v_{d+j,n}(\theta), w_{d+k,n}(\theta)\rangle)_{1\leq j,k\leq d}|\geq e^{-6n\varepsilon\cdot \hat{L}^d}.
\end{align}
\end{corollary}

The goal of the remainder of this section is to prove the following lemma.

\begin{lemma}\label{lem:lower}
Under the same conditions as in Lemma \ref{lem:Mn2}, for large enough $\kappa_0$-admissible $n$, with $n\geq 8N_0$, for any $\theta\in (\tilde{\mathcal{B}}^d_{n,E,\varepsilon})^c\cap (\mathcal{B}_{n,E,\varepsilon}^d)^c$, one has 
    \begin{align}
        |\det (M_{n,E}(\theta)-I_{2d})|\geq e^{n(1-8\varepsilon)\cdot \hat{L}^d}.
    \end{align}
\end{lemma}
\begin{proof}
Let $U_{v,n}(\theta)=(v_{1,n}(\theta),...,v_{2d,n}(\theta))$ and $U_{w,n}(\theta)=(w_{1,n}(\theta),...,w_{2d,n}(\theta))$,
and $$D_n(\theta)=\mathrm{diag}(\sigma_1(M_{n,E}(\theta)),...,\sigma_{2d}(M_{n,E}(\theta)))$$
The singular value decomposition $M_{n,E}(\theta)U_{v,n}(\theta)=U_{w,n}(\theta)D_n(\theta)$ yields
\begin{align}
    \det(M_{n,E}(\theta)-I_{2d})=\det (D_n(\theta)-U_{w,n}^*(\theta)U_{v,n}(\theta)).
\end{align}
Note that
\begin{align}
    U_{w,n}^*(\theta)U_{v,n}(\theta)=\left(\begin{matrix} w_{1,n}^*(\theta)\\ \vdots\\ w_{2d,n}^*(\theta)\end{matrix}\right)\cdot (v_{1,n}(\theta),...,v_{2d,n}(\theta))
    &=
    \left(\langle w_{j,n}(\theta), v_{k,n}(\theta)\rangle\right)_{1\leq j,k\leq 2d}
    =:Q_n(\theta).
\end{align}
Hence, 
\begin{align}\label{eq:D-UU}
&|\det(D_n(\theta)-U_{w,n}^*(\theta)U_{v,n}(\theta))|\\
\leq &|\det(-Q_n(\theta))|\\
&+\sum_{k=1}^{2d}\,\, \sum_{1\leq j_1<...<j_k\leq 2d}\,\, \prod_{\ell=1}^k \sigma_{j_{\ell}}(M_{n,E}(\theta))\cdot \left|\det (Q_n(\theta))_{\{1,...,2d\}\setminus \{j_1,...,j_k\}, \{1,...,2d\}\setminus \{j_1,...,j_k\}}\right|,
\end{align}
recall that for a matrix $M$, $M_{B_1,B_2}$ refers to the submatrix with row numbers in the set $B_1$ and column numbers taken from the set $B_2$.
By Hadamard's inequality,
\begin{align}\label{eq:detQ_0}
|\det Q_n(\theta)|\leq \prod_{\ell=1}^{2d} \|(Q_n(\theta))_{\ell}\|\leq 1,
\end{align}
and
\begin{align}\label{eq:detQ_1}
    \left|\det (Q_n(\theta))_{\{1,...,2d\}\setminus \{j_1,...,j_k\}, \{1,...,2d\}\setminus \{j_1,...,j_k\}}\right|\leq \prod_{\ell\notin \{j_1,..,j_k\}} \| (Q_n(\theta))_{\ell}\|\leq 1
\end{align}
where $(Q_n(\theta))_{\ell}$ refers to the $\ell$-th column of $Q_n(\theta)$. Moreover,   we used the following bound
\begin{align}
    \| (Q_n(\theta))_{\ell}\|\leq \|v_{\ell,n}(\theta)\|=1.
\end{align}
Corollary \ref{cor:Mn2_2} implies that 
\begin{align}\label{eq:detQ_2}
    &\left|\det (Q_n(\theta))_{\{d+1,...,2d\}, \{d+1,...,2d\}}\right|\\
=& \left|\langle w_{d+1,n}(\theta)\wedge\cdots \wedge w_{2d,n}(\theta), v_{d+1,n}(\theta)\wedge \cdots \wedge v_{2d,n}(\theta)\rangle\right|
\geq e^{-6n\varepsilon \cdot \hat{L}^d}.
\end{align}
Combining the estimates \eqref{eq:detQ_0}, \eqref{eq:detQ_1}, \eqref{eq:detQ_2} with \eqref{eq:D-UU}, and using Lemmas \ref{lem:prod_sigma_j_k_neq_d}, \ref{lem:prod_sigma_j_k=d} and Corollary~\ref{cor:Mn2_2}, we infer that  for $\theta\in (\tilde{\mathcal{B}}_{n,E,\varepsilon}^d)^c\cap (\mathcal{B}^d_{n,E,\varepsilon})^c$ 
\begin{align}
    |\det( M_n(\theta)-I_{2d})|
\geq &\prod_{j=1}^d \sigma_j(M_n(\theta)) \cdot e^{-6n\varepsilon \cdot\hat{L}^d}-C_d\, \sup_{k=0,\ldots,2d}\, \sup_{\substack{1\leq j_1<...<j_k\leq 2d\\ (j_1,...,j_k)\neq (1,...,d)}}\, \prod_{\ell=1}^k \sigma_{j_{\ell}}(M_{n,E}(\theta))
\\ 
\geq &e^{n(1-\varepsilon)\cdot \hat{L}^d}\cdot e^{-6n\varepsilon \cdot\hat{L}^d}-C_d\, e^{n(1+\varepsilon)\cdot \hat{L}^{d-1}}\\
\geq &e^{n(1-8\varepsilon)\cdot (\hat{L}^d-\varepsilon)},
\end{align}
where we used \eqref{def:epsilon}.
This is the claimed result.
\end{proof}

Alternatively, Lemma~\ref{lem:lower} can also be proved via Schur complements.
Lemma~\ref{lem:deno} follows immediately from combining Lemmas~\ref{lem:P=M-I} and \ref{lem:lower} with the measure estimates of $\tilde{\mathcal{B}}^d_{n,E,\varepsilon}$ in 
Lemma~\ref{lem:Mn2} and $\mathcal{B}^d_{n,E,\varepsilon}$ in Lemma \ref{lem:LDT}. \qed

\subsection{Proof of Lemma \ref{lem:deno_ARC}}
Let $z=e^{2\pi i\theta}$ and $F_{E,n}(z)=f_{E,n}(\theta)$, and $g_{E,n}(z)=z^{nd} F_{E,n}(z)$. 
It is easy to check that $|g_{E,n}(0)|=1$.
Let $\mathcal{C}$ be the unit circle.
In the following, we use Herman's subharmonic argument:
\begin{align}
    \frac{1}{nd}\int_{\T} \log |f_{E,n}(\theta)|\, \mathrm{d}\theta
    =&\frac{1}{nd}\int_{\mathcal{C}} \log |F_{E,n}(z)|\,\frac{ \mathrm{d}z}{2\pi i z}\\
    =&\frac{1}{nd}\int_{\mathcal{C}} \log |z^{nd} F_{E,n}(z)|\, \frac{\mathrm{d}z}{2\pi i z}\\
    = &\frac{1}{nd}\int_{\mathcal{C}}\log |g_{E,n}(z)|\, \frac{\mathrm{d}z}{2\pi i z}\\
    \geq &\frac{1}{nd} \log |g_{E,n}(0)|=0.
\end{align}
It remains to recall that by Haro-Puig's result, see Theorem \ref{thm:HP}, for $E\in \sum_{\alpha,0}$ in the subcritical regime, 
\begin{align}
    0=L(\alpha, A_E)=\hat{L}^d(\alpha, \hat{A}_E)+\log |\hat{v}_d|.
\end{align}
This proves Lemma \ref{lem:deno_ARC}. \qed

\section{Non-selfadjoint Haro-Puig formula}\label{sec:HP}
\subsection{Proof of Corollary \ref{cor:non-sa-HP}}
We shall write $\hat{L}_j(\alpha, \hat{A}_E^{\varepsilon})$ as $\hat{L}_j(\varepsilon)$ for simplicity.
By Theorem \ref{thm:non-sa-HP} and \eqref{eq:LA_eps=LA_0}  one has that 
\begin{align}
L(\alpha, A_{E,\varepsilon})
=&\sum_{j: \hat{L}_j(\varepsilon)\geq 0} \hat{L}_j(\varepsilon)+\log |\hat{v}_{-d} e^{2\pi d\varepsilon}|\\
=&\sum_{j: \hat{L}_j(0)\geq 2\pi \varepsilon} (\hat{L}_j(0)-2\pi \varepsilon)+\log |\hat{v}_{-d}|+2\pi d\varepsilon\\
=&L(\alpha, A_{E,0})-\sum_{\substack{1\leq j\leq d\\ 0\leq \hat{L}_j(0)<2\pi \varepsilon}} \hat{L}_j(0) +2\pi \varepsilon \cdot \#\{j: 1\leq j\leq d,\, 0\leq \hat{L}_j<2\pi \varepsilon\}\\
=&L(\alpha, A_{E,0})+\sum_{\substack{1\leq j\leq d\\ 0\leq \hat{L}_j(0)<2\pi\varepsilon}} (2\pi d\varepsilon-\hat{L}_j(0))
\end{align}
This can be further be simplified in the form:
\begin{align}
L(\alpha, A_{E,\varepsilon})
=&L(\alpha, A_{E,0})+\int_0^{2\pi \varepsilon} \sum_{\substack{1\leq j\leq d\\ 0\leq \hat{L}_j(0)<2\pi \varepsilon}} \chi_{[\hat{L}_j(0), 2\pi \varepsilon]}(y)\, \mathrm{d}y\\
=&L(\alpha,A_{E,0})+\int_0^{2\pi \varepsilon} \#\{j: 1\leq j\leq d,\, 0\leq \hat{L}_j(0)<y\}\, \mathrm{d}y.
\end{align}
This implies
\begin{align}
\kappa(\alpha, A_{E,\varepsilon})=\#\{j: 1\leq j\leq d,\, 0\leq \hat{L}_j(0)\leq 2\pi d\varepsilon\},
\end{align}
whence, in particular,
\begin{align}
\kappa(\alpha, A_{E,0})=\frac{1}{2}\cdot \#\{1\le j\le 2d\;:\; \hat{L}_j(0)=0\}.
\end{align}
This is the claimed result. \qed

\medskip

The remainder of this section consists of a proof of Theorem~\ref{thm:non-sa-HP}.

\subsection{Preparation}
In the following, let $(\alpha, A)\in \T\times C^{\omega}(\T,\mathrm{Mat}(d,\C))$ be an analytic cocycle.
The following lemma and theorems from \cite{AJS} play an important role in our proof. All fractions $p/q$ are reduced (in lowest terms), and we may assume in addition that $d|q$ (we will need this at some point below). Indeed, if $|\frac{p_n}{q_n}-\alpha|\to 0$ as $n\to\infty$ for an irrational~$\alpha$, then also $|\frac{dp_n+1}{dq_n}-\alpha|\to 0$. No prime dividing $d$ can divide $dp_n+1$, so the denominator remains divisible by~$d$ after removing common factors.

\begin{definition}\cite{AJS}
We say $(\alpha,A)$ is $k$-regular if $t\to L^k (\alpha, A(\cdot+it))$ is an affine function of $t$ in a neighborhood of $0$.
\end{definition}
We will only use the definition of $1$-regular for the cocycle $(\alpha,A_{E,\varepsilon})\in \T\times C^{\omega}(\T, \mathrm{Mat}(2,\C))$.

\begin{lemma}\label{lem:Lpq_uniform_upper}\cite{AJS}*{Lemma 5.1}
If $L^k(\alpha,A)>-\infty$, then for arbitrary small $\varepsilon_0>0$ and $p/q$ sufficiently close to $\alpha$, one has uniformly for all $\theta$ that
\begin{align}
L^k(p/q, A, \theta)\leq L^k(\alpha, A)+\varepsilon_0.
\end{align}
\end{lemma}

\begin{theorem}\label{thm:Lpq_regular_no_dev}\cite{AJS}*{Theorem 5.2}
Assume $L^k(\alpha,A)>-\infty$ and that $(\alpha,A)$ is $k$-regular. Then for arbitrary small $\varepsilon_0>0$ and $p/q$ sufficiently close to $\alpha$, one has uniformly for all $\theta\in \T$,
\begin{align}
|L^k(p/q, A, \theta)-L^k(\alpha, A)|\leq \varepsilon_0.
\end{align}
\end{theorem}

\begin{theorem}\label{thm:conti_L_alpha}\cite{AJS}*{Theorem 1.5}
The functions $(\T, C^{\omega}(\T, \mathrm{Mat}(d,\C)))\ni (\alpha, A)\to L_k(\alpha,A)\in [-\infty, \infty)$ are continuous at any $(\alpha', A')$ with $\alpha'\in \R\setminus \Q$.
\end{theorem}

The case of $L_k=-\infty$ does not arise here. 
Since there are only finitely many $\varepsilon$'s such that $(\alpha, A_{E,\varepsilon})$ is {\it not} $1$-regular, it suffices to prove Theorem \ref{thm:non-sa-HP} for any $\varepsilon$ at which $(\alpha, A_{E,\varepsilon})$ is $1$-regular.
Also, due to \eqref{eq:LA_eps=LA_0}, there are only finitely many values of $\varepsilon$ for which some of the exponents $\hat{L}_j(\alpha, \hat{A}_E^{\varepsilon})$ vanish. Hence it suffices to assume that {\it none} of the $\hat{L}_E^{\varepsilon}(\alpha, \hat{A}_E^{\varepsilon})$ are zero. The identity \eqref{eq:non-sa_HP} for arbitrary $\varepsilon$ will then follow from continuity in $\varepsilon$ by Theorem \ref{thm:conti_L_alpha}.
We will thus make the following assumption below.
\begin{assumption}\label{assume_regular}
Let $\varepsilon$ be such that $(\alpha, A_{E,\varepsilon})$ is $1$-regular.    
\end{assumption}

\subsection{Proof of Theorem \ref{thm:non-sa-HP}}
The goal is to prove the following rational frequency $\alpha=p/q$ version of Theorem \ref{thm:non-sa-HP}.
Let $k_0\in \{1,...,2d\}$ be such that
\begin{align}
\hat{L}_1(\alpha,\hat{A}_{E}^{\varepsilon})\geq\cdots \geq \hat{L}_{k_0}(\alpha,\hat{A}_{E}^{\varepsilon})>0> \hat{L}_{k_0+1}(\alpha,\hat{A}_{E}^{\varepsilon})\geq \cdots \geq \hat{L}_{2d}(\alpha,\hat{A}_{E}^{\varepsilon}),
\end{align}
assuming that $\hat{L}_1(\alpha,\hat{A}_{E}^{\varepsilon})>0$. We allow for  all exponents to be positive here. We will deal with the remaining case $\hat{L}_1(\alpha,\hat{A}_{E}^{\varepsilon})<0$ at the end of this section (the vanishing case can be ignored, see above).

\begin{lemma}\label{lem:non-sa-HP_pq}
For any $\varepsilon$ such that $(\alpha, A_{E}^{\varepsilon})$ is $1$-regular and for arbitrary small $\varepsilon_0$ such that
$$0<\varepsilon_0<\frac{1}{10}\min(\hat{L}_{k_0}(\alpha,\hat{A}_{E}^{\varepsilon}), -\hat{L}_{k_0+1}(\alpha,\hat{A}_{E}^{\varepsilon})),$$ 
the following holds for some constant $C>0${\em:}
\begin{align}
\big|L^1(\alpha,A_{E,\varepsilon})-\hat{L}^{k_0}(\alpha, \hat{A}_{E}^{\varepsilon}) -\log|\hat{v}_{d}e^{2\pi d\varepsilon}| \big|\leq C\varepsilon_0\left( L^{k_0}(\alpha,\hat{A}_{E}^{\varepsilon})+\big|\log|\hat{v}_{-d}e^{2\pi d\varepsilon}|\big|\right).
\end{align}
\end{lemma}
Clearly, Theorem \ref{thm:non-sa-HP} follows from taking $\varepsilon_0\to 0^+$ in Lemma \ref{lem:non-sa-HP_pq}.
To prove Lemma \ref{lem:non-sa-HP_pq}, we consider the following three periodic operators acting on $\ell^2(\Z_q)$: 
\begin{align}\label{def:Hp/q_1}
(H_{p/q,\theta,\varepsilon}u)_{[n]} &=u_{[n+1]}+u_{[n-1]}+v(\theta+i\varepsilon+np/q) u_{[n]},
\end{align}
\begin{align}
\label{def:Hp/q_2}
(\tilde{H}_{p/q,\theta,\varepsilon}u)_{[n]} &=e^{2\pi i\theta} u_{[n+1]}+e^{-2\pi i\theta} u_{[n-1]}+v(np/q+i\varepsilon) u_{[n]},
\end{align}
and
\begin{align}\label{def:tH_p/q}
(\hat{H}_{p/q,\theta,\varepsilon}\phi)_{[n]}=\sum_{|k|=1}^d \hat{v}_k e^{-2\pi k\varepsilon} \phi_{[n-k]}+2\cos(2\pi(\theta+np/q)) \phi_{[n]}.
\end{align}
It is clear that we have 
\begin{fact}\label{fact:Hpq=tHpq_theta=0}
$H_{p/q,0,\varepsilon}=\tilde{H}_{p/q,0,\varepsilon}$.
\end{fact}

We will work with   the Fourier transform on the cyclic group $\Z_q$. We write the characters on this finite Abelian group in the form $\chi_m(n)=e^{-2\pi i mnp/q}$, Haar measure being counting measure. Note that $np$ for $1\le n\le q$ is just a relabeling of the congruence classes modulo~$q$. The Fourier transform is thus defined in the form
\begin{align}\label{def:FT_Zq}
(\mathcal{F}\phi)_{[n]}&=\frac{1}{\sqrt{q}}\sum_{m\in \Z_q} e^{2\pi i mnp/q} \phi_{[m]}, \\
(\mathcal{F}^{-1}u)_{[n]} &=\frac{1}{\sqrt{q}}\sum_{m\in \Z_q} e^{-2\pi i mnp/q} u_{[m]},
\end{align}
and gives a unitary transformation on $\ell^2(\Z_q)$.
Lemma~\ref{lem:tH=hH} establishes Aubry duality for the periodic operators by means of this Fourier transform. 

\begin{lemma}\label{lem:tH=hH}
$
\mathcal{F}\tilde{H}_{p/q,\theta,\varepsilon}\mathcal{F}^{-1}
=\hat{H}_{p/q,\theta,\varepsilon}.
$
\end{lemma}
\begin{proof}
For simplicity, we write $\hat{H}_{p/q,\theta,\varepsilon}$ as $\hat{H}$. Then
\begin{align}
&(\mathcal{F}\hat{H}\mathcal{F}^{-1}u)_{[n]}\\
=&\frac{1}{\sqrt{q}} \sum_{m\in \Z_q} e^{2\pi i mnp/q} (\hat{H}\mathcal{F}^{-1}u)_{[m]}\\
=&\frac{1}{\sqrt{q}} \sum_{m\in \Z_q}e^{2\pi i mnp/q}\left( \sum_{|k|=1}^d \hat{v}_k e^{-2\pi k\varepsilon} (\mathcal{F}^{-1}u)_{[m-k]}+2\cos(2\pi(\theta+mp/q))(\mathcal{F}^{-1}u)_{[m]}\right)\\
=&\frac{1}{q}\sum_{m,\ell\in \Z_q}\sum_{|k|=1}^d e^{2\pi i mnp/q} e^{-2\pi i \ell(m-k)p/q}\, \hat{v}_k\, e^{-2\pi k\varepsilon} u_{[\ell]}\\
&+\frac{1}{q}\sum_{m,\ell\in \Z^q} e^{2\pi imnp/q} e^{-2\pi i m\ell p/q} (e^{2\pi i(\theta+mp/q)}+e^{-2\pi i(\theta+mp/q)})u_{[\ell]}\\
=&\sum_{|k|=1}^d \hat{v}_k e^{-2\pi k\varepsilon} e^{2\pi i nkp/q} u_{[n]}+e^{2\pi i\theta} u_{[n+1]}+e^{-2\pi i\theta} u_{[n-1]}\\
=&v(np/q+i\varepsilon)u_{[n]}+e^{2\pi i\theta} u_{[n+1]}+e^{-2\pi i\theta} u_{[n-1]}\\
=&(\tilde{H}_{p/q,\theta,\varepsilon}u)_{[n]},
\end{align}
as claimed.
\end{proof}

For $E\in \C$, let $D_{p/q,E,\varepsilon}(\theta):=\det(H_{p/q,\theta,\varepsilon}-E)$ and define $\tilde{D}$, $\hat{D}$ analogously for the operators $\tilde{H}$ and $\hat{H}$, respectively. The Aubry duality 
Lemma~\ref{lem:tH=hH} implies the following invariance property.

\begin{corollary}\label{cor:det_tH=hH}
For any $E\in \C$ and any $\theta\in \T$, we have
\begin{align}
\hat{D}_{p/q,E,\varepsilon}(\theta)=\tilde{D}_{p/q,E,\varepsilon}(\theta).
\end{align}
\end{corollary}

By inspection we have

\begin{fact}\label{fact:det_Hpq_periodic}
For any $E\in \C$, $D_{p/q,E,\varepsilon}(\theta)$ is a $(p/q)$-periodic (hence $(1/q)$-periodic) function in $\theta$. Analogous results hold for $\tilde{D}$ and $\hat{D}$.
\end{fact}
The following Chamber's type formula is an easy consequence of Fact \ref{fact:det_Hpq_periodic}.
\begin{lemma}\label{lem:tD_expression}
For any $E\in \C$, we have
\begin{align}
\hat{D}_{p/q,E,\varepsilon}(\theta)=\tilde{D}_{p/q,E,\varepsilon}(\theta)=a(p/q,E,\varepsilon)+2(-1)^{q+1} \cos(2\pi q\theta),
\end{align}
where $a(p/q,E,\varepsilon)$ is independent of $\theta$.
This, combined with Fact \ref{fact:Hpq=tHpq_theta=0} implies that
\begin{align}\label{eq:tD=D_0+cos}
\hat{D}_{p/q,E,\varepsilon}(\theta)=\tilde{D}_{p/q,E,\varepsilon}(\theta)=D_{p/q,E,\varepsilon}(0)+(-1)^{q+1}(2\cos(2\pi q\theta)-2).
\end{align}
\end{lemma}
\begin{proof}
The $\hat{D}=\tilde{D}$ part is due to Corollary \ref{cor:det_tH=hH}.
The term of highest degree (in $e^{2\pi i\theta}$) in $\tilde{D}_{p/q,E,\varepsilon}(\theta)$ equals $(-1)^{q+1}e^{2\pi i q\theta}$, and that of lowest degree is $(-1)^{q+1}e^{-2\pi iq\theta}$.
The lemma now follows from Fact~\ref{fact:det_Hpq_periodic}.
\end{proof}

Combining \eqref{eq:tD=D_0+cos}, Corollary \ref{cor:det_tH=hH} with Jensen's formula implies the following.

\begin{lemma}\label{lem:hD=D}
Denote the two zeros of \[z^2+((-1)^{q+1}D_{p/q,E,\varepsilon}(0)-2)z+1=0\] by $z_1, z_2$, respectively. Then 
for any $E\in \C$,
\begin{align}\label{eq:hD=D}
\int_{\T} \log |\hat{D}_{p/q,E,\varepsilon}(\theta)|\, \mathrm{d}\theta=\log\max(|z_1|, |z_2|, 1).
\end{align}
\end{lemma}

Next, we compare the right-hand side of \eqref{eq:hD=D} to  $qL^1(\alpha, A_{E,\varepsilon})$ for all large~$q$. Note that $A_{E,\varepsilon}\in \mathrm{SL}(2,\C)$ implies that $L^1(\alpha, A_{E,\varepsilon})\ge0$.
\begin{lemma}\label{lem:hD=L1}
If $(\alpha, A_{E,\varepsilon})$ is $1$-regular, for some constant $C>0$, we have
\begin{align}
\left|\frac{1}{q}\int_{\T} \log |\hat{D}_{p/q,E,\varepsilon}(\theta)|\, \mathrm{d}\theta-L^1(\alpha, A_{E,\varepsilon})\right|\leq C\varepsilon_0
\end{align}
provided $q$ is large. 
\end{lemma}

\begin{proof} 
First, simple considerations yield
\begin{align}
\left|\max(|z_1|, |z_2|)-|D_{p/q,E,\varepsilon}(0)|\right|\leq 3.
\end{align}
Hence by Lemma \ref{lem:hD=D}, we have
\begin{align}\label{eq:hD=D0}
    \frac{1}{q}\log \left(\max(|D_{p/q,E,\varepsilon}(0)|-3, 1\right)\leq \frac{1}{q}\int_{\T} \log |\hat{D}_{p/q,E,\varepsilon}(\theta)|\, \mathrm{d}\theta\leq \frac{1}{q}\log (|D_{p/q,E,\varepsilon}(0)|+3).
\end{align}

Next, we relate $D_{p/q,E,\varepsilon}(0)$ to $L^1(\alpha, A_{E,\varepsilon})$.
Let 
$$|\lambda_{1,E,\varepsilon}(\theta)|\geq |\lambda_{1,E,\varepsilon}(\theta)|^{-1}= |\lambda_{2,E,\varepsilon}(\theta)|$$ 
be the eigenvalues of 
$A_{q,E,\varepsilon}(p/q,\theta)$.
Applying \eqref{def:Lpq} to $A=A_{E,\varepsilon}$, we have
\begin{align}\label{eq:L1=lambda1}
L^1(p/q,A_{E,\varepsilon},\theta)=\frac{1}{q}\log |\lambda_{1,E,\varepsilon}(\theta)|.
\end{align}
The following lemma relates $D_{p/q,E,\varepsilon}(\theta)$ to $A_{q,E,\varepsilon}(p/q,\theta)$.
\begin{lemma}\label{lem:D=Aq}
We have for all $\theta\in\mathbb{T}$
\begin{align}
|D_{p/q,E,\varepsilon}(\theta)|=|\det(A_{q,E,\varepsilon}(p/q,\theta)-I_2)|=|\lambda_{1,E,\varepsilon}(\theta)-1|\cdot |\lambda_{2,E,\varepsilon}(\theta)-1|.
\end{align}
\end{lemma}
\begin{proof}
The proof is very similar to that of Lemma \ref{lem:P=M-I}.
As in the previous sections, we perform elementary row operations (which leave the determinant invariant), on the $q\times q$ matrix 
\begin{align}
H_{p/q,\theta,\varepsilon}-E=
\left(\begin{matrix}
v_{q-1} & 1 & & & & &1 \\
1 &v_{q-2} &\ddots & & & &\\
&\ddots &\ddots & & & &\\
\\
& & & & &\ddots &\\
& & & &\ddots &\ddots &1\\
1 & & & & &1 &v_0
\end{matrix}\right)=:
\left(\begin{matrix}
\text{Row}_1\\
\vdots\\
\text{Row}_q
\end{matrix}\right)
\end{align}
where we write $v(\theta+i\varepsilon+jp/q)-E$ as $v_j$ for simplicity.
We also denote $A_{n,E,\varepsilon}(p/q,\theta+jp/q):=A_n(j)$, and 
\begin{align}
A_n(j)=\left(\begin{matrix}
A_n^{11}(j)\, &A_n^{12}(j)\\
A_n^{21}(j) &A_n^{22}(j)
\end{matrix}\right).
\end{align}
Our calculations are based on the following recursive relations: 
\begin{align}
\begin{cases}
A_k^{11}(q-k)=-v_{q-k}A_{k-1}^{11}(q-k+1)+A_{k-1}^{12}(q-k+1)\\
A_k^{12}(q-k)=-A_{k-1}^{11}(q-k+1)\\
A_{1}^{11}(k)=-v_k\\
A_{1}^{12}(k)=-1
\end{cases}
\end{align}
Replacing $\text{Row}_1$ with $\text{Row}_1-v_{q-1}\cdot \text{Row}_2$ yields
\begin{align}
\text{Row}_1^{(1)}
=&(\begin{matrix}
0, &1-v_{q-1}v_{q-2}, &-v_{q-1}, &0, \cdots ,&0 ,&1
\end{matrix})\\
=&(\begin{matrix}
0, &-A_2^{11}(q-2), &-A_2^{12}(q-2), &0, \cdots , & 0 ,&1
\end{matrix})
\end{align}
Replacing $\text{Row}_1^{(1)}$ with $\text{Row}_1^{(1)}+A_2^{11}(q-2)\cdot \text{Row}_3$ yields
\begin{align}
\text{Row}_1^{(2)}
=&(\begin{matrix}
0, & 0, &-A_2^{12}(q-2)+v_{q-3}A_2^{11}(q-2), &A_2^{11}(q-2), & 0, \dots ,& 0, & 1
\end{matrix})\\
=&(\begin{matrix}
0, & 0, &-A_3^{11}(q-3), &-A_3^{12}(q-3), &0, \cdots ,& 0, & 1
\end{matrix})
\end{align}
We repeat this process until
\begin{align}
\text{Row}_1^{(q-2)}
=
(\begin{matrix}
0, \cdots &0, &-A_{q-1}^{11}(1), &1-A_{q-1}^{12}(1)
\end{matrix}).
\end{align}
Next, we carry out an analogous  row reduction on the second row.
Replacing $\text{Row}_2$ with $\text{Row}_2-v_{q-2}\cdot \text{Row}_3$ yields
\begin{align}
\text{Row}_2^{(1)}=
(\begin{matrix}
1, &0, &-A_2^{11}(q-3), &-A_2^{12}(q-3)\, , &0, \cdots  , &0    
\end{matrix}).
\end{align}
We repeat this process until
\begin{align}
\text{Row}_2^{(q-3)}=
(\begin{matrix}
1, &0, \cdots , &0, &-A_{q-2}^{11}(1), &-A_{q-2}^{12}(1))
\end{matrix}.
\end{align}
Note that the first, second and last rows, respectively, of the matrix we started with now read 
\begin{align}
\left(\begin{matrix}
\text{Row}_1^{(q-2)}\\
\text{Row}_2^{(q-3)}\\
\text{Row}_q
\end{matrix}\right)=
\left(\begin{matrix}
0 &0 &\cdots &0 &-A^{11}_{q-1}(1) &1-A^{12}_{q-1}(1)\\
1 &0 &\cdots &0 &-A^{11}_{q-2}(1) &-A^{12}_{q-2}(1)\\
1 &0 &\cdots &0 &1 &v_0
\end{matrix}\right).
\end{align}
We expand the determinant of the transformed matrix $H_{p/q,\theta,\varepsilon}-E$ along the second column using the identities $A_{q-2}^{11}(1)=A_{q-1}^{21}(1)$, $A_{q-2}^{12}(1)=A_{q-1}^{22}(1)$. This  yields
\begin{align}
|D_{p/q,E,\varepsilon}(\theta)|
=&\left|\det\left(\begin{matrix}
0 &-A_{q-1}^{11}(1) &1-A_{q-1}^{12}(1)\\
1 &-A_{q-2}^{11}(1) &-A_{q-2}^{12}(1)\\
1 &1 &v_0
\end{matrix}\right)\right|\\
=&\left|\det\left(\begin{matrix}
0 &-A_{q-1}^{11}(1) &1-A_{q-1}^{12}(1)\\
1 &-A_{q-1}^{21}(1) &-A_{q-1}^{22}(1)\\
1 &1 &v_0
\end{matrix}\right)\right|\\
=&\left|\det\left(\begin{matrix}
0& 0 &-A_{q-1}^{11}(1) &1-A_{q-1}^{12}(1)\\
0& 1 &-A_{q-1}^{21}(1) &-A_{q-1}^{22}(1)\\
1& 0 &0 &-1\\
0& 1 &1 &v_0
\end{matrix}\right)\right|\\
=&\left|\det\left(\begin{matrix}
\left(\begin{matrix}0 &0\\ 0 &1 \end{matrix}\right) 
&-A_{q-1}(1)+\left(\begin{matrix} 0 & 1\\ 0 &0\end{matrix}\right)\\
I_2 &-(A_1(0))^{-1}
\end{matrix}\right)\right|\\
=&\left|\det \left(\left(\begin{matrix}0 & 0 \\ 0 &1\end{matrix}\right)-\left(A_{q-1}(1)-\left(\begin{matrix} 0 & 1\\ 0 &0\end{matrix}\right)\right)A_1(0)\right)\right|\\
=&\left|\det(A_q(0)-I_2)\right|,
\end{align}
as claimed.
\end{proof}

By Assumption \ref{assume_regular}, $(\alpha, A_{E,\varepsilon})$ is 1-regular, hence by Theorem \ref{thm:Lpq_regular_no_dev} and \eqref{eq:L1=lambda1}, for $p/q$ sufficiently close to $\alpha$,
\begin{align}\label{eq:lambda1=L1}
\left|\frac{1}{q}\log |\lambda_{1,E,\varepsilon}(0)|-L^1(\alpha, A_{E,\varepsilon})\right|\leq \varepsilon_0.
\end{align}
This implies, by $\lambda_{1,E,\varepsilon}\lambda_{2,E,\varepsilon}= 1$, that
\begin{align}\label{eq:lambda2=-L1}
\left|\frac{1}{q}\log |\lambda_{2,E,\varepsilon}(0)|+L^1(\alpha, A_{E,\varepsilon})\right|\leq \varepsilon_0.
\end{align}
Hence by Lemma \ref{lem:D=Aq}, we have for all sufficiently large $q$
\begin{equation}\label{eq:D0L1}
\left|\frac{1}{q} \log |D_{p/q,E,\varepsilon}(0)|-L^1(\alpha,A_{E,\varepsilon})\right|\leq C\varepsilon_0.
\end{equation}
Combining the above with \eqref{eq:hD=D0} yields Lemma \ref{lem:hD=L1}.
\end{proof}

Finally it remains to relate $\hat{D}_{p/q,E,\varepsilon}$ to the dual Lyapunov exponents.
As stated, the following argument requires $q$ to be divisible by~$d$ (see above why can assume this).   

\begin{lemma}\label{lem:hD=hAq}
For $q=rd$ with $r\in \N$, we have
\begin{align}
\log |\hat{D}_{p/q,E,\varepsilon}(\theta)|
=&\log |\det(M_{r,E}^{\varepsilon}(p/r,\theta)-I_{2d})|+q\log |\hat{v}_{-d}e^{2\pi d\varepsilon}| \label{eq:hD=M}\\
=&\log |\det(\hat{A}_{q,E}^{\varepsilon}(p/q,\theta)-I_{2d})|+q\log |\hat{v}_{-d}e^{2\pi d\varepsilon}|. \label{eq:hD=A}
\end{align}
\end{lemma}
\begin{proof}
\eqref{eq:hD=A} follows from \eqref{eq:hD=M}, due to \eqref{eq:M=A2}.
The proof of \eqref{eq:hD=M} is similar to those of Lemmas~\ref{lem:D=Aq} and \ref{lem:P=M-I}, and is based on elementary row operations.
Consider the operator
\begin{align}
\hat{H}_{p/q,\theta,\varepsilon}-E=
\left(\begin{matrix}
C_{r-1} & \tilde{B} & & & & &B\\
B &C_{r-2} &\ddots & & & &\\
   &\ddots &\ddots & & & &\\
\\
 & & & &\ddots &\ddots &\\
 & & & &\ddots &C_1 &\tilde{B}\\
\tilde{B} & & & & &B &C_0
\end{matrix}\right)
=:
\left(\begin{matrix}
\text{Row}_1\\
\vdots\\
\text{Row}_r
\end{matrix}\right),
\end{align}
where $C_j=C_{\varepsilon}(\theta+jp/r)-E$ and $B=B_{\varepsilon}$, $\tilde{B}=\tilde{B}_{\varepsilon}$. Denote $M_{r,E}^{\varepsilon}(p/r,\theta+jp/r)=:M_r(j)$, and
\begin{align}
M_{r}(j)=
\left(\begin{matrix}
M_r^{11}(j) & M_r^{12}(j)\\
M_r^{21}(j) & M_r^{22}(j)
\end{matrix}\right).
\end{align}
We will use the following recursive relations repeatedly:
\begin{align}
\begin{cases}
M_r^{11}(r-k)=-M_{r-1}^{11}(r-k+1)C_{r-k}B^{-1}+M_{r-1}^{12}(r-k+1)B^{-1}\\
M_r^{12}(r-k)=-M_{r-1}^{11}(r-k+1)\tilde{B}\\
M_1^{11}(k)=-C_k B^{-1}\\
M_1^{12}(k)=-\tilde{B}
\end{cases}
\end{align}
Replacing $\text{Row}_1$ with $\text{Row}_1-C_{r-1}B^{-1}\cdot \text{Row}_2$ yields
\begin{align}
\text{Row}_1^{(1)}=
&(\begin{matrix}
0, &\tilde{B}-C_{r-1}B^{-1}C_{r-2}, &-C_{r-1}B^{-1}\tilde{B}, &0, \cdots, &0, &B
\end{matrix})\\
=&(\begin{matrix}
0, &-M_2^{11}(r-2)B, &-M_2^{12}(r-2), &0, \cdots, &0, &B).
\end{matrix}
\end{align}
Replacing $\text{Row}_1^{(1)}$ with $\text{Row}_1^{(1)}+M_2^{11}(r-2)\cdot \text{Row}_3$ yields
\begin{align}
\text{Row}_1^{(2)}=
&(\begin{matrix}
0, &0, &-M_2^{12}(r-2)+M_2^{11}(r-2)C_{r-3}, &M_2^{11}(r-2)\tilde{B}, &0, \cdots , &0, & B
\end{matrix})\\
=&(\begin{matrix}
0, &0, &-M_3^{11}(r-3)B, &-M_3^{12}(r-3), &0, \cdots, &0, & B
\end{matrix}).
\end{align}
Repeating this process  we arrive at
\begin{align}
\text{Row}_1^{(r-2)}=
(\begin{matrix}
0, \cdots, &0, & -M_{r-1}^{11}(1)B, &B-M_{r-1}^{12}(1)
\end{matrix}).
\end{align}
Next we replace $\text{Row}_2$ with $\text{Row}_2-C_{r-2}B^{-1}\cdot \text{Row}_3$, which yields
\begin{align}
\text{Row}_2^{(1)}=
&(\begin{matrix}
B, &0, &\tilde{B}-C_{r-2}B^{-1}C_{r-3}, &-C_{r-2}B^{-1}\tilde{B}, &0, \cdots, &0
\end{matrix})\\
=&(\begin{matrix}
B, &0, &-M_2^{11}(r-3)B, &-M_2^{12}(r-3), &0, \cdots, &0
\end{matrix})
\end{align}
Repeating this process we obtain
\begin{align}
\text{Row}_2^{(r-3)}
=(\begin{matrix}
B, &0, \cdots &0, &-M_{r-2}^{11}(1)B, &-M_{r-2}^{12}(1)
\end{matrix})
\end{align}
The first, second and last rows are
\begin{align}
\left(\begin{matrix}
\text{Row}_1^{(q-2)}\\
\text{Row}_2^{(q-3)}\\
\text{Row}_{r}
\end{matrix}\right)
=
\left(\begin{matrix}
0 &0 &\cdots &0 &-M_{r-1}^{11}(1)B &B-M_{r-1}^{12}(1)\\
B &0 &\cdots &0 &-M_{r-2}^{11}(1)B &-M_{r-2}^{12}(1)\\
\tilde{B} &0 &\cdots &0 &B &C_0
\end{matrix}\right)
\end{align}
Expanding determinants and using that $M_{r-2}^{11}(1)=BM_{r-1}^{21}(1)$, $M_{r-2}^{12}(1)=BM_{r-1}^{22}(1)$, we have
\begin{align}
|\hat{D}_{p/q,E,\varepsilon}(\theta)|
=&|\det B|^{r-2}\cdot 
\left|\det\left(\begin{matrix}
0 &-M_{r-1}^{11}(1)B &B-M_{r-1}^{12}(1)\\
B &-M_{r-2}^{11}(1)B &-M_{r-2}^{12}(1)\\
\tilde{B} &B &C_0
\end{matrix}\right)\right|\\
=&|\det B|^{r-2}\cdot 
\left|\det\left(\begin{matrix}
0 &-M_{r-1}^{11}(1)B &B-M_{r-1}^{12}(1)\\
B &-BM_{r-1}^{21}(1)B &-BM_{r-1}^{22}(1)\\
\tilde{B} &B &C_0
\end{matrix}\right)\right|\\
=&|\det B|^{r-1}\cdot |\det \tilde{B}|\cdot 
\left|\det\left(\begin{matrix}
0 &-M_{r-1}^{11}(1) &B-M_{r-1}^{12}(1)\\
B &-BM_{r-1}^{21}(1) &-BM_{r-1}^{22}(1)\\
I_2 &\tilde{B}^{-1} &\tilde{B}^{-1}C_0
\end{matrix}\right)\right|\\
=&|\det B|^{r-1}\cdot |\det \tilde{B}|\cdot 
\left|\det\left(\begin{matrix}
0 &0 &-M_{r-1}^{11}(1) &B-M_{r-1}^{12}(1)\\
0 &I_d &-M_{r-1}^{21}(1) &-M_{r-1}^{22}(1)\\
I_d &0 &0 &-B\\
0& I_d &\tilde{B}^{-1} &\tilde{B}^{-1}C_0
\end{matrix}\right)\right|\\
=&|\det B|^{r-1}\cdot |\det \tilde{B}|\cdot 
\left|\det\left(\begin{matrix}
\left(\begin{matrix}0 &0\\ 0 &I_d\end{matrix}\right)
&-M_{r-1}(1)+\left(\begin{matrix} 0 &B\\ 0 &0\end{matrix}\right)\\
I_{2d} &-M_1(0)^{-1}
\end{matrix}\right)\right|\\
=&|\det B|^{r-1}\cdot |\det \tilde{B}|\cdot |\det M_1(0)^{-1}|\cdot 
\left|\det\left(\left(\begin{matrix}
0 & 0\\
0 &I_d
\end{matrix}\right)-M_r(0)+\left(\begin{matrix} 0 &B \\ 0 &0\end{matrix}\right)M_1(0)\right)\right|\\
=&|\det B|^r |\det(M_r(0)-I_{2d})|.
\end{align}
This implies the claimed result by noting $|\det B|=|\hat{v}_{-d}e^{2\pi d\varepsilon}|^d$.
\end{proof}

Taking Lemmas \ref{lem:hD=L1} and \ref{lem:hD=hAq} into account, Lemma \ref{lem:non-sa-HP_pq} will follow from the following.
\begin{lemma}\label{lem:Aq_L^k0}
Under the assumptions of Lemma \ref{lem:non-sa-HP_pq} one has 
\begin{align}
&\left|\frac{1}{q}\int_{\T} \log |\det(\hat{A}_{q,E}^{\varepsilon}(p/q,\theta)-I_{2d})|\, \mathrm{d}\theta-\hat{L}^{k_0}(\alpha,\hat{A}_{E}^{\varepsilon})\right|\\
&\qquad\leq C\varepsilon_0\left(q^{-1}\varepsilon_0^{-\frac12}+\hat{L}^{k_0}(\alpha, \hat{A}_{E}^{\varepsilon})+\left|\log|\hat{v}_{-d}e^{2\pi d\varepsilon}|\right|\right)
\end{align}
for all sufficiently large $q$.
\end{lemma}

\begin{proof} 
Let 
$$|\hat{\lambda}_{1,E,\varepsilon}(\theta)|\geq \cdots \geq |\hat{\lambda}_{2d,E,\varepsilon}(\theta)|$$
be the eigenvalues of $\hat{A}_{q,E}^{\varepsilon}(p/q,\theta)$.
We need the following elementary estimate.
\begin{lemma}\label{lem:spe_rad_eig}
For any $1\leq m\leq 2d$, we have
\begin{align}
\hat{L}^m(p/q, \hat{A}_{E}^{\varepsilon},\theta)=\frac{1}{q}\sum_{j=1}^m \log |\hat{\lambda}_{j,E,\varepsilon}(\theta)|.
\end{align}
\end{lemma}
\begin{proof}
Note by \eqref{eq:L1=lambda1} the claimed is true for $m=1$. Let 
\begin{align}
\hat{A}_{q,E}^{\varepsilon}(p/q,\theta)=U_q^{-1}(\theta) B_q(\theta) U_q(\theta),
\end{align}
be the Jordan form of $\hat{A}_{q,E}^{\varepsilon}(p/q,\theta)$, where we omit the dependence of $U$ and $B$ on $E,\varepsilon$ for simplicity.
Clearly $U_q, U_q^{-1}, B_q$ are all continuous in $\theta$.
Thus there exists some constants $C_m>0$ such that
\begin{align}\label{eq:wedge_Uq<Cm}
\sup_{\theta\in \T} \max(\|\wedge^m U_q^{-1}(\theta)\|, \|\wedge^m U_q(\theta)\|)\leq C_m.
\end{align}
It is also easy to see that
\begin{align}\label{eq:wedge_Bq}
\|\wedge^m (B_q(\theta))^k\|=\prod_{j=1}^m |\hat{\lambda}_{j,E,\varepsilon}(\theta)|^k
\end{align}
Combining \eqref{eq:wedge_Bq} with \eqref{eq:wedge_Uq<Cm}, we have for any $k\geq 1$,
\begin{align}
&\frac{1}{k}\log \|\wedge^m (\hat{A}_{q,E}^{\varepsilon}(p/q,\theta))^k\|\\
\leq &\frac{1}{k}\left(\log \|\wedge^m U_q^{-1}(\theta)\|
+\log \|\wedge^m (B_q(\theta))^k\|+\log \|(\wedge^m U_q(\theta))^k\|\right)\\
\leq &\sum_{j=1}^m \log |\hat{\lambda}_{j,E,\varepsilon}(\theta)|+\frac{2C_m}{k},
\end{align}
which implies 
\begin{align}
\hat{L}^m(p/q,\hat{A}_{E}^{\varepsilon},\theta)\leq \frac{1}{q}\sum_{j=1}^m \log |\hat{\lambda}_{j,E,\varepsilon}(\theta)|.
\end{align}
This other direction of the inequality can be proved similarly.
\end{proof}

\begin{lemma}\label{lem:prod_lambda_not_largest}
If $p/q$ is sufficiently close to $\alpha$, then for any $1\leq j_1<j_2<...<j_m\leq 2d$ such that $(j_1,...,j_{m})\neq (1,...,k_0)$, one has
\begin{align}\label{eq:prod_lambda_1}
\sup_{\theta\in \T} \frac{1}{q} \sum_{\ell=1}^m \log 
|\hat{\lambda}_{j_{\ell},E,\varepsilon}(\theta)|\leq 
\max(\hat{L}^{k_0-1}(\alpha, \hat{A}_{E}^{\varepsilon}), \hat{L}^{k_0+1}(\alpha, \hat{A}_{E}^{\varepsilon}))+\varepsilon_0,
\end{align}
and
\begin{align}\label{eq:prod_lambda_2}
\sup_{\theta\in \T} \frac{1}{q}\sum_{j=1}^{k_0}\log |\hat{\lambda}_{j,E,\varepsilon}(\theta)|\leq \hat{L}^{k_0}(\alpha,\hat{A}_{E}^{\varepsilon})+\varepsilon_0.    
\end{align}
\end{lemma}
\begin{proof}
{\em Case 1.} If $m\neq k_0$, we estimate by Lemmas~\ref{lem:Lpq_uniform_upper} and \ref{lem:spe_rad_eig} that uniformly in $\theta$
\begin{align}\label{eq:m_neq_k0_1}
 \frac{1}{q} \sum_{\ell=1}^m \log|\hat{\lambda}_{j_{\ell},E,\varepsilon}(\theta)|
\leq \frac{1}{q} \sum_{j=1}^m \log |\hat{\lambda}_{j,E,\varepsilon}(\theta)|=\hat{L}^m(p/q,\hat{A}_{E}^{\varepsilon},\theta)\leq \hat{L}^m(\alpha,\hat{A}_{E}^{\varepsilon})+\varepsilon_0.
\end{align}

{\em Case 2.1.} If $m=k_0$ and $|\hat{\lambda}_{j_{k_0},E,\varepsilon}(\theta)|\leq 1$, we estimate
\begin{align}
\frac{1}{q}\sum_{\ell=1}^{k_0} \log |\hat{\lambda}_{j_{\ell},E,\varepsilon}(\theta)|
\leq \frac{1}{q}\sum_{\ell=1}^{k_0-1} \log |\hat{\lambda}_{j_{\ell},E,\varepsilon}(\theta)|
\leq \hat{L}^{k_0-1}(\alpha, \hat{A}_{E}^{\varepsilon})+\varepsilon_0,
\end{align}
where we used the bound in Case 1.

{\em Case 2.2.} If $m=k_0$ and $|\hat{\lambda}_{j_{k_0},E,\varepsilon}(\theta)|>1$. We have $|\hat{\lambda}_{j,E,\varepsilon}(\theta)|>1$ for any $1\leq j\leq k_0$. Since $(j_1,...,j_{k_0})\neq (1,...,k_0)$, we have $j_{k_0}\geq k_0+1$. 
We estimate, in this case,
\begin{align}
\frac{1}{q}\sum_{\ell=1}^{k_0}\log |\hat{\lambda}_{j_{\ell},E,\varepsilon}(\theta)|
\leq \frac{1}{q}\sum_{j=1}^{j_{k_0}} \log|\hat{\lambda}_{j,E,\varepsilon}(\theta)|
\leq \hat{L}^{j_{k_0}}(\alpha, \hat{A}_{E}^{\varepsilon})+\varepsilon_0,
\end{align}
where we used again the estimate in Case 1.
Finally, note that due to the definition of $k_0$, we have
\begin{align}
\max_{\substack{m=1,\ldots,2d \\ m\neq k_0}} \hat{L}^m(\alpha, \hat{A}_{E}^{\varepsilon})\leq \max(\hat{L}^{k_0-1}(\alpha,\hat{A}_{E}^{\varepsilon}), \hat{L}^{k_0+1}(\alpha,\hat{A}_{E}^{\varepsilon})),
\end{align}
which proves \eqref{eq:prod_lambda_1}. Finally, 
\eqref{eq:prod_lambda_2} follows directly from Lemmas \ref{lem:spe_rad_eig} and \ref{lem:Lpq_uniform_upper}.
\end{proof}

Let 
\begin{align}\label{def:delta_L_dif}
\delta
:=&\frac{1}{2}\left(\hat{L}^{k_0}(\alpha,\hat{A}_{E}^{\varepsilon})-\max(\hat{L}^{k_0-1}(\alpha,\hat{A}_{E}^{\varepsilon}), \hat{L}^{k_0+1}(\alpha,\hat{A}_{E}^{\varepsilon}))\right)\\
=&\frac{1}{2}\min(\hat{L}_{k_0}(\alpha,\hat{A}_{E}^{\varepsilon}), -\hat{L}_{k_0+1}(\alpha,\hat{A}_{E}^{\varepsilon}))>5\varepsilon_0, 
\end{align}
and introduce the {\em good set} for fixed $q$
\begin{align}\label{def:G_Lpq}
\mathcal{G}_q=\mathcal{G}:=\left\{\theta\in \T: \frac{1}{q}\sum_{j=1}^{k_0} \log |\hat{\lambda}_{j,E,\varepsilon}(\theta)|>\hat{L}^{k_0}(\alpha,\hat{A}_{E}^{\varepsilon})-\delta\right\}.
\end{align}

\begin{lemma}\label{lem:mes_G}
The Lebesgue measure of $\mathcal{G}$ satisfies
\begin{align}
|\mathcal{G}|\geq 1-\frac{2\varepsilon_0}{\delta+\varepsilon_0}.
\end{align}
\end{lemma}
\begin{proof}
By Lemma \ref{lem:prod_lambda_not_largest}, we have
\begin{align}\label{eq:spe_rad_upp}
\sup_{\theta} \frac{1}{q}\sum_{j=1}^{k_0} \log |\hat{\lambda}_{j,E,\varepsilon}(\theta)|\leq \hat{L}^{k_0}(\alpha,\hat{A}_{E}^{\varepsilon})+\varepsilon_0.
\end{align}
Hence by the definition of $\mathcal{G}$ in \eqref{def:G_Lpq} and \eqref{eq:spe_rad_upp}, we have
\begin{align}
|\mathcal{G}|\cdot (\hat{L}^{k_0}(\alpha,\hat{A}_{E}^{\varepsilon})+\varepsilon_0)+(1-|\mathcal{G}|)\cdot (\hat{L}^{k_0}(\alpha,\hat{A}_{E}^{\varepsilon})-\delta) 
\geq &\int_{\T}\hat{L}^{k_0}(p/q,\hat{A}_{E}^{\varepsilon},\theta)\, \mathrm{d}\theta\\
=&\hat{L}^{k_0}(p/q,\hat{A}_{E}^{\varepsilon})>
\hat{L}^{k_0}(\alpha,\hat{A}_{E}^{\varepsilon})-\varepsilon_0.
\end{align}
This implies
\begin{align}\label{eq:mea_G_varepsilon0}
|\mathcal{G}|\geq 1-\frac{2\varepsilon_0}{\delta+\varepsilon_0},
\end{align}
as desired.
\end{proof}

On the good set we have the following approximation result. 

\begin{lemma}\label{lem:lam_comp}
If $p/q$ is sufficiently close to $\alpha$, then
\begin{align}\label{eq:intG_dif}
\frac{1}{q} \int_{\mathcal{G}} \sum_{j=1}^{2d} \log|\hat{\lambda}_{j,E,\varepsilon}(\theta)-1|\, \mathrm{d}\theta
=\frac{1}{q}
\int_{\mathcal{G}}\sum_{j=1}^{k_0}\log |\hat{\lambda}_{j,E,\varepsilon}(\theta)|\, \mathrm{d}\theta+O(1/q).
\end{align}
\end{lemma}
\begin{proof}
    For each $\theta\in \mathcal{G}$, for $q$ large, we estimate by Lemma \ref{lem:prod_lambda_not_largest} that
\begin{align}\label{eq:prod_lambda-1_1}
\prod_{j=1}^{2d}|\hat{\lambda}_{j,E,\varepsilon}(\theta)-1|\geq 
&\prod_{j=1}^{k_0} |\hat{\lambda}_{j,E,\varepsilon}(\theta)|-\sum_{\substack{j_1<j_2<...<j_m\\ (j_1,...,j_m)\neq (1,...,k_0)}} \prod_{\ell=1}^m |\hat{\lambda}_{j_{\ell},E,\varepsilon}(\theta)|\\
\geq &\frac{1}{2}\prod_{j=1}^{k_0} |\hat{\lambda}_{j,E,\varepsilon}(\theta)|.
\end{align}
Analogously, for all $\theta\in \mathcal{G}$ we have the upper bound
\begin{align}\label{eq:prod_lambda-1_2}
\prod_{j=1}^{2d}|\hat{\lambda}_{j,E,\varepsilon}(\theta)-1|\leq 2\prod_{j=1}^{k_0} |\hat{\lambda}_{j,E,\varepsilon}(\theta)|.
\end{align}
The result follows by \eqref{eq:prod_lambda-1_1} and \eqref{eq:prod_lambda-1_2}. 
\end{proof}

Next, we estimate 
\begin{lemma}\label{lem:Gc_1}
Under the assumptions of Lemma~\ref{lem:non-sa-HP_pq}, 
\begin{align}
\left|\frac{1}{q}\int_{\mathcal{G}^c}\log |\det(\hat{A}_{q,E}^{\varepsilon}(p/q,\theta)-I_{2d})|\, \mathrm{d}\theta\right|\leq C\varepsilon_0
\left(q^{-1}\varepsilon_0^{-\frac12}+\hat{L}^{k_0}(\alpha, \hat{A}_{E}^{\varepsilon})+\big|\log |\hat{v}_{d}e^{2\pi d\varepsilon}|\big|\right).
\end{align}
\end{lemma}
\begin{proof}
By Lemma \ref{lem:prod_lambda_not_largest}, we have
\begin{align}
|\det(\hat{A}_{q,E}^{\varepsilon}(p/q,\theta)-I_{2d})|\leq \sum_{1\leq j_1<...<j_m\leq 2d}\,\, \prod_{\ell=1}^m |\hat{\lambda}_{j_{\ell},E,\varepsilon}(\theta)|\leq C_d \exp(q\cdot (\hat{L}^{k_0}(\alpha, \hat{A}_{E}^{\varepsilon})+\varepsilon_0)).
\end{align}
This implies, after combining with \eqref{eq:mea_G_varepsilon0},
\begin{align}\label{eq:int_Gc_upp}
\frac{1}{q}\int_{\mathcal{G}^c}\log |\det(\hat{A}_{q,E}^{\varepsilon}(p/q,\theta)-I_{2d})|\, \mathrm{d}\theta
\leq &|\mathcal{G}^c|\cdot(\hat{L}^{k_0}(\alpha,\hat{A}_{E}^{\varepsilon})+2\varepsilon_0)\\
\leq &C\varepsilon_0(\hat{L}^{k_0}(\alpha,\hat{A}_{E}^{\varepsilon})+2\varepsilon_0).
\end{align}
By Lemmas \ref{lem:hD=hAq}, \ref{lem:tD_expression},  and \eqref{eq:mea_G_varepsilon0}, we have 
\begin{align}
&\frac{1}{q}\int_{\mathcal{G}^c}\log |\det(\hat{A}_{q,E}^{\varepsilon}(p/q,\theta)-I_{2d})|\, \mathrm{d}\theta\\
=&\frac{1}{q}\int_{\mathcal{G}^c}\left(\log |\hat{D}_{p/q,E,\varepsilon}(\theta)|-q\log |\hat{v}_{-d}e^{2\pi d\varepsilon}|\right)\, \mathrm{d}\theta\\
=&\frac{1}{q}\int_{\mathcal{G}^c}\log |D_{p/q,E,\varepsilon}(0)+2(-1)^{q+1}\cos(2\pi q\theta)|\, \mathrm{d}\theta-|\mathcal{G}^c| \cdot \log |\hat{v}_{-d}e^{2\pi d\varepsilon}|\\
\geq &-C\varepsilon_0 \left|\log |\hat{v}_{-d}e^{2\pi d\varepsilon}|\right| 
\end{align}
provided $|D_{p/q,E,\varepsilon}(0)|\ge 3$. On the other hand, if $|D_{p/q,E,\varepsilon}(0)|\le 3$, then 
\begin{align}
&\frac{1}{q}\int_{\mathcal{G}^c}\log |D_{p/q,E,\varepsilon}(0)+2(-1)^{q+1}\cos(2\pi q\theta)|\, \mathrm{d}\theta \ge  \\
&-\frac{1}{q}|\mathcal{G}^c|^{\frac12} \Big( \int_{\mathbb{T}} \big(\log |D_{p/q,E,\varepsilon}(0)+2(-1)^{q+1}\cos(2\pi \theta)|\big)^2\, \mathrm{d}\theta \Big)^{\frac12} \ge -Cq^{-1}\sqrt{\varepsilon_0},  
\end{align}
as desired.
\end{proof}
We also have
\begin{lemma}\label{lem:Gc_2}
\begin{align}
\left|\frac{1}{q}\int_{\mathcal{G}^c} \sum_{j=1}^{k_0} \log |\hat{\lambda}_{j,E,\varepsilon}(\theta)|\, \mathrm{d}\theta-|\mathcal{G}^c|\cdot \hat{L}^{k_0}(\alpha,\hat{A}_{E}^{\varepsilon})\right|\leq 2\varepsilon_0.
\end{align}
\end{lemma}
\begin{proof}
For simplicity, we write $g(\theta):=\frac{1}{q}\log \sum_{j=1}^{k_0} |\hat{\lambda}_{j,E,\varepsilon}(\theta)|$ and $\hat{L}^{k_0}:=\hat{L}^{k_0}(\alpha,\hat{A}_{E}^{\varepsilon})$.
By Lemmas~\ref{lem:spe_rad_eig} and~\ref{lem:prod_lambda_not_largest}, we have
\begin{align}
\hat{L}^{k_0}-\varepsilon_0\leq \int_{\T} g(\theta)\, \mathrm{d}\theta
\leq \int_{\mathcal{G}^c}g(\theta)\, \mathrm{d}\theta+|\mathcal{G}|\cdot (\hat{L}^{k_0}+\varepsilon_0),
\end{align}
which implies
\begin{align}\label{eq:Gc_1}
\int_{\mathcal{G}^c} (g(\theta)-\hat{L}^{k_0})\, \mathrm{d}\theta\geq -2\varepsilon_0.
\end{align}
We also have by \eqref{eq:mea_G_varepsilon0} and Lemma \ref{lem:prod_lambda_not_largest} that
\begin{align}
\int_{\mathcal{G}^c}(g(\theta)-\hat{L}^{k_0})\, \mathrm{d}\theta\leq \varepsilon_0 \cdot |\mathcal{G}^c|\leq \frac{2\varepsilon_0^2}{\delta+\varepsilon_0}.
\end{align}
This proves the claimed result.
\end{proof}
Finally, combining \eqref{eq:intG_dif} with Lemmas \ref{lem:Gc_1} and \ref{lem:Gc_2}, we have
\begin{align}
&\left|\frac{1}{q}\int_{\T} \log |\det(\hat{A}_{q,E}^{\varepsilon}(p/q,\theta)-I_{2d})|\, \mathrm{d}\theta-\hat{L}^{k_0}(\alpha,\hat{A}_{E}^{\varepsilon})\right|\\
\leq &\left|\frac{1}{q}\int_{\mathcal{G}}\sum_{j=1}^{2d}\log |\hat{\lambda}_{j,E,\varepsilon}(\theta)-1|\, \mathrm{d}\theta-\frac{1}{q}\int_{\mathcal{G}}\sum_{j=1}^{k_0}\log |\hat{\lambda}_{j,E,\varepsilon}(\theta)|\, \mathrm{d}\theta\right|\\
&+\left|\frac{1}{q}\int_{\mathcal{G}^c}\log |\det(\hat{A}_{q,E}^{\varepsilon}(p/q,\theta)-I_{2d})|\, \mathrm{d}\theta\right|+\left|\frac{1}{q}\int_{\T}\sum_{j=1}^{k_0}\log |\hat{\lambda}_{j,E,\varepsilon}(\theta)|\, \mathrm{d}\theta-\hat{L}^{k_0}(\alpha,\hat{A}_{E,\varepsilon})\right|\\
&+\left|\frac{1}{q}\int_{\mathcal{G}^c}\sum_{j=1}^{k_0}\log |\hat{\lambda}_{j,E,\varepsilon}(\theta)|\, \mathrm{d}\theta-|\mathcal{G}^c|\cdot \hat{L}^{k_0}(\alpha,\hat{A}_{E}^{\varepsilon})\right|
+|\mathcal{G}^c|\cdot \hat{L}^{k_0}(\alpha,\hat{A}_{E}^{\varepsilon})\\
\leq &C\varepsilon_0(q^{-1}\varepsilon_0^{-\frac12}+\hat{L}^{k_0}(\alpha, \hat{A}_{E}^{\varepsilon})+\big|\log|\hat{v}_{-d}e^{2\pi d\varepsilon}|\big|).
\end{align}
Combining Lemmas \ref{lem:hD=L1}, \ref{lem:hD=hAq}, \ref{lem:Aq_L^k0}, we have proved Lemma \ref{lem:non-sa-HP_pq}.
\end{proof}

It remains to deal with the case $\hat{L}^{1}(\alpha, \hat{A}_{E}^{\varepsilon})<0$. By Lemmas~\ref{lem:Lpq_uniform_upper} and~\ref{lem:Aq_L^k0}, there exists~$\delta>0$ so that uniformly in~$\theta$, and provided $p/q$ is close to~$\alpha$, one has
\[
-\delta\ge \frac{1}{q}\sum_{j=1}^m \log |\hat{\lambda}_{j,E,\varepsilon}(\theta)|
\]
 for all $1\le m\le 2d$. 
Hence, for all $1\le j_1\le \ldots \le j_m\le 2d$ we have
\[
\sup_\theta\prod_{\ell=1}^m |\hat{\lambda}_{j_{\ell},E,\varepsilon}(\theta)|\le e^{-\delta q}
\]
Therefore, 
\[
\prod_{j=1}^{2d}|\hat{\lambda}_{j,E,\varepsilon}(\theta)-1| =1 + O_d(e^{-\delta q})
\]
and 
\[
\frac{1}{q}\int_{\T} \log |\det(\hat{A}_{q,E}^{\varepsilon}(p/q,\theta)-I_{2d})|\, \mathrm{d}\theta = O_d(e^{-\delta q})
\]
We now conclude Theorem~\ref{thm:non-sa-HP} from Lemmas \ref{lem:hD=L1} and \ref{lem:hD=hAq} as before.

\appendix

\numberwithin{equation}{section}

\setcounter{equation}{0}

\section{Proof of Theorem \ref{thm:ARC}}\label{app:A}

The argument closely follows~\cite{AJ2}, with a simplification from~\cite{A_ar_ac}.
Throughout this section we write $\hat{L}^j(\alpha,\hat{A}_E), \hat{L}_j(\alpha,\hat{A}_E)$ as $\hat{L}^j(\hat{A}_E), \hat{L}_j(\hat{A}_E)$ for simplicity.

\subsection{Preparation}
\begin{theorem}\cite{AJ2}*{Theorem 3.3}\label{thm:bdd_solution}
If $E\in\sigma(H_{\alpha,\theta})$, then there exist $\theta\in \R$ and a bounded solution of $\hat{H}_{\alpha,\theta}\hat{u}=E\hat{u}$ with
$\hat{u}_0=1$ and $|\hat{u}_n|\leq 1$.
\end{theorem}

Given Fourier coefficients $\hat{w}=(\hat{w}_k)_{k\in \Z}$ and an interval $I\subset \Z$, we let $w_I=\sum_{k\in I} \hat{w}_k e^{2\pi i kx}$. The length of the interval $I= [a, b]$ is $|I|= b-a$.
We will say that a trigonometrical polynomial $p:\T \to \C$ has {\em essential degree at most} $k$ if its
Fourier coefficients outside an interval $I$ of length $k$  vanish. We let $\frac{p_n}{q_n}$ be the convergents of~$\alpha$. 

\begin{theorem}\cite{AJ2}*{Theorem 6.1}\label{thm:p_LI}
Let $1\leq r\leq \lfloor q_{n+1}/q_n\rfloor$. If $p$ has essential degree at most $k= rq_n-1$ and $x_0\in \T$,
then for some absolute constant $C$
\begin{align}
\|p\|_{0}\leq Cq_{n+1}^{Cr}\sup_{0\leq j\leq k} |p(x_0+j\alpha)|.
\end{align} 
\end{theorem}

\subsection{Proof of Theorem \ref{thm:ARC}}
Fix $E\in \sigma(H_{\alpha,\theta})$, and let $\theta(E)$ be as in Theorem~\ref{thm:bdd_solution}. Let $\varepsilon_0$ be such that $0<\varepsilon_0<\frac{\hat{L}_d(\hat{A}_E)}{30\hat{L}^d(\hat{A}_E)+100C_3d}$ where $C_3$ is the absolute constant in Lemma~\ref{lem:uniform}. Finally, let $\{n_j\}$ be the $\varepsilon_0$-resonances of $\theta(E)$ and $\hat{u}$ be the almost localized solution as in Theorem~\ref{thm:almost_AL}. Hence for $3(|n_j|+1)<|n|<|n_{j+1}|/3$ and some $C_2=C_2(v,d)>0$, one has 
\begin{align}
|\hat{u}_n|\leq C_2 e^{-\frac{\hat{L}_d(\hat{A}_E)}{10} |n|}.
\end{align}
For simplicity we shall write $\theta(E)$ as~$\theta$ for the remainder of this appendix.

The cocycle $(\alpha, A_E)$ being subcritical, which precisely means that $L(\alpha, A_{E,\eta})= 0$ for all~$|\eta|\leq \hat{L}_d/(2\pi)$, implies via Lemma~\ref{lem:upper_semi_cont} the following subexponential growth \begin{align}\label{eq:A_sub_exp_upp}
\|A_{k,E}\|_{\frac{\hat{L}_d}{2\pi}}\leq C_{\delta} e^{\delta |k|},
\end{align} 
for arbitrary $\delta>0$.
Henceforth, we fix $\delta=\min(\frac{\hat{L}_d(\hat{A}_E)}{5000}, \frac{\varepsilon_0}{25})$.
Theorem~\ref{thm:Ak_poly} improves this subexponential bound to a polynomial one.

\begin{theorem}\label{thm:Ak_poly}
For some absolute constants $C_4>0$, and $C_5=C_5(\alpha,v,E,\varepsilon_0)>0$,  
\begin{align}
\|A_{k,E}\|_{\frac{\delta}{2\pi}}\leq C_5(1+|k|)^{C_4},
\end{align}
for any $k\geq 1$.
\end{theorem}
\begin{proof}
The proof proceeds by constructing an $\mathrm{SL}(2,\C)$ almost conjugacy of $A_{E}(x)$.
Thus, let $n$ be such that 
\begin{align}\label{eq:n=rqm}
\begin{cases}
9|n_j|<n<|n_{j+1}|/9, \text{ for some }j,\\ 
n=rq_m-1<q_{m+1}, \text{ for some } r,m\geq 1
\end{cases}
\end{align}
To see that such an $n$ exists, note that Lemma~\ref{eq:nj_jump} implies that
\begin{align}
\frac{|n_{j+1}|}{81|n_j|}>3,
\end{align}
Moreover, for any $r,m\geq 1$, one has 
\begin{align}
\frac{(r+1)q_m-1}{rq_m-1}\leq 2.
\end{align}
Let $I=[x_1,x_2]:=[-\lfloor n/2\rfloor+1, \lfloor n/2\rfloor]\cap\Z$, and define
\begin{align}
U^I(x)=\left(\begin{matrix} e^{2\pi i\theta} u^I(x)\\ u^I(x-\alpha)\end{matrix}\right).
\end{align}
By  direct computation one obtains 
\begin{align}\label{eq:AU=U+}
A_E(x) U^I(x)=e^{2\pi i\theta} U^I(x+\alpha)+G^I(x),
\end{align}
where 
\begin{align}\label{eq:G_upp}
\|G^I(x)\|_{\frac{\hat{L}_d(\hat{A}_E)}{40\pi}}\leq C_2 e^{-\frac{\hat{L}_d(\hat{A}_E)}{40}n}.
\end{align} 
We select $U^I(x)$ as the first column of our putative conjugacy matrix. In order to construct the second column, we first establish a lower bound on $\|U^I(x)\|$.

\begin{lemma}\label{lem:UI_lower}
For $n>n_0(\alpha,v,E,\varepsilon_0)$ satisfying \eqref{eq:n=rqm}, we have
\begin{align}
\inf_{|\mathrm{Im}(x)|\leq \frac{\hat{L}_d(\hat{A}_E)}{40\pi}} \|U^I(x)\|\geq e^{-2\delta n}.
\end{align}
\end{lemma}
\begin{proof}
If the lemma fails, then there exists $x_0$ with 
$|\mathrm{Im}(x_0)|\leq \frac{\hat{L}_d(\hat{A}_E)}{40\pi}$ so that
\begin{align}\label{assume:UIx0}
\|U^I(x_0)\|<e^{-2\delta n}.
\end{align}
Iterating \eqref{eq:AU=U+} yields 
\begin{align}
e^{2\pi i\ell \theta}U^I(x_0+\ell\alpha)=&A_{\ell,E}(x_0)U_I(x_0)\\
&-\sum_{m=1}^{\ell} e^{2\pi i (m-1)\theta} A_{\ell-m,E}(x_0+m\alpha) G^I(x_0+(m-1)\alpha).
\end{align}
Combining the above with \eqref{eq:A_sub_exp_upp}, \eqref{assume:UIx0} and \eqref{eq:G_upp}, we infer that for $0\leq \ell\leq n$,
\begin{align}
|u^I(x_0+\ell\alpha)|\leq \|U^I(x_0+\ell\alpha)\|
\leq C_{\delta} e^{\delta n} e^{-2\delta n}+\sum_{m=1}^{\ell} C_{\delta} C_2 e^{\delta(\ell-m)} e^{-\frac{\hat{L}_d(\hat{A}_E)}{40}n}
\leq C_{\delta}C_2 e^{-\delta n}.
\end{align}
By Theorem \ref{thm:p_LI}, we finally arrive at the uniform bound 
\begin{align}\label{eq:uI_uni_upper}
\|u^I(x+i\mathrm{Im}(x_0))\|_{\T}\leq Cq_{m+1}^{Cr}\sup_{0\leq \ell\leq n} |u^I(x_0+\ell\alpha)|\leq C_{\delta}C_2 e^{-\delta n/2},
\end{align}
where we used that $q_{m+1}\leq C_{\delta} e^{\delta q_m/(2C)}$, see~\eqref{eq:n_alpha_varepsilon_1}. 
Clearly, \eqref{eq:uI_uni_upper} contradicts  
$$\int_{\T} u(x+i\mathrm{Im}(x_0))\, \mathrm{d}x=\hat{u}_0=1$$ for $n$ large enough.
\end{proof}

The now invoke the following lemma to construct the second column of our conjugacy matrix.

\begin{lemma}\label{lem:2nd_col}\cite{A_ar_ac}*{Lemma 4.2}
Let $V: \T\to \C^2$ be analytic in $|\mathrm{Im}(x)|\leq \eta$. Assume 
$\delta_1<\|V(x)\|<\delta_2^{-1}$ holds on $|\mathrm{Im}(x)|<\eta$. Then there exists $M: \T\to \mathrm{SL}(2,\C)$ analytic on $|\mathrm{Im}(x)|<\eta$ with the first column equals to $V$ and $\|M\|_{\eta}\leq C\delta_1^{-2}\delta_2^{-1}(1-\log (\delta_1\delta_2))$.
\end{lemma}
Let $M^I(x)$ be the matrix obtained by means of Lemma~\ref{lem:2nd_col} whose first column equals~$U^I$. 
By Lemma~\ref{lem:UI_lower} and the fact that $|\hat{u}_k|\leq 1$,  for any $|\mathrm{Im}(x)|\leq \frac{\delta}{2\pi}$,
\begin{align}\label{eq:UI_lower_upper}
e^{-2\delta n}\leq \|U^I(x)\|\leq e^{\delta n}.
\end{align}
By Lemma \ref{lem:2nd_col}, 
\begin{align}\label{eq:M_upper}
\|M^I\|_{\frac{\delta}{2\pi}}=\|(M^I)^{-1}\|_{\frac{\delta}{2\pi}}\leq e^{6\delta n}.
\end{align}
and by \eqref{eq:AU=U+}, 
\begin{align}\label{eq:AM=M+}
A_E(x)M^I(x)=M^I(x+\alpha)
\left(\begin{matrix}
e^{2\pi i\theta} & a_{12}(x)\\
0 & e^{-2\pi i\theta}+a_{22}(x)\end{matrix}\right)
+\left(\begin{matrix} G^I(x) &0\end{matrix}\right).
\end{align}
Thus, 
\begin{align}\label{eq:MAM}
M^I(x+\alpha)^{-1} A_E(x)M^I(x)
=&
\left(\begin{matrix}
e^{2\pi i\theta} & a_{12}(x)\\
0 & e^{-2\pi i\theta}+a_{22}(x)\end{matrix}\right)
+M^I(x+\alpha)^{-1} \left(\begin{matrix} G^I(x) &0\end{matrix}\right) \notag\\
=:&
\left(\begin{matrix}
e^{2\pi i\theta} & a_{12}(x)\\
0 & e^{-2\pi i\theta}+a_{22}(x)
\end{matrix}\right)+
\left(\begin{matrix}
a_{11}(x) & 0\\
a_{21}(x) & 0
\end{matrix}\right)
\end{align}
In view of \eqref{eq:G_upp} and \eqref{eq:M_upper}, 
\begin{align}\label{eq:a11_a21}
\max(\|a_{11}\|_{\frac{\delta}{2\pi}}, \|a_{21}\|_{\frac{\delta}{2\pi}})\leq \|M^I(x+\alpha)^{-1} \left(\begin{matrix} G^I(x) &0\end{matrix}\right)\|_{\frac{\delta}{2\pi}}\leq e^{-\frac{\hat{L}_d(\hat{A}_E)}{100} n},
\end{align}
whence 
\begin{align}\label{eq:a12}
\|a_{12}\|_{\frac{\delta}{2\pi}}\leq \|(M^I)^{-1}\|_{\frac{\delta}{2\pi}}\cdot \|A_E\|_{\frac{\delta}{2\pi}}\cdot \|M^I\|_{\frac{\delta}{2\pi}}+\max(\|a_{11}\|_{\frac{\delta}{2\pi}}, \|a_{21}\|_{\frac{\delta}{2\pi}})\leq C_{\delta}e^{12\delta n}.
\end{align}
Passing to determinants in \eqref{eq:MAM}, and using \eqref{eq:a11_a21}, we conclude that 
\begin{align}
(e^{2\pi i\theta}+a_{11}(x)) (e^{-2\pi i\theta}+a_{22}(x))-a_{12}(x)a_{21}(x)=1,
\end{align}
which implies
\begin{align}
(e^{2\pi i\theta}+a_{11}(x)) a_{22}(x)=a_{12}(x)a_{21}(x)-e^{-2\pi i\theta}a_{11}(x).
\end{align}
Combining the above equation with the bounds in \eqref{eq:a11_a21} and \eqref{eq:a12} implies that 
\begin{align}\label{eq:a22}
\|a_{22}\|_{\frac{\delta}{2\pi}}\leq e^{-\frac{\hat{L}_d(\hat{A}_E)}{200}n}.
\end{align}
Let 
\begin{align}\label{def:Phi=M}
\Phi^I(x)=M^I(x)
\left(\begin{matrix} 
e^{\frac{\hat{L}_d(\hat{A}_E)}{400}n} & 0\\
0 &e^{-\frac{\hat{L}_d(\hat{A}_E)}{400}n}
\end{matrix}\right).
\end{align}
Then by \eqref{eq:MAM}, \eqref{eq:a11_a21}, \eqref{eq:a12} and \eqref{eq:a22},
\begin{align}\label{eq:Phi conj}
\Phi^I(x+\alpha)^{-1} A_E(x) \Phi^I(x)=
\left(\begin{matrix}
e^{2\pi i\theta} &0\\
0 &e^{-2\pi i\theta}
\end{matrix}\right)
+H(x),
\end{align}
where 
\begin{align}
\|H\|_{\frac{\delta}{2\pi}}\leq &\max(\|a_{12}\|_{\frac{\delta}{2\pi}}e^{-\frac{\hat{L}_d(\hat{A}_E)}{200}}, \|a_{11}\|_{\frac{\delta}{2\pi}}, \|a_{22}\|_{\frac{\delta}{2\pi}}, \|a_{21}\|_{\frac{\delta}{2\pi}}e^{\frac{\hat{L}_d(\hat{A}_E)}{200}n})\\
\leq &C_{\delta} e^{12\delta n} e^{-\frac{\hat{L}_d(\hat{A}_E)}{200}n}\leq e^{-\frac{\hat{L}_d(\hat{A}_E)}{400}n}.
\end{align}
Moreover,  by \eqref{eq:M_upper} and \eqref{def:Phi=M},
\begin{align}
\|\Phi^I\|_{\frac{\delta}{2\pi}}=\|(\Phi^I)^{-1}\|_{\frac{\delta}{2\pi}}\leq \|M^I\|_{\frac{\delta}{2\pi}} e^{\frac{\hat{L}_d(\hat{A}_E)}{400}}\leq e^{\frac{\hat{L}_d(\hat{A}_E)}{200}}. 
\end{align}
In view of~\eqref{eq:Phi conj}, this implies that for all $1\leq s\leq e^{\frac{\hat{L}_d(\hat{A}_E)}{800}n}$ one has 
\begin{align}\label{eq:AsE}
\|A_{s,E}(x)\|_{\frac{\delta}{2\pi}}\leq 2\|\Phi^I\|_{\frac{\delta}{2\pi}}\cdot \|(\Phi^I)^{-1}\|_{\frac{\delta}{2\pi}} \leq 2e^{\frac{\hat{L}_d(\hat{A}_E)}{100}n}.
\end{align}
There exists an absolute constant $c_4>0$, such that for $s>s(\alpha,v,E,\varepsilon_0)$ large enough, there exists $n$, with $9|n_j|<n<|n_{j+1}|/9$, and satisfying the requirement of Lemma~\ref{lem:UI_lower},  such that $c_4 n<\frac{800}{\hat{L}_d(\hat{A}_E)}\log s<n$.
Thus by~\eqref{eq:AsE}, we finally obtain 
\begin{align}
\|A_{s,E}\|_{\frac{\delta}{2\pi}}\leq C_5 (1+|s|)^{8/c_4}, 
\end{align}
as claimed.
\end{proof}

Next, let $J=[-\lfloor |n_{j+1}|/9\rfloor, \lfloor |n_{j+1}|/9\rfloor]\cap \Z$, and define  
\begin{align}
U^J(x)=\left(\begin{matrix}
e^{2\pi i\theta} u^J(x)\\
u^J(x-\alpha)
\end{matrix}\right).
\end{align}
By the {\em almost localization} Theorem~\ref{thm:almost_AL}, one has 
\begin{align}\label{eq:AUJ=UJ+}
A_E(x)U^J(x)=e^{2\pi i\theta} U^J(x+\alpha)+G^J(x), \text{ with } \|G^J \|_{\frac{\hat{L}_d}{32\pi}}\leq e^{-\frac{\hat{L}_d(\hat{A}_E)}{150}|n_{j+1}|},
\end{align}
and
\begin{align}\label{eq:UJ_upper}
\|U^J\|_{\frac{\delta}{2\pi}}\leq e^{3\delta |n_j|}.
\end{align}
Let $M_J(x)$ be the matrix with the first column equal to $U^J(x)$, and the second column equal to~$\overline{U^J(x)}$, respectively. 
By the preceding, 
\begin{align}\label{eq:AMJ=MJ+}
A_E(x)M_J(x)=M_J(x+\alpha)
\left(\begin{matrix}
e^{2\pi i\theta} & 0\\
0 &e^{-2\pi i\theta}
\end{matrix}\right)
+H(x),
\end{align}
where $H(x)=\left(\begin{matrix} G(x) & \overline{G(x)}\end{matrix}\right)$ satisfies $\|H\|_{\frac{\hat{L}_d}{32\pi}}\leq e^{-\frac{\hat{L}_d(\hat{A}_E)}{150}|n_{j+1}|}.$

We also define $\tilde{U}^J(x):=e^{\pi in_jx} U^J(x)$. 
When $n_j$ is even, $\tilde{U}^J$ is well-defined on $\T$, and when $n_j$ is odd, $\tilde{U}^J(x+1)=-\tilde{U}^J(x)$.
By \eqref{eq:AUJ=UJ+}, with $\tilde{\theta}:=\theta-n_j\alpha/2$, 
\begin{align}\label{eq:AtUJ=tUJ+}
A_E(x)\tilde{U}^J(x)=e^{2\pi i\tilde{\theta}} \tilde{U}^J(x+\alpha)+\tilde{G}^J(x), \text{ where } \|\tilde{G}^J\|_{\frac{\hat{L}_d(\hat{A}_E)}{32\pi}}\leq e^{-\frac{\hat{L}_d(\hat{A}_E)}{200}|n_{j+1}|}.
\end{align}
By Theorem \ref{thm:almost_AL}, we also have
\begin{align}\label{eq:tUJ_upper}
\|\tilde{U}^J\|_{\frac{\delta}{2\pi}}\leq e^{4\delta |n_j|}.
\end{align}
Let $\tilde{M}_J(x)$ be the matrix with columns $\tilde{U}^J(x)$ and $\overline{\tilde{U}^J(x)}$.
Clearly, 
\begin{align}\label{eq:dtMJ=dMJ}
|\det \tilde{M}_J(x)|=e^{-2\pi n_j \mathrm{Im}(x)} |\det M_J(x)|.
\end{align}

\begin{theorem}\label{thm:det_MJ}
Let $L=\|2\theta-n_j\alpha\|_{\T}^{-1}$. Then for $j>j_0(\alpha,v,E,\varepsilon_0)$,  
\begin{align}
|\det M_J(x)|\geq L^{-5C_4}, \text{ for } |\mathrm{Im}(x)|\leq \frac{\delta}{2\pi},
\end{align}
where $C_4$ is the constant in Theorem \ref{thm:Ak_poly}.
\end{theorem}
\begin{proof}
We first  establish the following estimates on $L$.
\begin{lemma}\label{lem:n<L<N}
For $j>j_0(\alpha,\hat{L}_d)$,
\begin{align}
e^{\varepsilon_0 |n_j|}\leq L\leq e^{\frac{\hat{L}_d(\hat{A}_E)}{4000 C_4}|n_{j+1}|}.
\end{align}
\end{lemma}
\begin{proof}
For $\varepsilon_1=\frac{\hat{L}_d(\hat{A}_E)}{4000C_4}$,
we have by \eqref{eq:n_alpha_varepsilon_1} that for $n>N(\alpha,\hat{L}_d(\hat{A}_E))$,
\begin{align}
C(\alpha,\hat{L}_d)e^{-2\varepsilon_1 |n_{j+1}|} &\leq C(\alpha,\hat{L}_d) e^{-\varepsilon_1|n_{j+1}-n_j|}\leq \|(n_{j+1}-n_j)\alpha\|_{\T}\\
& \leq 2\|2\theta-n_j\alpha\|_{\T}=2L^{-1}\leq 2e^{-\varepsilon_0 |n_j|}.
\end{align}
This implies the claimed result.
\end{proof}
We now prove Theorem~\ref{thm:det_MJ} by contradiction,  based on the next lemma. 

\begin{lemma}\label{lem:U-kappaU}
For any $\kappa\in \C$ such that $|\kappa|=1$,  
\begin{align}
\inf_{x: |\mathrm{Im}(x)|\leq \frac{\delta}{2\pi}} \|\tilde{U}^J(x)-\kappa\cdot \overline{\tilde{U}^J(x)}\|\geq L^{-4C_4}.
\end{align}
\end{lemma}
\begin{proof}
Suppose for some $x_0$, $|\mathrm{Im}(x_0)|\leq \frac{\delta}{2\pi}$, and some constant $|\kappa|=1$ we have 
\begin{align}
\|\tilde{U}^J(x_0)-\kappa\cdot  \overline{\tilde{U}^J(x_0)}\|<L^{-4C_4}.
\end{align}
By \eqref{eq:AtUJ=tUJ+}, Theorem \ref{thm:Ak_poly} and Lemma~\ref{lem:n<L<N}, for all $0\leq |\ell|\leq L^2$ one concludes that  
\begin{align}\label{eq:UJ_ell-tUJ_ell}
&\|\tilde{U}^J(x_0+\ell\alpha) e^{2\pi i\ell\tilde{\theta}}-\kappa\cdot \overline{\tilde{U}^J(x_0+\ell\alpha)} e^{-2\pi i\ell\tilde{\theta}}\| \notag\\
&\,\,\,\,\leq \left\|\sum_{m=0}^{\ell-1} A_{\ell-m,E}(x_0+m\alpha) \tilde{G}^J(x_0+m\alpha)-\kappa \sum_{m=0}^{\ell-1}A_{\ell-m,E}(x_0+m\alpha) \overline{\tilde{G}^J(x_0+m\alpha)}\right\|\\
&\,\,\,\,+\|A_{\ell,E}(x_0)\|L^{-4C_4}\notag\\
&\,\,\,\,\leq C_5 L^{2C_4+3}e^{-\frac{\hat{L}_d(\hat{A}_E)}{200}|n_{j+1}|}+C_5 L^{-2C_4}<L^{-C_4}.
\end{align}
Setting $\ell=\ell_0=\lfloor L/2\rfloor$ in \eqref{eq:UJ_ell-tUJ_ell}, we conclude that 
\begin{align}\label{eq:tUJ_ell0}
\|\tilde{U}^J(x_0+\ell_0\alpha)e^{2\pi i\ell_0\tilde{\theta}}-\kappa\cdot \overline{\tilde{U}^J(x_0+\ell_0\alpha)} e^{-2\pi i\ell_0 \tilde{\theta}}\|<L^{-C_4}.
\end{align}
Writing $\ell_0=L/2+\delta$ with $|\delta|\leq 1/2$, and $2\theta-n_j\alpha=m_0\pm \|2\theta-n_j\alpha\|_{\T}$, we infer that 
\begin{align}\label{eq:ell0_ttheta}
2\pi \ell_0\tilde{\theta}=\pi(L/2+\delta)(2\theta-n_j\alpha)=\pi\ell_0m_0\pm \frac{\pi}{2}\pm \pi \delta L^{-1}=\pi\ell_0m_0\pm \frac{\pi}{2}+O(L^{-1}).
\end{align}
Plugging \eqref{eq:ell0_ttheta} into \eqref{eq:tUJ_ell0}, and combining the resulting estimate with \eqref{eq:tUJ_upper}, yields
\begin{align}\label{eq:tUJ_ell0--}
\|\tilde{U}^J(x_0+\ell_0\alpha)+\kappa \cdot \overline{\tilde{U}^J(x_0+\ell_0\alpha)}\| &\leq  L^{-C_4}+L^{-1} \|\tilde{U}^J\|_{\frac{\delta}{2\pi}}\leq L^{-C_4}+L^{-1} e^{4\delta |n_j|}\notag\\
&<L^{-1/2}.
\end{align}
For $|\ell|\leq L^{1/2}$, we have in analogy to \eqref{eq:ell0_ttheta} that
\begin{align}\label{eq:ell_ttheta}
2\pi \ell\tilde{\theta}=\pi\ell m_0+O(L^{-1/2}).
\end{align}
Plugging \eqref{eq:ell_ttheta} into \eqref{eq:AtUJ=tUJ+} yields
\begin{align}\label{eq:tUJ_ell-tUJ_ell}
\|\tilde{U}^J(x_0+\ell \alpha)-\kappa \cdot \overline{\tilde{U}^J(x_0+\ell \alpha)}\| &\leq L^{-C_4}+L^{-1/2} \|\tilde{U}^J\|_{\frac{\delta}{2\pi}}
\leq L^{-C_4}+L^{-1/2}e^{4\delta |n_j|} \notag\\
&<L^{-1/3}.
\end{align}
Next, we let $\tilde{n}$ be such that
\begin{align}
\begin{cases}
9|n_j|<\tilde{n}<\min(27|n_j|, |n_{j+1}|/9)\\
\tilde{n}=rq_m-1<q_{m+1}, \text{ for some } r,m\geq 1
\end{cases}
\end{align}
Let $\tilde{I}=[-\lfloor \tilde{n}/2\rfloor+1, \lfloor \tilde{n}/2\rfloor]\cap \Z$ and define
\begin{align}
U^{\tilde{I}}(x)=
\left(\begin{matrix}
e^{2\pi i\theta} u^{\tilde{I}}(x)\\
u^{\tilde{I}}(x-\alpha)
\end{matrix}\right),
\end{align}
as well as  $\tilde{U}^{\tilde{I}}(x)=e^{\pi in_jx}U^{\tilde{I}}(x)$.
Then by Theorem \ref{thm:almost_AL}, we have
\begin{align}\label{eq:tUJ-tUJ_9n}
\|\tilde{U}^J-\tilde{U}^{\tilde{I}}\|_{\frac{\delta}{2\pi}}\leq e^{-\frac{\hat{L}_d(\hat{A}_E)}{2} |n_j|}
\end{align}
Combining the above with \eqref{eq:tUJ_ell-tUJ_ell} yields
\begin{align}\label{eq:tUJ_9n-tUJ_9n}
\|\tilde{U}^{\tilde{I}}(x_0+\ell \alpha)-\kappa \cdot \overline{\tilde{U}^{\tilde{I}}(x_0+\ell \alpha)}\|\leq L^{-1/3}+2e^{-\frac{\hat{L}_d(\hat{A}_E)}{2} |n_j|}.
\end{align}
Denoting
\begin{align}
g(x)=e^{-\pi n_j\mathrm{Im}(x_0)} e^{2\pi i n_j x}\sum_{|k|\leq 9|n_j|} e^{-2\pi k\mathrm{Im}(x_0)} \left(\hat{u}_k e^{2\pi i kx}-\kappa \, \overline{\hat{u}_k}e^{-2\pi i kx}\right),
\end{align}
we obtain from \eqref{eq:tUJ_9n-tUJ_9n} that 
\begin{align}
|g(\mathrm{Re}(x_0)+\ell\alpha)|
=&\left|e^{\pi i n_j(x_0+\ell\alpha)} u^{\tilde{I}}(x_0+\ell\alpha)-\kappa \;\overline{e^{\pi i n_j(x_0+\ell\alpha)} u^{\tilde{I}}(x_0+\ell\alpha)}\right|\\
\leq &\|\tilde{U}^{\tilde{I}}(x_0+\ell \alpha)-\kappa \cdot \overline{\tilde{U}^{\tilde{I}}(x_0+\ell \alpha)}\|\leq L^{-1/3}+e^{-\frac{\hat{L}_d(\hat{A}_E)}{2}|n_j|}
\end{align}
for any $|\ell|\leq L^{1/2}$.
 Lemma \ref{lem:n<L<N} implies that $L^{1/2}>e^{\varepsilon_0 |n_j|}>\tilde{n}$, whence by Theorem \ref{thm:p_LI}, 
\begin{align}
\sup_{x\in \T}|g(x)|\leq C q_{m+1}^{Cr} (L^{-1/3}+e^{-\frac{\hat{L}_d}{2}|n_j|})<L^{-1/4}+e^{-\frac{\hat{L}_d(\hat{A}_E)}{4}|n_j|}.
\end{align}
Here we used that $q_{m+1}<C_{\varepsilon_0} e^{\varepsilon_0 q_m/350 C}$, see \eqref{eq:n_alpha_varepsilon_1}, and $\tilde{n}=rq_m-1<27|n_j|$.
This implies
\begin{align}
\sup_{x\in \T} \|\tilde{U}^{\tilde{I}}(x+\mathrm{Im}(x_0))-\kappa \cdot \overline{\tilde{U}^{\tilde{I}}}(x+\mathrm{Im}(x_0))\|\leq 2\sup_{x\in \T} |g(x)|\leq 2L^{-1/4}+2e^{-\frac{\hat{L}_d(\hat{A}_E)}{4}|n_j|},
\end{align}
and in particular
\begin{align}\label{eq:tUtI_1}
\|\tilde{U}^{\tilde{I}}(x_0+\ell_0\alpha)-\kappa \cdot \overline{\tilde{U}^{\tilde{I}}(x_0+\ell_0\alpha)}\|
\leq 2L^{-1/4}+2e^{-\frac{\hat{L}_d(\hat{A}_E)}{4}|n_j|}.
\end{align}
Combining \eqref{eq:tUJ_ell0--} with \eqref{eq:tUJ-tUJ_9n}, we have
\begin{align}\label{eq:tUJ_ell0++}
\|\tilde{U}^{\tilde{I}}(x_0+\ell_0\alpha)+\kappa \cdot \overline{\tilde{U}^{\tilde{I}}(x_0+\ell_0\alpha)}\| \leq L^{-1/2}+e^{-\frac{\hat{L}_d(\hat{A}_E)}{2}n}.
\end{align}
Finally combining \eqref{eq:tUJ_ell0++} with \eqref{eq:tUtI_1}, we conclude that
\begin{align}
\|\tilde{U}^{\tilde{I}}(x_0+\ell_0\alpha)\|\leq 2L^{-1/4}+2e^{-\frac{\hat{L}_d(\hat{A}_E)}{4}|n_j|}.
\end{align}
This leads to a contradiction with Lemma~\ref{lem:UI_lower} which implies 
\begin{align}\label{eq:inf_tUtI}
\inf_{|\mathrm{Im}(x)|\leq \frac{\delta}{2\pi}}\|\tilde{U}^{\tilde{I}}(x)\|\geq e^{-2\delta \tilde{n}} e^{-\pi |n_j \mathrm{Im}(x)|}\geq e^{-55\delta |n_j|}>L^{-1/2}+e^{-\frac{\hat{L}_d(\hat{A}_E)}{2}|n_j|},
\end{align}
where we used Lemma~\ref{lem:n<L<N} and $\tilde{n}<27|n_j|$.
\end{proof}
Next, we return to the estimate of $\det M_J$.
Combining \eqref{eq:tUJ-tUJ_9n} with \eqref{eq:inf_tUtI}, we obtain
\begin{align}
\inf_{|\mathrm{Im}(x)|\leq \frac{\delta}{2\pi}}\left( \max(|e^{\pi i n_j x}u^J(x)|, |e^{\pi i n_j x}u^J(x-\alpha)|) \right) \geq \inf_{|\mathrm{Im}(x)|\leq \frac{\delta}{2\pi}} \|U^J(x)\|\geq &e^{-55\delta |n_j|}-e^{-\frac{\hat{L}_d(\hat{A}_E)}{2}|n_j|} \\ \geq &e^{-60\delta |n_j|}.
\end{align}
For an arbitrary $x$ with $|\mathrm{Im}(x)|\leq \frac{\delta}{2\pi}$, suppose we have (the other case is similar)
\begin{align}\label{eq:assume_uJx}
|e^{\pi i n_j x}u^J(x)|\geq e^{-60\delta |n_j|}.
\end{align}
Setting $\kappa:=\overline{e^{\pi i n_j x}u^J(x)}/(e^{\pi i n_j x}u^J(x))$, which satisfies $|\kappa|=1$, one has
\begin{align}
|\det \tilde{M}_J(x)|
=&|\det (\begin{matrix}\tilde{U}^J(x),\, \, \, \overline{\tilde{U}^J(x)}-\kappa\cdot \tilde{U}^J(x)\end{matrix})|
=|e^{\pi i n_j x}u^J(x)|\cdot \|\tilde{U}^J(x)-\overline{\kappa}\cdot \overline{\tilde{U}^J(x)}\|\\
\geq &L^{-4C_4} e^{-60\delta |n_j|},
\end{align}
where we used \eqref{eq:assume_uJx} and Lemma~\ref{lem:U-kappaU} in the last line.
The claimed result on $|\det M_J(x)|$ follows by combining \eqref{eq:dtMJ=dMJ} with Lemma~\ref{lem:n<L<N}.
\end{proof}

Next, we construct an $\mathrm{SL}(2,\R)$ conjugacy. 
Let $S(x)=\mathrm{Re}\, U^J(x)$ and $T(x)=\mathrm{Im}\, U^J(x)$, and define 
 $W_1(x)$ to be the matrix with columns $S(x)$ and $T(x)$.
Thus
\begin{align}
M_J(x)=W_1(x)\left(\begin{matrix} 1\ & 1\\ i &-i\end{matrix}\right).
\end{align}
Clearly \eqref{eq:UJ_upper} implies 
\begin{align}\label{eq:W1_upper}
\|W_1\|_{\frac{\delta}{2\pi}}\leq e^{3\delta |n_j|}.
\end{align}
By Theorem \ref{thm:det_MJ},  $\det W_1(x)\neq 0$ for $|\mathrm{Im}(x)|\leq \frac{\delta}{2\pi}$. Without loss of generality we assume $\det W_1(x)>0$ on $\{x: |\mathrm{Im}(x)|\leq \frac{\delta}{2\pi}\}$ (otherwise we take $W_1(x)$ be the matrix with columns $S(x)$ and $-T(x)$). By \eqref{eq:AMJ=MJ+}, we conclude that 
\begin{align}\label{eq:AW1=W1+}
\|A_E(x) W_1(x)-W_1(x+\alpha) R_{-\theta}\|_{\frac{\hat{L}_d(\hat{A}_E)}{32\pi}}\leq e^{-\frac{\hat{L}_d(\hat{A}_E)}{150}|n_{j+1}|}.
\end{align}
Combining this estimate with \eqref{eq:W1_upper}, we infer that for $|\mathrm{Im}(x)|\leq \frac{\delta}{2\pi}$,
\begin{align}\label{eq:detW1_invariance}
\det W_1(x)=\det W_1(x+\alpha)+O(e^{-\frac{\hat{L}_d(\hat{A}_E)}{200}|n_{j+1}|}).
\end{align}
By Theorem \ref{thm:det_MJ}, Lemma \ref{lem:n<L<N} and \eqref{eq:detW1_invariance},   for all $|\mathrm{Im}(x)|\leq \frac{\delta}{2\pi}$,
\begin{align}\label{eq:detW1_constant}
\det W_1(x)=w_0+O(e^{-\frac{\hat{L}_d(\hat{A}_E)}{300}|n_{j+1}|}),
\end{align}
with $w_0\geq L^{-5C_4}$. 
Combining \eqref{eq:AW1=W1+} with \eqref{eq:nj_jump}, \eqref{eq:W1_upper}, \eqref{eq:detW1_constant} and Lemma \ref{lem:n<L<N}, we have
\begin{align}\label{eq:W1AW1}
\|W_1(x+\alpha)^{-1} A_E(x)W_1(x)-R_{-\theta}\|_{\frac{\delta}{2\pi}}
\leq &\frac{1}{\inf_{|\mathrm{Im}(x)|\leq \frac{\delta}{2\pi}} \det W_1(x)} \|W_1\|_{\frac{\delta}{2\pi}} e^{-\frac{\hat{L}_d(\hat{A}_E)}{150}|n_{j+1}|}\notag\\
= &O(L^{6C_4} e^{3\delta |n_j|} e^{-\frac{\hat{L}_d(\hat{A}_E)}{150}|n_{j+1}|})=O(e^{-\frac{\hat{L}_d(\hat{A}_E)}{200}|n_{j+1}|}).
\end{align}
Define $W_2(x)=(\det W_1(x))^{-1/2} W_1(x) \in C^{\omega}_{\frac{\delta}{2\pi}}(\T, \mathrm{SL}(2,\R))$.
By \eqref{eq:W1AW1} and \eqref{eq:detW1_constant}, we have that for $|\mathrm{Im}(x)|\leq \frac{\delta}{2\pi}$,
\begin{align}
W_2(x+\alpha)^{-1} A_E(x) W_2(x)
=&\frac{(\det W_1(x+\alpha))^{1/2}}{(\det W_1(x))^{1/2}} 
R_{-\theta}+O(e^{-\frac{\hat{L}_d(\hat{A}_E)}{200}|n_{j+1}|})\\
=&R_{-\theta}+O(e^{-\frac{\hat{L}_d(\hat{A}_E)}{200}|n_{j+1}|}).
\end{align}
Taking \eqref{eq:W1_upper}, \eqref{eq:detW1_constant} and Lemma \ref{lem:n<L<N} into account, one has
\begin{align}
\|W_2\|_{\frac{\delta}{2\pi}}\leq L^{6C_4} e^{3\delta |n_j|}\leq L^{7C_4}.
\end{align}
The estimate of $\mathrm{deg}(W_2)$, see Definition \ref{def:Degree}, follows the same  Rouch\'e argument as the proof of Theorem~3.5 in \cite{AJ2}. In fact, in the subcritical case, one should be able to obtain an improved bound $|\mathrm{deg}(W_2)|\lesssim \log |n_j|$. We leave those details to interested readers.

\section{Proof of Theorem \ref{thm:dryTen}}
\subsection{Rotation number}
\begin{definition}\label{def:Degree}
Any $A:\T\to \mathrm{SL}(2,\R)$ is homotopic to $x\to R_{nx}$ for some integer $n$, called the degree of $A$ and denoted $\mathrm{deg}(A)=n$.
\end{definition}

Assume that $A:\T\to \mathrm{SL}(2,\R)$ is homotopic to the identity. There exist $\psi: \T^2\to \R$ (lift of $A$) and $v: \T^2\to \R^+$ such that
\begin{align}
A(x) \left(\begin{matrix}
    \cos(2\pi y)\\
    \sin(2\pi y)
\end{matrix}\right)
=
v(x,y)
\left(\begin{matrix}
    \cos(2\pi(y+\psi(x,y)))\\
    \sin(2\pi(y+\psi(x,y)))
\end{matrix}\right).
\end{align}
Let $\mu$ be any probability measure on $\T^2$ that is invariant under the map $$T: (x,y)\to (x+\alpha, y+\psi(x,y)),$$ projecting over the Lebesgue measure on the first coordinate. 
Then the number
\begin{align}
\rho(\alpha,A)=\int \psi\, \mathrm{d}\mu \qquad (\mathrm{mod} \Z),
\end{align}
is independent of $\psi$ and $\mu$, and is called the fibered rotation number of $(\alpha, A)$.
It is clear from the definition that
\begin{align}\label{eq:rho1}
|\rho(\alpha, A)-\theta|\leq C \|A-R_{\theta}\|_0.
\end{align}

If the cocycles $(\alpha, A^{(1)})$ and $(\alpha, A^{(2)})$ are conjugate to each, i.e.,  $B(x+\alpha)^{-1}A^{(1)}(x)B(x)=A^{(2)}(x)$ for some  continuous $B:\T\to \mathrm{PSL}(2,\R)$ with $\mathrm{deg}(B)=k$, then
\begin{align}\label{eq:rho2}
\rho(\alpha, A^{(1)})=\rho(\alpha, A^{(2)})+\frac{k}{2}\alpha.
\end{align}
The following relation between the rotation number on the one hand,  and integrated density of states (see \eqref{def:IDS}) on the other hand, is well-known:
\begin{align}\label{eq:N=2rho}
\mathcal{N}(E)=1-2\rho(\alpha, A_E).
\end{align}

\subsection{Almost localization and almost reducibility}
The following connection between localization and reduciblity was first explored in \cite{Puig}, leading to the proof the Ten Martini problem for the almost Mathieu operator for all Diophantine $\alpha$'s.
The  exact form cited below appeared in \cite{AJ2}.
\begin{theorem}\cite{AJ2}*{Theorem 2.5}\label{thm:L_R}
For $\alpha\in \R\setminus \Q$, let $H_{\alpha,\theta,v}$ be the Schr\"odinger operator with a real-valued analytic potential $v$. Let $\theta,E\in \R$ be such that there exists a non-trivial exponentially decaying solution of the dual operator $\hat{H}\hat{u}=E\hat{u}$. Then
\begin{itemize}
    \item If $2\theta\notin \Z\alpha+\Z$, then there is a $B: \T\to \mathrm{SL}(2,\R)$ analytic such that $B(x+\alpha)^{-1}A_E(x)B(x)=R_{\pm \theta}$. 
    \item If $2\theta-n\alpha\in \Z$ for some $n\in \Z$, then there exists $B: \T\to \mathrm{PSL}(2,\R)$ and $w: \T\to \R$ analytic such that
    \begin{align}
    B(x+\alpha)^{-1}A_E(x)B(x)=
    \left(
    \begin{matrix}
        \pm 1\, &w(x)\\
        0 &\pm 1
    \end{matrix}
    \right).
    \end{align}
    If moreover $\beta(\alpha)=0$, then $w$ can be chosen to be a constant.
\end{itemize}
In either case, $\rho(\alpha,A_E)=\pm \theta+\frac{m}{2}\alpha$ for some $m\in \Z$.
\end{theorem}

\subsection{Proof of Theorem \ref{thm:dryTen}}
By Theorem \ref{thm:ARC}, we have the following spectral consequence, which is the analogue of Theorem~4.1 in~\cite{AJ2}.
\begin{theorem}\label{thm:ARC_consequence}
    For $\beta(\alpha)=0$ and $E\in \Sigma_{\alpha,0}$, there exists a constant $c>0$ (depending on $\hat{L}_d(\alpha, \hat{A}_E)$) with the following properties:
    \begin{enumerate}
        \item If $\rho(\alpha, A_E)$ is $c-$resonant, then there exists a sequence $B^{(n)}: \T\to \mathrm{SL}(2,\R)$ such that $B^{(n)}(x+\alpha)^{-1}A_E(x)B^{(n)}(x)$ converges to a constant rotation uniformly in $\{x: |\mathrm{Im}(x)|<c\}$.
        \item If $\rho(\alpha, A_E)$ is not $c-$resonant and $2\rho(\alpha, A_E)\notin \Z\alpha+\Z$, then there exists $B\in C^{\omega}_c(\T,\mathrm{SL}(2,\R))$ such that $B(x+\alpha)^{-1}A_E(x)B(x)$ is a constant rotation.
        \item If $2\rho(\alpha,A_E)\in \Z\alpha+\Z$, then there exists $B: \T\to \mathrm{PSL}(2,\R)$ analytic such that \linebreak $B(x+\alpha)^{-1}A_E(x)B(x)$ is a constant.
    \end{enumerate}
\end{theorem}
\begin{proof}
Let $\varepsilon_0:=\frac{\hat{L}_d(\alpha,\hat{A}_E)}{50C_3}$ and $\theta=\theta(E), \{n_j\}, \delta, L_j, W$ be defined as in Theorem \ref{thm:ARC}.
First we show that if $\theta$ is resonant, then $\rho(\alpha,A_E)$ must be resonant.
Thus suppose that $\theta$ is $\varepsilon_0$-resonant. By Theorem \ref{thm:ARC}, for $j>j_0(\alpha,v,E)$ there exists $W_j\in C^{\omega}_{\frac{\delta}{2\pi}}(\T, \mathrm{SL}(2,\R))$ such that
\begin{align}\label{eq:rho3}
\|W_j(x+\alpha)^{-1}A_E(x)W(x)-R_{-\theta}\|_{\frac{\delta}{2\pi}}\leq e^{-\frac{\hat{L}_d(\alpha,\hat{A}_E)}{300}|n_{j+1}|},
\end{align}
and $|\mathrm{deg}(W_j)|\leq C_0|n_j|$ with absolute constant $C_0>0$.
By \eqref{eq:rho1}, \eqref{eq:rho2} and \eqref{eq:rho3},
\begin{align}
|\rho(\alpha,A_E)-\theta+\mathrm{deg}(W_j)\, \alpha |\leq Ce^{-\frac{\hat{L}_d(\alpha,\hat{A}_E)}{300}|n_{j+1}|}.
\end{align}
This implies, after combining the preceding with Lemma~\ref{lem:n<L<N}, that 
\begin{align}
\|2\rho(\alpha,A_E)-(n_j-2\mathrm{deg}(W_j))\alpha\|_{\T}
\geq &L_j^{-1}-Ce^{-\frac{\hat{L}_d(\alpha,\hat{A}_E)}{300}|n_{j+1}|}\\
\geq &e^{-\frac{\hat{L}_d(\alpha, \hat{A}_E)}{4000C_4}|n_{j+1}|}-Ce^{-\frac{\hat{L}_d(\alpha,\hat{A}_E)}{300}|n_{j+1}|}>0,
\end{align}
and
\begin{align}
\|2\rho(\alpha,A_E)-(n_j-2\mathrm{deg}(W_j))\alpha\|_{\T}
\leq &L_j^{-1}+Ce^{-\frac{\hat{L}_d(\alpha,\hat{A}_E)}{300}|n_{j+1}|}\\
\leq &e^{-\varepsilon_0|n_j|}+Ce^{-\frac{\hat{L}_d(\alpha,\hat{A}_E)}{300}|n_{j+1}|}\\
\leq &e^{-(2C_0+2)^{-1} \varepsilon_0 |2\mathrm{deg}(W_j)-n_j|}.
\end{align}
Together with \eqref{eq:res_suf} this implies that 
$\rho(\alpha,A_E)$ is $(2C_0+2)^{-1}\varepsilon_0$-resonant at $n_j-2\mathrm{deg}(W_j)$.
Let $c=(2C_0+2)^{-1}\varepsilon_0, \frac{\delta}{2\pi}$. 
We divide the analysis into three cases.

Case 1. If $\rho(\alpha, A_E)$ is $c$-resonant and $\theta_0$ is non-$\varepsilon_0$-resonant, then  Theorem \ref{thm:AL_cos} implies that there is an exponentially decaying eigenfunction. Theorem \ref{thm:L_R} applies and yields the claimed result.

Case 2. If $\rho(\alpha, A_E)$ is $c$-resonant and $\theta_0$ is $\varepsilon_0$-resonant, then \eqref{eq:rho3} is the claimed result. 

Case 3. If $\rho(\alpha, A_E)$ is non-$c$-resonant, then by the arguments above $\theta$ is non-$\varepsilon_0$ resonant. Theorem~\ref{thm:L_R} implies the claimed result.
\end{proof}

Finally, Theorem \ref{thm:dryTen} is the consequence of the following statement.

\begin{theorem}
For $\beta(\alpha)=0$, $|\lambda|<1$ and $|\nu|<\nu_0(\alpha,\lambda,w)$ sufficiently small. Suppose $E\in \sigma(H^{\mathrm{PAMO}}_{\alpha,\theta,\nu,w})$ is such that $\mathcal{N}(E)\in \Z\alpha+\Z$, then $E$ belongs to the boundary of a component of $\R\setminus \sigma(H^{\mathrm{PAMO}}_{\alpha,\theta,\nu,w})$.
\end{theorem}
\begin{proof}
Due to the arguments in \cite{Puig}, it is enough to show that $(\alpha, A^{\mathrm{PAMO}}_E)$ is reducible. Suppose $(\alpha, A^{\mathrm{PAMO}}_E)$ is not reducible, then Theorem~\ref{thm:ARC_consequence} implies that $\rho(\alpha, A_E)$ is $c$-resonant for some $c>0$. But this contradicts the property $2\rho(\alpha, A_E)\in \Z\alpha+\Z$, which follows from \eqref{eq:N=2rho}.
\end{proof}

\end{document}